\documentclass[leqno,english, 11pt]{article}
\usepackage[T1]{fontenc}
\usepackage[latin9]{inputenc}
\usepackage{stackrel,amsmath,amsthm,cleveref}
\usepackage[cmyk]{xcolor}
\usepackage{amsmath}
\usepackage{esint}
\usepackage{babel}
\usepackage{ifsym}
\usepackage{amstext}
\usepackage{bbm}
\usepackage{geometry}
\usepackage{color}
\usepackage{textcomp}
\usepackage{amssymb}
\usepackage{stackrel}
\usepackage{setspace}
\usepackage{graphicx}
\usepackage[round, sort&compress, authoryear]{natbib}
\usepackage{caption}
\usepackage{subcaption}
\usepackage{booktabs}
\usepackage[subtle]{savetrees}
\usepackage{xurl}
\usepackage{float}
\usepackage{placeins}
\usepackage{amsthm}
\usepackage{pgfplots}
\usepgfplotslibrary{fillbetween}
\usetikzlibrary{patterns}
\usepackage{bm}
\usetikzlibrary{shapes.geometric}
\pgfplotsset{width=13cm,compat=1.10}

\pgfmathdeclarefunction{gauss}{2}{%
  \pgfmathparse{1/(#2*sqrt(2*pi))*exp(-((x-#1)^2)/(2*#2^2))}%
}

\newtheorem{lemma}{{\bf \sc Lemma}}

\newtheorem{proposition}{{\bf \sc Proposition}}
\newtheorem{remark}{{\bf \sc Remark}}

\newtheorem{observation}{{\bf \sc Observation}}

\crefname{claim}{claim}{claims}

\title{\vspace{-2cm} Equal Pay for \emph{Similar} Work\thanks{We thank the editor and four anonymous referees We are especially grateful to Koji Yokote for numerous helpful comments and for directing us to mathematical results used for our analysis. We also thank Bo Cowgill, Zoe Cullen, Mayara Felix, Takuma Habu, Michihiro Kandori, Erik Madsen, Akihiko Matsui, Teddy Mekonnen, Ellen Muir, Yusuke Narita, Chris Neilson, Yoko Okuyama, Debraj Ray,  Jesse Shapiro, Neil Thakral, Rakesh Vohra, and Seth Zimmerman for their helpful comments and conversations, as well as to seminar audiences at UPenn, NYU, Yale, Brown, USC, UCSD, Princeton, Wash U, Northwestern, University of Queensland, University of Melbourne and conference audiences at NBER Market Design, SITE (Market Design, Gender), EC, SOLE, EAAMO, Virtual Market Design Seminar, MAD, Tokyo Gender Economics Conference, Organization@Cornell24, and CETC.   We acknowledge research assistance from Atsumu Akazawa, Nanami Aoi, Kento Hashimoto, Asuka Hirano, Nadeen Kablawi, Masanori Kobayashi, Shinji Koiso, Kevin Li, Anya Marchenko, Sota Minowa, Daiji Nagara, Leo Nonaka, Ryosuke Sato, and Kenji Utagawa. Fuhito Kojima is supported by the JSPS KAKENHI Grant-In-Aid 21H04979 and JST ERATO Grant Number JPMJER2301. Bobak Pakzad-Hurson acknowledges support from the James M. and Cathleen D. Stone  Inequality Initiative, and the Orlando Bravo Center for Economic Research.}}
\author{Diego Gentile Passaro
\hspace{0.5cm}Fuhito Kojima
\hspace{0.5cm} Bobak Pakzad-Hurson
\thanks{Gentile Passaro: Amazon Science, diegogentilepassaro@gmail.com; this research does not relate to Gentile Passaro's position at Amazon Science. Kojima: Department of Economics, University of Tokyo, fuhitokojima1979@gmail.com. Pakzad-Hurson: Department of Economics, Brown University, bph@brown.edu.}}

\begin{document}

\maketitle

\begin{abstract}

\noindent Equal pay laws increasingly  require that workers with different group identities doing ``similar'' work are paid equal wages within firm. We study such  ``equal pay for similar work'' (EPSW) policies theoretically and test our models' predictions empirically using evidence from a 2009 gender-based Chilean EPSW. Under EPSW, firms segregate their workforce by gender. When there are more men than women in a labor market, EPSW increases the gender wage gap.
\end{abstract}

\section*{Introduction}

\begin{flushright}

\begin{quote}
``No employee with status within one or more protected
class or classes shall be paid a wage at a rate less than the rate at
which an employee without status within the same protected class or
classes in the same establishment is paid for... \textbf{similar work} [emphasis added]''\\
\end{quote}
-New York Labor Code, Section 194
\end{flushright}

Firms have some degree of wage-setting power in many labor markets \citep[see, e.g.][]{Manning2005,Card2022}. Because of this, they may pay workers different relative salaries in ways that are repugnant to society at large. In particular, wage gaps between groups of workers, often men and women, are frequent rallying points for governmental action. A popular form of legislation seeks to prohibit firms from paying disparate wages to different workers, guided by the principle of ``equal pay for equal work'' (EPEW). In the United States, 49 states had EPEW laws in effect in 2015, requiring each firm to pay equal wages to all of its workers doing equal work.

However, EPEW may be difficult to enforce; ``equal pay'' is straightforward to define, but it is likely that no two workers do exactly ``equal work'' within a firm. Firms can avoid the intent of these laws by pointing out differences between workers or making other maneuvers such as job title proliferation  to marginally heterogenize their workforce \citep{BB1986,goldin}.\footnote{For example, a manufacturer told the \emph{Washington Post} (1964)\nocite{washingtonpost} that his firm would ``downgrade some job classifications for women and reassign higher-level,
higher-paying duties to men'' in response to EPEW. We are grateful to Martha Bailey for bringing this article to our attention.
} To combat this enforceability issue, many EPEW laws have been updated to include a measure of coarseness--they require a firm to set ``equal pay for similar work'' (EPSW).\footnote{Some jurisdictions refer to such policies as ``Equal Pay for Comparable Work'' or ``Equal Pay for Substantially Similar Work.'' Given the similarity in both the stated goals and remaining wording of these policies to EPSW, we view them all synonymously.} California, among the first states in the US that moved from EPEW to EPSW, described a ``key difference'' of EPSW as ``mak[ing] it more difficult for employers to justify unequal pay.''\footnote{See \url{https://www.dir.ca.gov/dlse/california_equal_pay_act.htm}.} The judicial system has taken a similar view of the ``coarseness'' of EPSW; the highest court in the US to have adjudicated an alleged violation of EPSW was the 9$^{th}$ Circuit Court of Appeals, whose majority opinion in \emph{Freyd v. University of Oregon} states ```[similar] work'
is a more inclusive standard than equal work,'' and interprets ``similar'' work in the  university faculty setting to be defined by the department-by-rank cell (e.g. all full professors in the economics department of a particular university are ``similar'' to one another). As of January 2023, more of the US workforce is under the jurisdiction of a state EPSW law than a state EPEW law.\footnote{The percent of the US workforce that is under the jurisdiction of EPSW and EPEW, respectively, are 45.9 and 45.6. These figures are calculated from 1) finding all states covered by each of these policies (see \url{https://www.dol.gov/agencies/wb/equal-pay-protections}), and 2) the share of the US workforce employed in each state (see \url{https://www.statista.com/statistics/223669/state-employment-in-the-us/)}.} The equal pay provision in EPSW is frequently group based in that it binds only across groups of workers, prohibiting, for example, that a man is paid more than a ``similar'' woman (and vice versa).

Despite the rapid growth of EPSW, little is known about its effects on labor market outcomes. Since EPSW is more constraining on firms than EPEW, EPSW may lead to a larger direct effect on wages. But how will  firms adjust their employment policies to adapt in equilibrium? How will potential employment changes affect the goal of ensuring fairer pay? 

We theoretically and empirically study the labor market effects of EPSW. Our findings suggest that the equilibrium effects of EPSW overwhelm the direct effects, leading to increased occupational segregation and a shift in the wage gap in favor of the majority group of workers in a labor market.
Therefore, these policies may counterintuitively exacerbate the problem they were intended to solve: for example, as in our empirical setting, if male labor force participation is higher than female labor force participation and men receive higher average wages than women, then implementing a gender-based EPSW will \emph{increase} the gender wage gap, exacerbating the problem it was designed to solve.

We develop a theoretical framework to elucidate key economic forces at play. Homogeneous firms with a constant-returns-to-scale production function compete for a continuum of heterogeneous workers. Each worker belongs to one of two groups, $A$ or $B$  (e.g. men or women), and perform ``similar'' work in the eyes of the law. Therefore, our theoretical predictions should be viewed as applying ``within job'' in a particular labor market, and should not be used to predict differential effects on, for example, custodians and lawyers working within the same firm who do not perform ``similar'' work.

An important point in our analysis is identifying firms' willingness to pay for each worker with and without EPSW. We wish to include three features on willingness to pay in our model. First, recalling that EPSW is a policy undertaken to reduce the impact of discrimination in the labor market, we wish to include in our model the possibility that firms are biased against one group of workers.\footnote{For example, taste-based bias by (managers of) firms can be incorporated by reducing the willingness to pay for workers of a particular group by a fixed amount $b>0$. Moreover, even though our model assumes complete information, our model can easily be extended to include uncertainty over a worker's output at the time of hiring, allowing us to nest statistical discrimination into our framework.} Second, and as discussed above, the intentional choice of EPSW to require equal pay among workers doing unequal work leads us to include in our model workers with heterogeneous output. Third, heterogeneities across groups may cause the distribution of willingness to pay to differ for workers of different groups. We satisfy these goals by endowing each worker with a \emph{(net of discrimination) productivity}---which we henceforth abbreviate to \emph{productivity}---drawn from a distribution that potentially differs by group identity. We highlight that we make no restrictions on the relative shape of the distributions for $A$ and $B$ group workers, i.e. do not assume that the $B$ group is more or less ``productive'' than the $A$ group.  Firms' willingness to pay for a worker will equal that worker's productivity, but EPSW may constrain the wages a firm can pay.

Our model analysis reveals important effects of EPSW on worker sorting across firms. We begin by describing the results of a simple, static model containing two firms and no frictions for intuition. Without EPSW, each worker can be hired by either firm in equilibrium, regardless of group identity or productivity. Similarly to the classic Bertrand model, firms compete fiercely for each worker, and as a result, the average gap in pay across groups $A$ and $B$ is equal to the difference in average productivity between the groups. Thus, any  discriminatory factors affecting firms'  willingness to pay are exactly reflected in the wage gap. 

With EPSW, we show that firms segregate the workforce, with one firm hiring all $A-$group workers and the other hiring all $B-$group workers. To understand why this is the case,  note that because each firm hires from only one group of workers, no firm is exposed to the constraint of equal pay in equilibrium. By contrast, EPSW makes poaching workers from its competitor costly: EPSW requires equal pay to any two workers from different groups and, by transitivity, this implies that  equal wages must be paid to \emph{all} workers it hires. Thus, EPSW serves as the enforcement mechanism for segregation, similarly to location choices in Hotelling's competition model.

We then show that EPSW moves the wage gap in favor of the majority group of workers. This result follows from an  equal profit condition between firms that must be satisfied in equilibrium. More specifically, if there are more $A-$group workers than $B-$group workers, the firm that hires these workers must receive smaller average profit from each worker than the other firm receives from the average $B-$group worker, so $A-$group workers' average wage is relatively higher. Notably, the directional effect of EPSW on the wage gap follows simply because the majority group has a larger population and, in particular, this conclusion holds regardless of the distributions of productivities of the two groups. Moreover, we also show that firm profit and the magnitude of increase in the wage gap co-move, implying that firms would benefit from selecting equilibria with larger wage gaps.

One might claim that our model stacks the deck against EPSW because the outcome without EPSW is already ``fair'' in the sense that workers from both groups are paid their productivity. We disagree.
Recall that our model does not make strong assumptions on  discriminatory factors, so ``productivity'' could incorporate firms' bias. Given this observation, our model does \emph{not} take a stance on whether or not the outcome without EPSW is fair. What we do show, by contrast, is that EPSW is relatively more advantageous to the majority group.  And this  conclusion is robust in that it does not depend on whether the outcome without EPSW is fair or not. Of course, if the labor market is not fair in that the majority group is favored without EPSW for discriminatory reasons, then our result implies that EPSW makes the labor market even less fair. 

Our basic model is stylized in several respects. To study whether our main predictions hold under more realistic assumptions, we also study a dynamic model of search and bargaining. The model features an arbitrary number of firms, search frictions, and reallocation frictions. Our search model produces  qualitatively similar results to those in the static model. Specifically, firms  segregate eventually, if not necessarily immediately, after the introduction of EPSW. A worker's bargaining power is endogenously determined by the number of firms willing to hire her, therefore, the wage gap moves in favor of the group for which more firms segregate.  Under an additional regularity condition we show that  more firms segregate for the majority group, and thus the wage gap moves in favor of the majority group.  These results are dynamic analogues of our main theoretical results in the static model. Moreover, we also find that firms with small reallocation frictions in the sense of having a nearly segregated pool of existing employees segregate early, while other firms may opt to stay desegregated for longer.

Our theoretical framework yields clear empirical tests. We predict that EPSW leads firms to \emph{fully} segregate, and not merely that it increases the share of the majority group of workers at each firm. Moreover, we predict EPSW leads to an increase in the wage gap toward majority-group workers in a particular labor market, therefore, EPSW's directional impact on wage inequality is heterogeneous depending on the labor market. These specific, heterogeneous predictions serve to differentiate our framework from others. For instance, consider an alternative hypothesis that EPSW reduces demand for labor from the historically-discriminated against group $g$. We would expect a decrease  in the share of firms' workforce from group $g$, but not necessarily all-or-nothing segregation. Moreover, we would expect a relative decrease in the wages of $g-$group workers, regardless of whether $g-$group workers are in the majority of a particular labor market.

We test our theoretical predictions by studying the effect of the enactment of EPSW in Chile in 2009. This EPSW was the first equal pay law in Chile, and it constrained how a firm could pay its workers across gender: no firm is permitted to pay a woman less than it pays a man for similar work, and vice versa. The law subjects firms in violation to substantial fines, and through a public-records request, we find direct evidence of policy enforcement. Importantly, EPSW binds only for firms above a particular size threshold. This allows for a straightforward event-study (difference-in-differences) design to estimate the causal effects of EPSW, wherein we compare differences in outcomes of firms  above (``treated'') and below (``control'') the threshold following EPSW. Following \cite{BSTW2018,gulyas,duchini,boheim}, we restrict our sample to firms just above and below this size threshold to limit size-based wage dynamics. That is, our identifying assumption is that parallel trends hold for similarly-sized firms. 

Using matched employer-employee data from 2005--2013 we identify the following effects of EPSW consistent with our model predictions. First, EPSW increases gender segregation across firms. The share of firms with workers of only one gender increases by 4.4 percentage points, off a baseline of 34.3\% of gender-segregated firms prior to EPSW enactment.  We show that EPSW leads to a ``missing mass'' of firms that are nearly-but-not-fully gender segregated, and moreover, that the share of ``missing,'' nearly segregated firms is of a similar magnitude as the increased share of fully segregated firms. Moreover, firms that are nearly-but-not-fully segregated pre policy are the most likely to fully segregate post policy. These patterns are consistent with the predictions from our search model, and suggest that firms on the margin of full segregation can more easily transition to full segregation (e.g.  by firing the only worker of the ``wrong'' gender).

Second, we show that EPSW moves the gender wage gap in favor of the majority group of workers in a local labor market. 
In local labor markets---defined by firm industry and county---where a majority of workers are men, we find that EPSW \emph{increases} (in favor of men) the gender wage gap by 4.3 percentage points, while in local labor markets where a majority of workers are women, we find that EPSW \emph{decreases} (in favor of women) the gender wage gap by 6.2 percentage points. These findings exactly match our prediction that EPSW benefits the majority group of workers in a labor market. Because men in Chile dominate the overall labor market (5/6 of all workers are employed in majority male local labor markets), the overall effect of EPSW is to increase the gender wage gap (in favor of men) by 2.7 percentage points.

\subsection*{Related Literature}
While we are the first, we believe, to analyze the novel equilibrium effects of EPSW, there are rich theoretical and empirical literatures related to EPEW.

Theoretical studies of EPEW have typically focused on its unintended effects. This focus can be traced back to Milton Friedman who once famously said, ``What you are doing, not intentionally, but by misunderstanding, when you try to get equal pay for equal work law... is reducing to zero the cost imposed on people who are discriminating for irrelevant reasons.''\footnote{See \url{https://www.aei.org/carpe-diem/milton-friedman-makes-the-case-against-equal-pay-for-equal-work-laws/}.} More recent work studies EPEW in Salop's classic location model; the first such paper is \cite{Bhaskar2002} and is succeeded by \citet{Kaas2009,lagerlof2020,Lanning,Berson2016}. These papers must contend with the very motivation that led to EPSW: what is ``equal work''? Doing so results in at least two difficulties. 
First, the authors interpret ``equal work'' literally and assume workers are equally productive, while in reality there may be very few workers whose productivities are exactly equal.
Second, their analyses predict that EPEW can either increase or decrease differences in outcomes across groups of workers, often within the same paper. The lack of clear policy-relevant predictions is reflected in the empirical literature on EPEW, which we discuss shortly. 
By contrast, we find that EPSW has clear, if unintended, effects: our theoretical analysis unambiguously predicts both job segregation and a change in the wage gap that favors the majority group of workers in the labor market, and our empirical analysis of Chilean data confirms both predictions.

The empirical literature investigating equal-labor-rights legislation primarily considers US policies in the 1960s and 1970s. As with the theoretical literature we detail above, this empirical literature draws mixed conclusions about whether  such legislation improves the employment rate or wages of protected classes of workers (see \cite{bailey,AB1999,BK1992,DH1991,NeumarkStock2006,HDG2020} for detailed discussions).\footnote{The findings of \cite{bailey} suggest that the conclusions from this literature may be sensitive to the econometric methods used.}  One difficulty in much of the literature is assessing the impacts of individual policies, as many related labor policies were enacted in quick succession.\footnote{The Equal Pay Act of 1963 requires equal pay for men and women for equal work, while Title VII of the Civil Rights Act of 1964 prohibits discrimination in hiring, layoffs, and promotions. There were also other federal equal pay policies---Executive order 11246 in 1965 banning discrimination in hiring by federal contractors against minority candidates, and an extension to include women in 1967; the Equal Employment Opportunity Act in 1972 to increase enforcement; and many individual state policies.} \cite{DH1991} argue that it is difficult to attribute observed effects to any one of the contemporaneous policies, as there may be complementarities between them. Our empirical setting of Chile is notable as no existing equal pay laws were on the books at the time EPSW was enacted in 2009, and no significant related policies were enacted in quick succession.

Methodologically, our paper is more related to the literature on ``best-price'' guarantees, which commit firms to rebating past consumers if prices fall in the future. These policies have the direct effect of equalizing payments across heterogeneous buyers, but have the unintended equilibrium effect of raising firm market power \citep{Butz1990,CF1991,SM1997,SM1997b}. In our paper, EPSW plays the role of a best-wage guarantee, but due to competition between firms, it importantly binds off the equilibrium path; firms segregate in equilibrium because a deviating firm faces a competitive disadvantage in hiring. As a result, while firms in our model have an ex-ante identical willingness to pay for each particular worker, the costs of hiring workers from the ``wrong'' group are differentiated in equilibrium, driving the unintended wage effects of EPSW. 
 This force is similar to ``artificial'' costs that heterogenize ex-ante identical products in consumer markets, which can lead to local market power for firms \citep{klemperer}.

 Therefore, a key force in our model is firms' equilibrium behavior to segregate their workforce by group identity. Indeed, we show empirically that the Chilean EPSW leads to an increase in gender segregation across firms. One may suspect that such segregation is less likely to occur in other localities that enact EPSW.\footnote{For example, recent research shows that gender-based occupational segregation may be especially likely when the local language has gendered nouns, as firms can target their hiring to workers of specific genders \citep{KSZ2020,kuhn,ccl21}. This may explain the  high baseline level of gender segregation in Chile, where Spanish is the official language. Gendered nouns and targeted hiring may also play a role in the ability of Chilean firms to further segregate once EPSW is enacted. Our search model includes this feature, by allowing workers to direct their search based on the segregation status of firms.} Speaking to this point, however,  group-based segregation across firms has been noted in the US \citep{blau,NBV1996,HN2008,goldin}, and recent research \citep{FK2018} argues that this segregation has increased over time. Therefore, it seems plausible that EPSW could further affect segregation in a wide variety of labor markets.

\section{Static Model}\label{sec:static model}

\subsection{Model Description}
There are two firms, 1 and 2, and a continuum of workers. Each worker is endowed with a type $e=(g,v) \in \{A,B\} \times [0,1]$, where  $g \in \{A,B\}$ is the worker's group identity (say, men and women) and $v\in[0,1]$ is the worker's productivity. 
There is a $\beta \geq 1$
measure of $A-$group workers and $1$ measure of $B-$group workers. $F_A$ and $F_B$ are cumulative distribution functions governing the productivities of workers in groups $A$ and $B$, respectively. $F_A$ and $F_B$ are absolutely continuous and thus admit density functions $f_A$ and $f_B$, respectively. For each $g \in \{A,B\}$, we assume that $0<\underline f_g \le \bar f_g<+\infty$, where $\underline f_g=\inf \{f_g(v)| v \in [0,1]\}$ and $\bar f_g=\sup \{f_g(v)| v \in [0,1]\}$.
 Note that the 
distribution of $A-$group workers may be different from that of $B-$group workers, allowing us to model situations in which firms discriminate against one of the groups of workers (i.e. the output of $B-$group workers are drawn from the same distribution as $A-$group workers, but the firms' willingness to pay for them is lower because firms have a taste-based preference for $A-$group workers). For each $g\in\{A,B\}$ we define 
$
\mathbb{E}_g(v):=\int_0^1 v f_g(v) dv.
$

A \emph{(labor) market} is a tuple $(\beta,F_A,F_B).$ For expositional ease, we study this environment via a cooperative game (all of our model predictions are unchanged if we instead consider a non-cooperative game, see \Cref{noncooperative-and-cooperative}). 
An outcome specifies, for each worker, the firm she works for (or the outside option of staying unemployed) and the wage received (if employed). Each worker only cares about her wage. A firm generates infinitesimal profit $v-w$ if it hires a worker of productivity $v$ and pays her wage $w$. The firm does not have any capacity constraint (i.e. the firm can hire any measure of workers), and its payoff is the integral of profit generated from workers it employs.

Formally, an outcome for firm $i$ is
 $O_i:=\{(f_g^i(v),w_i^g(v))\}_{v \in [0,1], g=A,B}$, where: 
\begin{enumerate}
    \item $f_g^i(v) \in [0,f_g(v)]$ is the density of workers of type $e=(g,v)$ hired by firm $i$,
    \item $w_i^g(v) \in [0,\infty)$ is the wage firm $i$ pays to each worker of type $e=(g,v)$ it hires. If $f_g^i(v)=0$, then we fix $w_i^g(v)=0$.
\end{enumerate}
An outcome is a tuple $O:=(O_1,O_2)$ where $O_i$ is the outcome for firm $i$ such that $f_g^1(v)+f_g^2(v) \le f_g(v)$ for each $v$ and $g$. That is, the (overall) outcome specifies the outcome for both firms such that the total hiring does not exceed the supply of workers (a feasibility requirement). We assume that $f^i_g$ and $w^g_i$ are measurable functions for each $i$ and $g$.

Under an outcome for $i$, $O_i:=\{(f_g^i(v),w_i^g(v))\}_{v \in [0,1], g=A,B}$, firm $i$ receives profit 
$$
\pi^{O_i}_i:=\beta \int_0^1 [v-w_i^A(v)]f^i_A(v)dv + \int_0^1 [v-w_i^B(v)]f^i_B(v)dv.
$$ Given an outcome $O=(O_1,O_2)$ and firm $i$, we denote $\pi^{O}_i:=\pi^{O_i}_i$, and denote by $AW^O_g$  the average wages for group $g \in \{A,B\}$, i.e.\footnote{Note that each unemployed worker from group $g$ contributes a wage of 0 to the calculation of the average wage for group $g$.}
\begin{align*}
AW^O_g:=\stackrel[0]{1}{\int}w^g_1(v) f^1_g(v)dv+\stackrel[0]{1}{\int}w^g_2(v) f^2_g(v)dv.
\end{align*}
\noindent  
We refer to $AW^O_A-AW^O_B$ as the \emph{wage gap} in outcome $O$. Similarly, for  $(AW^O_A,AW^O_B) \neq (0,0)$, we refer to $AW^O_A/AW^O_B$ as the \emph{wage ratio (or log wage gap)} in outcome $O$ (assuming that $AW^O_B \neq 0$, and if $AW^O_A>AW^O_B=0$ then we define the wage ratio in outcome $O$ to be $+\infty$).\footnote{The wage ratio can be interpreted as the log wage gap because the latter, i.e. the gap in the log wages between the two groups, $\ln(AW_A^O)-\ln(AW_B^O)$, is a monotonic transformation of the former by the natural logarithm;  $\ln(AW^O_A/AW^O_B)=\ln(AW_A^O)-\ln(AW_B^O)$.}

We view two outcomes for firm $i$, $O_i:=\{(f_g^i(v),w_i^g(v))\}_{v \in [0,1],  g=A,B}$ and $\tilde O_i:=\{(\tilde f_g^i(v),\tilde w_i^g(v))\}_{v \in [0,1], g=A,B}$ as equivalent if, for each $g \in \{A,B\}$, $f^i_g(v)=\tilde f^i_g(v)$ and $w_i^g(v)=\tilde w_i^g(v)$ for almost all $v$. We view two outcomes $O$ and $\tilde O$ as equivalent if either 1) $O_i$ is equivalent to $\tilde O_i$ for every $i \in \{1,2\}$, or 2) $O_1$ is equivalent to $\tilde O_2$ and $O_2$ is equivalent to $\tilde O_1$.
The first condition captures the usual notion that the outcomes are regarded as equivalent if both the employment patterns and wages are identical between them except for a measure-zero set. The second condition captures the case in which the employment patterns and wages  are identical almost everywhere once the names of the firms are relabeled---recall that firms are homogeneous in the present model.

An outcome is said to be a \emph{core outcome} (without EPSW) if there is no firm and an alternative wage schedule for a subset of workers such that they are made better off being matched to each other, that is,  both the firm and each worker in the hired subset obtain a higher payoff than in the present outcome.  Formally, 
we say that 
an outcome $O:=\{(f_g^i(v),w_i^g(v))\}_{v \in [0,1], i =1,2, g=A,B}$ is \emph{blocked} by firm $j$ via an alternative outcome (for $j$)    $\tilde O_j:=\{(\tilde f_g^j(v),\tilde w_j^g(v))\}_{v \in [0,1],  g=A,B}$ if $\pi^{\tilde O_j}_j > \pi_j^{O_j}$ and, for each $g \in \{A,B\}$ and almost all $v \in [0,1]$, one of the following conditions hold. Note that, because we define $\tilde O_j$ to be an outcome, it must satisfy all restrictions imposed on an outcome in addition to those listed below:
\begin{enumerate}
\item $\tilde w_j^g(v) \ge  w_j^g(v)$ and $\tilde w_j^g(v) >  w_{-j}^g(v)$,
\item $\tilde w_j^g(v) \ge  w_j^g(v)$ and $\tilde f_g^j(v)+f_g^{-j}(v) \le f_g(v)$, 
\item $\tilde w_j^g(v) >  w_{-j}^g(v)$ and $\tilde f_g^j(v)+f_g^{j}(v) \le f_g(v)$, or
\item $\tilde f_g^j(v)+f_g^j(v)+f_g^{-j}(v) \le f_g(v)$.
\end{enumerate}

These cases enumerate all possibilities for the formation of a blocking coalition. Condition 1 states a ``no wage cuts'' requirement; if firm $j$ weakly raises the wages of all workers involved, and strictly raises wages for workers employed by the other firm $-j$, then these workers are all willing to join the blocking coalition. Condition 2 considers the case in which firm $j$ does not need to poach workers from firm $-j$ to construct the blocking outcome, so the only constraint on wages is that existing workers' wages are not reduced. Condition 3 considers the case in which firm $j$ does not need to keep any existing workers to construct the blocking outcome, so the only restriction  on  wages is that the wage paid to poached workers is higher than those paid by $-j$ to the same workers. Condition 4 considers the case in which firm $j$ can hire from unemployed workers to construct the blocking outcome, so there is no additional restriction on the wages of these workers.

We say that an outcome $O$ is a \emph{core outcome} if there exists no outcome that blocks it.

\begin{remark}\rm \label{noncooperative-and-cooperative}
Implicit in the above machinery is the assumption that each firm pays the same wage to all workers of the same type that it employs. We also note that our model is cooperative in nature. Core is the standard solution concept in models of matching with transfers, e.g. \citet{shapley1971assignment} and \citet{kelso1982job}. Indeed, our setting without EPSW can be regarded as a specific case of the latter paper with two homogeneous firms with additive surplus functions (except that we have a continuum of workers, a feature we regard as merely a technical difference).

An alternative approach would be to set up a non-cooperative game and analyze its equilibria. In Appendix \ref{section-noncooperative-games}, we present a non-cooperative game wherein the firms simultaneously make wage offers to workers, and each worker accepts at most one of the offers she receives. The pure-strategy subgame perfect Nash equilibrium outcomes of this game exactly coincide with the set of core outcomes of the cooperative game we describe above under the assumption that firms must pay the same wage to all workers of the same type. We choose  to present the cooperative framework in the main text because its exposition is simpler, and the equivalence mentioned here provides a noncooperative foundation for doing so. We also note that our results extend (in both the cooperative and non-cooperative models) straightforwardly to the case in which firms can pay workers of the same type different wages under a maintained assumption of measurability of the wage schedule.
\end{remark}

\begin{remark}\rm \label{EPC} The definition of block implies two restrictions any core outcome must satisfy. First, Condition 3 of the definition of block immediately implies the following \textbf{Equal Profit Condition}:
\begin{quote}
In any core outcome $O$, $\pi_1^{O}=\pi_2^{O}$.
\end{quote}

This is because otherwise the firm earning strictly lower profit could fire all of its existing workers and hire all of the workers employed by the other firm with an arbitrarily small wage increase. Second, the definition of the core implies the following 
\textbf{Individual Rationality Condition} for firms: 
\begin{quote}
In any core outcome, there is no set $V\subset[0,1]$ of positive Lebesgue measure, a group $g\in\{A,B\}$, and a firm $i\in\{1,2\}$ such that $w_i^g(v)>v$ for all $v\in V.$ 
\end{quote}

Intuitively this is because, if there were, firm $i$ could simply fire all of the workers in question (i.e. set $\tilde f_g^i(v)=0$ for all $v\in V$) and increase its profit. A formal argument for the case without EPSW is given in the proof of \Cref{no_EPL_equilibrium_prop}, and an essentially identical argument extends this observation to the case with EPSW as well.

\end{remark}

\subsection{Results}

In this section, we present theoretical results from our model. Throughout, we fix an arbitrary labor market $(\beta,F_A,F_B)$ and present results within this labor market.

\subsubsection{Core without EPSW}

\begin{proposition}
Without EPSW, there exist a continuum of (non-equivalent) core outcomes. In any core outcome, almost every worker is employed and earns a wage equal to her productivity (formally, for all  $i\in \{1,2\}$, all $g\in\{A,B\}$, and almost all $v\in[0,1]$: $f_g^1(v)+f_g^2(v) = f_g(v)$, and $w_i^g(v)=v$ if $f_g^i(v)>0$). \label{no_EPL_equilibrium_prop}\end{proposition}

\Cref{no_EPL_equilibrium_prop} establishes that, while there are multiple core outcomes, they all feature full employment and result in wages to each worker equal to their productivity. Our proof demonstrates that this result obtains via classic ``Bertrand'' reasoning: in any outcome satisfying individual rationality that fails the conditions of \Cref{no_EPL_equilibrium_prop} for (a positive measure set of) workers of type $(g,v)$ one firm can outcompete the other and hire $(g,v)$ workers in a way that increases its profit.

By \Cref{no_EPL_equilibrium_prop}, the wage gap and wage ratio  in any core outcome $O$ without EPSW are, respectively,
\begin{align*}
    AW_A^O-AW_B^O
    =\stackrel[0]{1}{\int}v f_A(v)dv-\stackrel[0]{1}{\int}v f_B(v)dv =\mathbb{E}_A(v)-\mathbb{E}_B(v)
\quad \text{ and } \quad
AW_A^O/AW_B^O=\mathbb{E}_A(v)/\mathbb{E}_B(v).
\end{align*}

\subsubsection{Core with EPSW}\label{group-basedEPL-subsection}
Now we study core outcomes  under EPSW. Informally, this restriction requires that each firm pays the same wages to almost all workers it hires  if it hires a positive measure of workers from both groups. Formally, we modify the definition of outcome  $O_i=\{( f_g^i(v),w_i^g(v))\}_{v \in [0,1], g=A,B}$ for all $i\in\{1,2\}$  to include the following restriction:

\begin{quote}
 There exist
 no  sets $V_g\subset[0,1]$ and $V_{-g}\subset[0,1]$ with positive Lebesgue measure such that:
\begin{enumerate}
    \item $f_g^i(v)>0$ for all $v\in V_g$ and $f_{-g}^i(v)>0$ for all $v\in V_{-g}$, and 
    \item $\underset{v\in V_{-g}}{\inf~}w^{-g}_i(v)>\underset{v\in V_{g}}{\sup~}w^{g}_i(v)$.
\end{enumerate}
\end{quote}

The preceding restriction prevents a firm from employing  sets of workers from both groups with positive measure (point 1) such that all workers in one set receive strictly higher pay than all workers in the other (point 2).\footnote{Note that because the above restriction must hold for every set $V_g$ and $V_{-g}$ of positive measure, we could equivalently state point 2 using the essential infemum and essential supremum of the wages, respectively.} Given its symmetry across groups,  EPSW implies, by transitivity, that if a firm hires a positive measure of workers from both groups, it must pay almost all workers the same wages.

The next result shows  that generically  firms must fully segregate by group in any core outcome under EPSW.\footnote{We consider the space of productivity distributions satisfying our regularity conditions. We endow this space with the weak-$^*$ topology and  say that a property holds generically if it holds in an open and dense subset with respect to the product topology over the product set of distributions.  An example of a non-generic case in which the conclusion of this result fails features $\beta=1$ and $F_A(v)=F_B(v)=v$ for all $v\in[0,1]$. For this parameterization, it is straightforward to verify that there exists a core outcome where firm 1 hires all workers from both groups with $v \in [0,\frac{1}{2}]$ at wage zero while firm 2 hires all other workers at wage $\frac{1}{2}$.} 
Our proof shows that any other market configuration--both firms segregate for the same group, both desegregate, or only one firm desegregates--generically cannot be supported in a core outcome. The intuition for this result can be obtained by analogy to Hotelling's model: a firm can ``cater to'' a group $g$ by segregating towards $g$, which allows this firm to fiercely compete for these workers as it is not bound by EPSW. Both firms segregating toward some group $g$ leads them to ``Bertrand compete'' away all profit, while workers of the other group could instead be employed at low wages to form a block. If only one firm desegregates, it must pay a common wage to all workers, and can thus be outcompeted for workers of the group the other firm segregates toward to form a block. If both firms desegregate, we show that the only core outcome that can potentially exist involves exactly one firm paying a wage of zero to all its workers, but that such an outcome cannot be in the core for ``almost all'' distribution functions $F_A$ and $F_B$.

\begin{proposition} Generically, in any core outcome under EPSW, firms completely segregate. Specifically, one firm hires almost all $A-$group workers,  and the other hires almost all $B-$group workers (formally, for some $i\in\{1,2\}$, $f^i_A(v)=f_A(v)$ for almost all $v\in[0,1]$ and $f^{-i}_B(v)=f_B(v)$ for almost all $v\in[0,1]$). \label{generic_segregation}\end{proposition}

Following the previous result, we assume throughout that any core outcome $O$ under EPSW exhibits full segregation by group: without loss of generality let firm 1 hire all $A-$group workers ($f^1_A(v)=f_A(v)$ for all $v$) and let firm 2 hire all $B-$group workers ($f^2_B(v)=f_B(v)$ for all $v$). Moreover, 
letting $w_1(\cdot)$ specify firm 1's wages to $A-$group workers and letting $w_2(\cdot)$ specify firm 2's wages to $B-$group workers, Individual Rationality (see \Cref{EPC}) implies that it suffices to consider  $w_i(\cdot):[0,1]\to[0,1]$ for each $i\in\{1,2\}$ in any core outcome.   Note that we can therefore represent the wage gap and the wage ratio in a core outcome $O$ under  EPSW as, respectively,\footnote{The given expression for the wage ratio holds under the assumption that $w_2(v)$ is not equal to zero almost everywhere. Otherwise, by construction, the wage ratio is defined as $\infty.$}
\begin{align*}
\hspace{-8mm} AW^O_A-AW^O_B=\stackrel[0]{1}{\int}w_1(v) f_A(v)dv-\stackrel[0]{1}{\int}w_2(v) f_B(v)dv, \quad  \text{and} \quad 
AW^O_A/AW^O_B=\stackrel[0]{1}{\int}w_1(v) f_A(v)dv/\stackrel[0]{1}{\int}w_2(v) f_B(v)dv.
\end{align*}

The existence of a core outcome holds  generally under EPSW; consider the zero-profit outcome in which firm 1 hires all $A-$group workers and firm 2 hires all $B-$group workers, and all workers are paid wages equal to their productivities. Such an outcome is a core outcome  because no firm is bound by the equal pay constraint of EPSW due to group-level segregation and because adding any new workers requires poaching workers from the competing firm, and such poaching would violate individual rationality. This outcome is similar to core outcomes without EPSW, except that the workforce is now necessarily segregated. 
However, the next proposition shows that there
are also other core outcomes under EPSW; Indeed, there always exist a continuum of (non-equivalent) core outcomes. The proposition also shows that EPSW relatively benefits the majority group of workers, and both firms prefer core outcomes in which wage inequality is maximized.

\begin{proposition} Suppose there is an EPSW restriction.~ \label{two-sided-EPL-equilibrium}
\begin{enumerate}
\item There exist a continuum of non-equivalent core outcomes. \label{continuous-equilibria}
\item Let $\beta > 1$. There exists one core outcome (and its equivalent outcomes) that yields the same wage gap as  the (essentially unique) core outcome without EPSW. In all other core outcomes under EPSW, the wage gap is strictly larger. \label{gap-is-larger-with-EPL}
\item Let $\beta>1$. Consider any two core outcomes. The wage gap  is larger in the first outcome if and only if firm profit is higher in the first outcome.   \label{gap_increases_profit} 
\item Let $\beta \mathbb{E}_A[v]>\mathbb{E}_B[v]$. Then the statements of Parts \ref{gap-is-larger-with-EPL} and \ref{gap_increases_profit} hold (for any $\beta\geq 1$) when replacing ``wage gap'' with ``wage ratio.''\label{ratio_increases_profit}

\end{enumerate}
\end{proposition}

As demonstrated by Part \ref{gap-is-larger-with-EPL} of \Cref{two-sided-EPL-equilibrium}, EPSW moves the wage gap in favor of the majority group (and does so strictly except for one core outcome among a continuum). Moreover, Part \ref{gap_increases_profit} shows that larger wage gaps are associated with higher firm profits under EPSW. An implication of Part \ref{gap_increases_profit} is that firms prefer core outcomes that result in larger wage gaps, suggesting that a core outcome with a larger wage gap may be more likely to occur if firms can coordinate to select an outcome from the core. Part \ref{ratio_increases_profit} 
 replicates the findings of Parts \ref{gap-is-larger-with-EPL} and \ref{gap_increases_profit} when inequality is measured via the wage ratio instead of the wage gap, under the assumption that $\beta \mathbb{E}_A[v]>\mathbb{E}_B[v].$

 We emphasize the paucity of assumptions about relative productivities of the two groups.
Parts \ref{gap-is-larger-with-EPL} and \ref{gap_increases_profit} make no such restrictions, implying our model predicts that EPSW moves the wage gap in favor of the majority group regardless of underlying differences (including discrimination) across groups. The condition that $\beta \mathbb{E}_A[v]>\mathbb{E}_B[v],$ which is sufficient to extend this conclusion to the wage ratio, says that the majority group is not ``much'' less productive (or is not ``much'' more discriminated against) than the minority group. While this is clearly a restriction, evidence from our empirical setting (detail is provided in \Cref{empirical section}) suggests that it is true: across male-majority labor markets, the pre-EPSW log wage gap is larger than one, which implies by \Cref{no_EPL_equilibrium_prop} that $\beta \mathbb{E}_A[v]>\mathbb{E}_B[v]$. Across female-majority labor markets the pre-EPSW log wage gap is 0.994 while $\beta$ is 1.31, implying $\beta \mathbb{E}_A[v]>\mathbb{E}_B[v]$.

\section{Search Model}\label{section:search-model}

In this section, we present a search model and use it to study the evolution of wages and segregation upon the enactment of EPSW. In order to consider the incentives facing agents over time, we consider a non-cooperative dynamic game. A continuum of workers enters a labor market in each period, and each worker searches at, and bargains with, firms in set $I$  where $\vert I \vert=n\geq 1$. Each worker has a type $(g,v)$, where $g\in \{A,B\}$ is her group identity and $v \in [0,1]$ is her productivity. 

Before proceeding, we discuss several important features of our model which match our empirical context. First, we study a dynamic setting in which EPSW is a ``surprise,'' meaning that agents initially optimize their strategies without consideration of EPSW and then unexpectedly must contend with it. Second, not all firms are bound by EPSW when it is enacted; such ``unconstrained'' firms represent those below a binding size threshold, or those which hire predominantly fixed-term labor, for which EPSW in Chile does not bind. Third, workers' choice sets are endogenously determined by firm segregation decisions, as firms segregated for a group $g$ will be unwilling to hire workers of group $g'$ as doing so would expose the firm to the bite of the policy. Moreover, recent research \citep{KSZ2020,kuhn,ccl21} suggests that a firm in a labor market where the written language has gendered nouns---as does Spanish, the de facto official language of Chile---are able to significantly affect the set of workers that apply for employment through use of masculine or feminine nouns in job postings (e.g. ``mec\'anico'' and ``mec\'anica'' are respectively the masculine and feminine Spanish translations of ``mechanic''). We capture both of these features by allowing workers to observe firms' segregation choices before engaging in directed search. Finally, the workers  entering in each period have productivities distributed according to $F_A$ and $F_B$ (satisfying the same regularity conditions as in \Cref{sec:static model}), where informally, the density $f_g(v)$ denotes the proportion of $g-$group workers with productivity $v$. We proceed in the sequel to study the economic forces underlying our agents with this informal understanding of the meaning of our distribution functions. In \Cref{section-noncooperative-games} we show that this interpretation can be made rigorous in the sense that we construct measures of worker types that mimic the underlying distributions $F_A$ and $F_B.$

\subsection{Model Description}
\label{search-model}

\subsubsection{Basic Setup}\label{basic setup}

There is a countably infinite number of time periods, $\{-\underline t, -\underline t+1,\dots, -1,0,1,2,\dots\}$, with $\underline t$  being a positive integer. EPSW is introduced between the end of period $-1$ and the beginning of period $0$. We refer to the set of periods $\{-\underline t, -\underline t+1,\dots, -1\}$ as the ``pre EPSW'' periods and the set of periods $\{0,1,2,\dots\}$ as the ``post EPSW'' periods, respectively. We study a game of perfect information, where each player observes all prior actions taken by all players.

First, we describe the game in the pre EPSW periods $\{-\underline t, -\underline t+1,\dots, -1\}$.
At the beginning of period $-\underline t$, no worker is employed by any firm.  Each period $t \in \{-\underline t, -\underline t+1,\dots, -1\}$ unfolds as follows:

\begin{enumerate}
    \item \label{workers born step} A measure $\beta \geq 1$ of $A-$group workers and a unit measure of $B-$group workers are born. Worker productivities are distributed according to distributions $F_A$ and $F_B$, respectively.   
    \item Each newly born worker directs her search by selecting a subset of firms and a linear order over them. The worker meets with firms along the selected ordering to bargain over wages. Bargaining details will be described in subsequent sections. \label{bargaining stage}
    \item All newly born workers who do not come to a bargaining agreement with any firm exit the market. All other workers exit the market, each with independent probability $d\in(0,1).$
\end{enumerate}

 When EPSW is introduced between the end of period  $-1$ and the beginning of period $0$, there are $n_U \ge 1$ exogenously fixed firms that are ``unconstrained'' by the policy. In each period $t \in \{0,1,\dots\}$, the actions available to unconstrained firms do not change. The remaining $n_C =n-n_U$ ``constrained'' firms must decide how to adhere to the policy: either by desegregating at a common wage for all workers, or by segregating for one group. Because the introduction of an EPSW is a surprise and firms have not yet had an opportunity to make a segregation decision, we take all constrained firms to be desegregated at the beginning of time 0.\footnote{As mentioned earlier, some firms are  unconstrained by EPSW in our empirical setting. Our results---specifically, \Cref{existence-result} and \Cref{near seg prop}---also hold more generally if 
an arbitrary number of firms are initially segregated at the start of period $0$.
In our empirical setting, 34.3\% of firms employ workers from only a single group at the time of announcement of EPSW.}
Formally, for each $t\geq 0$, the following step is added immediately prior to Step \ref{workers born step} of the timing of the game for constrained firms.
 
\begin{enumerate}
 \setcounter{enumi}{-1}
    \item ~
    \begin{enumerate}
  \item Following an exogenously given order, each currently desegregated firm sequentially chooses from the following options: segregate for group $A$, segregate for group $B$, or desegregate at any endogenously selected wage $w\in[0,1]$. The decision to segregate is  irreversible, i.e. once the firm has segregated for a group $g$, then the firm continues to be segregated for group $g$ in all subsequent periods.\footnote{The assumption that the firms' segregation decisions are sequential is made for tractability. Specifically, this assumption gives each firm $i$ perfect information about the (de)segregation decisions of others when $i$ makes its (de)segregation decision in any period $t>0.$ In the proof of \Cref{existence-result}, we show how the presence of perfect information leads to the existence of a pure-strategy subgame perfect Nash equilibrium, following the well-known existence result of \cite{har85perfect}. The assumption that segregation is irreversible is made for simplicity, and similar results obtain with minor modifications to our equilibrium refinement.} The (de)segregation decision potentially involves job separation and wage adjustments:
\begin{itemize}
    \item If a firm is segregated for group $g$, then the firm is immediately unmatched from all previously employed workers with group identity $g' \neq g.$
    \item If a firm is desegregated at wage $w$, then all previously employed workers with productivity $v<w$ and all previously employed workers earning wage $w'>w$ are immediately unmatched from the firm. The wages of all remaining workers are immediately increased to $w$. The firm pays $w$ to all newly hired workers in period $t$.
\end{itemize}
\item Each worker who is separated from their previous employer at this stage exits the market, and with probability $\rho\in[0,1]$ is replaced by a new unemployed worker of the same type.\footnote{$\rho\in\{0,1\}$ corresponds, respectively, to markets without and with replacement, which are commonly considered settings in the existing literature  (for examples of search models with replacement, see \citet{BC1999}, \citet{BR2000}, \citet{chade}, and \citet{CPH}). We allow for the possibility of ``partial'' replacement, $\rho\in(0,1)$, for generality.}
\end{enumerate}
\end{enumerate}

Moreover, in post EPSW periods, a worker with group identity $g$ is not eligible to be matched with any firm that is segregated for group $g'\neq g$. Similarly, a worker with productivity $v$ is not eligible to be matched with any desegregated firm whose currently selected wage is strictly higher than $v$.\footnote{\label{footnote cost of violation}Here, we are directly prohibiting a worker from being matched with firms that are segregated for the opposite group. This assumption is made only for convenience. If we assume that there is a sufficiently large penalty for firms hiring from both groups while paying different wages, then no such matching occurs in equilibrium. In a similar spirit, no worker is matched to a desegregated firm whose wage is higher than her productivity in equilibrium.} Formally, a worker of type $(g,v)$ is unable to include such firms in her bargaining order.

Each worker of type $(g,v)$  
who is employed by a firm  at a going wage of $w$ receives a per-period payoff of $w$,  and the employing firm  receives an infinitesimal per-period payoff of $v-w$ from each such worker if it employs a measurable set of workers and pays a measurable wage schedule, and otherwise the firm receives a per-period payoff of $-1$. Unemployed workers receive a payoff of zero. All agents geometrically discount future payoffs according to discount factor $\delta\in[0,1).$ 

\subsubsection{Bargaining}

 We study a ``Nash-in-Nash'' bargaining solution a la \cite{horn88nash}. Specifically, and as we define shortly, the available surplus between each worker and firm is fractionally split, assuming that the outcome between the worker and each subsequent firm she potentially bargains with follows an analogous fractional split. The worker's ``disagreement point'' is the wage she receives from bargaining with the subsequent firm in her selected order, i.e. if bargaining were to break down with the current firm. Agent strategies, which take as fixed the outcome of any bilateral bargaining as described above, must satisfy subgame perfection in the standard way. In \Cref{microfoundation app}, we discuss hiring capacities within this model (\Cref{remark: capacity constraint}), and we also provide an alternating-offer microfoundation for the fractional split bargaining protocol (\Cref{nash_bargaining}).
 
 One complexity our setting induces is that a worker's employment status and wage is potentially affected by endogenous decisions. As described above, a worker of type $(g,v)\in\{A,B\}\times[0,1]$ can be ``fired'' from her job in three cases: her employing firm segregates for the opposite group $g'\neq g$, her employing firm desegregates at a wage $w>v$, or her employing firm desegregates at some time $t$ at a wage $w$ that is lower than what it was previously paying the worker. Our upcoming framework accounts for these endogenous firm decisions that affect the ``available surplus'' to be split.

 To define the fractional split our bargaining protocol takes, consider a period $t$ and consider a newly born worker of some type $(g,v)\in\{A,B\}\times[0,1]$. Let $D_t^v$ be the set of all desegregated firms that set a wage $w_i(t)\leq v$ in period $t$. Throughout, we adopt the expositional convention that all firms are unconstrained in pre-EPSW periods $t<0$. Therefore,  $D^v_t=\emptyset$ for all $v$ if $t<0.$  Let $T^v_i \in \mathbb N \cup \{\infty\}$ be the number of periods until $i \in D_t^v$ fires  the relevant worker, i.e. period $T^v_i+t$ is defined as the first period in which either $i$ segregates for group $g'\neq g$ or $i$ desegregates at some wage $w_i(T^v_i+t)>v$ or some wage $w_i(T^v_i+t)<w_i(\tau)$ for some $\tau\in\{t,t+1,...,T^v_i+t-1\}$. Therefore, $T_i^{v}$ is an equilibrium object for workers born in periods $t\geq 0,$ and $T_i^{v}=\infty$ for all workers born in periods $t<0$ due to the ``surprise'' feature of EPSW.
Define
\begin{align}\label{V-definition}
V^v:=\sum_{\tau=1}^\infty \delta^{\tau-1} (1-d)^{\tau-1}v, \quad \text{and} \quad 
W^v_i:=\sum_{\tau=1}^{T_i} \delta^{\tau-1}(1-d)^{\tau-1} w_i(t+\tau-1).
\end{align}
\noindent $V^v$ and $W^v_i$ represent, respectively, the total expected surplus generated by the worker if she is never fired and her total expected earnings if she accepts a job at desegregated firm $i$. Let $W^v=\max \{W_i^v | i \in D^v_t\}$ if $D^v_t \neq \emptyset$ and $W^v=0$ otherwise. Note that $V^v \ge W^v$ by definition, i.e. $W^v$ serves as the worker's ``outside option'' that she has access to by working at a desegregated firm.

Now, we will define our bargaining outcome for any bargaining order. Let $(i_1,\dots, i_m)$, $m \ge 0$, denote the worker-selected bargaining order over a subset of  eligible firms, where we denote the bargaining order by $\emptyset$ if $m=0$. Let $V^{(i_1,\dots, i_m)}$ be the surplus accruing to the worker  given this order. We define the bargaining outcome by induction on the number of the firms $m$ who bargain with the worker as follows.
\begin{enumerate}
    \item Suppose that $m=0$,  i.e., the worker elects not to bargain with any firm. Then, the worker is permanently unmatched and permanently receives a wage of zero.
    \item Given the bargaining outcome for an order $(i_2,...,i_{m+1})$, we will define the bargaining outcome for order $(i_1,...,i_{m+1})$ as follows.\footnote{We adopt the convention that $(i_2,...,i_{m+1})=\emptyset$ if $m=0$.}
    \begin{enumerate}
        \item Suppose that firm $i_1$ is desegregated and $W^v_{i_1}<V^{(i_2,\dots,i_{m+1})}$. Then, the bargaining between the worker and firm $i_1$ breaks down, and the outcome for bargaining order $(i_2,\dots,i_{m+1})$ is realized.
        \item Otherwise,  
        \begin{enumerate}
            \item if $i_1$ is desegregated, then the worker is matched to firm $i_1$ at a wage of $w_{i_1}(t)$, with a realized payoff of $W^v_{i_1}$.
            \item if $i_1$ is segregated, then the worker is matched with firm $i_1$, with a realized payoff of 
            \begin{align}\label{inductive surplus}
                V^{(i_1,\dots,i_{m+1})}=\Delta V^{(i_2,\dots,i_{m+1})} + (1-\Delta)V^v ,
            \end{align}
            where $\Delta \in (0,1)$  is an exogenously given parameter that represents the bargaining power of the firm.  The worker earns the unique per-period wage $w\in[0,v]$ that yields the worker expected lifetime payoff $V^{(i_1,\dots,i_{m+1})}.$\footnote{To show the existence and uniqueness of such a wage, note that $V^{(i_1,\dots,i_{m+1})}\in[0,\sum_{\tau=1}^\infty \delta^{\tau-1} (1-d)^{\tau-1}v]$ because $0\leq V^{(i_1,\dots,i_{m+1})}\leq  V^{v}=\sum_{\tau=1}^\infty \delta^{\tau-1} (1-d)^{\tau-1}v,$ where the first inequality follows because both terms in \eqref{inductive surplus} are weakly positive, the second inequality follows because $V^{(i_2,\dots,i_{m+1})}\leq V^v$ by the induction hypothesis (recall that $V^\emptyset=0$), and the equality follows from \eqref{V-definition}. Therefore, the worker's expected lifetime earnings, as a function of the per-period wage $w$, is $\sum_{\tau=1}^\infty \delta^{\tau-1} (1-d)^{\tau-1}w$, which equals 0 for $w=0$, equals $V=\sum_{\tau=1}^\infty \delta^{\tau-1} (1-d)^{\tau-1}v$ for $w=v$ and is strictly increasing and continuous in $w$ over the interval $[0,v].$ Therefore, there is a unique per-period wage $w\in[0,v]$ that corresponds with expected worker lifetime earnings specified in \eqref{inductive surplus}.}
        \end{enumerate}
    \end{enumerate}
\end{enumerate}

 Now, let there be $m$ firms that are either segregated for the worker's group $g$ or unconstrained, and  let $W^v$ be as defined in \eqref{V-definition}.  Then, consider the strategy of the worker such that the worker bargains with the $m$ segregated or unconstrained firms and then bargains with a desegregated firm with wage $W^v$ if $D^v_t \neq \emptyset$. By repeated application of \eqref{inductive surplus} to this bargaining order, the worker receives an expected payoff of
\begin{align}\label{worker-surplus-solution}
V_1^v:=\Delta^{m} W^v + \left (1-\Delta^{m} \right ) V^v.
\end{align}
Note that the payoff $V_1^v$ in \eqref{worker-surplus-solution} is weakly increasing in $m$, and strictly so if and only if $W^v<V^v$.

We note that an optimal strategy for the worker is to choose a bargaining order that first lists all segregated and unconstrained firms (we take there to be $m$ such firms) and then lists a desegregated firm with wage $W^v$ if $D^v_t \neq \emptyset$. To see this, let $i$ be the first desegregated firm in the worker's bargaining order (if such a firm exists) with the property that, conditional on failing to reach a bargaining  agreement with all firms ordered before it, the worker and firm reach an agreement. Let $m' \le m$ be the number of desegregated or unconstrained firms ordered before $i$, and let  $W' \le W^v$ be the wage set by firm $i$ if such firm $i$ exists, and let $m'=m$ and $W'=0$ if no such firm $i$ exists. Then, by the same argument as before, the worker reaches an agreement with the first firm (potentially excepting desegregated firms other than $i$) and obtains expected surplus
\begin{align}
\Delta^{m'} W' + \left (1-\Delta^{m'}\right ) V^v,
\end{align}
where we note that the set of desegregated firms ordered before $i$ does not affect this surplus because they will not reach an agreement with the worker.
Because $m' \le m$ and $W' \le W^v$ by definition of $m$ and $W^v$, this expression is no larger than \eqref{worker-surplus-solution}, showing the optimality of the indicated order. Therefore, in any $t$ and in any equilibrium, we consider bargaining orders that first list all desegregated or unconstrained firms, and then (one of) the desegregated firm(s) offering $W^v$.

We highlight two equilibrium restrictions that rule out pathological strategies. First, implicit behind \eqref{V-definition}  is a ``Markovian'' assumption on firm strategies, that is, we focus on subgame perfect equilibria such that a firm's decisions are not impacted by payoff-irrelevant deviations by workers. Specifically, given any time period $t$, a firm's decision to 
segregate for some group $g$ at any time $t'>t$ or to desegregate at some wage $w$ at any time $t'>t$ is not affected by a zero-measure set of workers who deviate in their selected bargaining order in period $t$.
Second, we focus on equilibria in which worker search is affected by firms' (de)segregation choices, and not by other aspects of the history of play.  In particular, we exclude strategies in which workers coordinate such that they jointly avoid searching first at some firm $i$, as this prevents firm $i$ from hiring workers regardless of its segregation decision. To formally define our refinement, for $t\geq 0$, let $I^g_t$ be the set of unconstrained firms and firms segregated toward group $g$ by the hiring stage of period $t$, and for $t<0$ let $I^g_t=I$. We assume there is a  \emph{search intensity} $r=(r_i^g)_{i \in I, g \in \{A,B\}} \in (0,1)^{2n}$ such that at any time $t$ and any history of play, for any firm $i \in I_t^g$, and for each productivity $v$,  $\frac{r_i^g}{\sum_{i' \in I_t^g}r_{i'}^g}$ share of the entering workforce of type $(g,v)$ first elect to bargain with firm $i$. 

 In the following, we focus on pure-strategy subgame perfect Nash equilibria with the above restrictions on bargaining and search intensity (and the aforementioned ``surprise'' feature of EPSW), and refer to them simply as \emph{equilibria.}

\subsection{Results}\label{search_model_results}

In this section, we present the main results arising from our search model. Our first such result, \Cref{existence-result}, makes three points. First, for any $r\in (0,1)^{2n},$ the set of equilibria with search intensity $r$ is non-empty and in each equilibrium every newly born worker is hired.

Second, all firms eventually segregate in any equilibrium. Moreover, there is a uniform bound $T\geq 0$ such that in any equilibrium, all firms have segregated by time $T$. For any $t\geq 0$, let $n^{t}_g,$ $g\in\{D,A,B\}$ represent the number of firms that are desegregated, segregated toward $A$, and segregated toward $B$, respectively, at the end of time period $t$. Therefore, $n^{t+1}_D=n^t_D(=0),$ $n^{t+1}_A=n^t_A,$ and $n^{t+1}_B=n^t_B$ for any $t\geq T$ in equilibrium. We will denote those constant values by $n_D, n_A,$ and $n_B$, respectively. Note that because every newly born worker is hired, there is full employment for workers born after period $T$ when all firms have segregated. Thus, the only possible source of unemployment is worker displacement in periods $t\in\{0,...,T\}$ as firms make (de)segregation choices. 
These displaced workers will be ``rehired'' at rate $\rho$, and therefore, full employment is achieved for all $t$ if $\rho=1$.

Third, the long-run wage ratio moves in favor of the group toward which more firms segregate in equilibrium. For any $t$, we define the wage ratio (or log wage gap) at time $t$ to be the average wage of $A-$group workers over the average wage of $B-$group workers who are employed at the end of time $t$. In the pre-EPSW periods, in equilibrium, every worker is hired by the first firm she bargains with, and therefore both groups' wages follow \eqref{worker-surplus-solution} with $m=n$ and $W^v=0$ for all $v\in[0,1]$, which implies:

\begin{observation}\label{observation_pre_period}
    The wage ratio is $\mathbb{E}_A[v]/\mathbb{E}_B[v]$ in any period $t<0$. 
\end{observation}

We say that the wage ratio is \emph{more in favor of} group $A$ ($B$) in period $t$ compared to period $t'$ if the wage ratio is higher (lower) in period $t$ than in period $t'$. Part \ref{wage-ratio-result} of \Cref{existence-result} shows that the the key determinant of the wage ratio in the long run--that is, the wage ratio between any pre-EPSW period and any sufficiently large period since the introduction of EPSW--is the relationship between $n_A$ and $n_B$: if $n_A>n_B$ then the wage ratio is more in favor of group $A$ in the long run compared to pre-EPSW, and vice versa, and if $n_A=n_B$, then EPSW has no impact on the long-run wage ratio.

\begin{proposition} \label{existence-result} For any $r\in(0,1)^{2n}$, the set of  equilibria with search intensity $r$ is nonempty. In any equilibrium, every newly born worker is hired. Moreover, there exists $T\geq 0$ such that in any  equilibrium:
\begin{enumerate}
    \item \label{claim-segregation} $n^t_D=0$ for all $t \ge T$, and
    \item \label{wage-ratio-result}  If $n_A > n_B$, then the wage ratio is more in favor of group $A$ in any period $t \ge T$ when compared to the wage ratio in any period $t'<0.$ If $n_B > n_A$, then the wage ratio is more in favor of group $B$ in any period $t \ge T$ when compared to the wage ratio in any period $t'<0.$ If $n_A=n_B$ then the wage ratio in period $t$  converges  to the wage ratio in any period $t'<0$ as $t \to \infty$.
\end{enumerate}
\end{proposition}

Part \ref{wage-ratio-result} of \Cref{existence-result} reveals that the number of firms segregating for each group is the key determinant of EPSW's effect on the long-run wage ratio.  
While it may seem intuitive for $n_A\geq n_B$ to hold if  there are more $A-$group workers than $B-$group workers, this condition does not necessarily hold in equilibrium. Three forces can counteract this intuition. First, the minority group may be more productive than the majority: if $\beta \mathbb{E}_A[v]<\mathbb{E}_B[v]$ then the majority group is ``less valuable'' to firms than the minority group, potentially incentivizing more firms to segregate for group $B$. Second, worker search intensity may be skewed: if too many $A-$group workers search for some firm $i$ (i.e. $r_j^A$ is small compared to $\beta$ for all $j\neq i$) while $B-$group workers search more equitably across firms, any firm $j\neq i$ that segregates in favor of group $A$ hires only a small fraction of workers if firm $i$ also segregates toward group $A$, or is unconstrained. Thus, skewed search intensity could incentivize more firms to segregate for group $B$. Third, firms may face high reallocation frictions: if $d$ is low, then firms retain a large fraction of their workforce across periods implying that they have ``more to lose'' in competition with other firms. Therefore,  a firm may segregate toward group $B$ in equilibrium.
 
Therefore, a policy-relevant consideration is under what conditions $n_A>n_B$ holds. Indeed, we note a near-ubiquitous empirical pattern in \Cref{mapping theory section}---more than 94\% of labor markets feature more firms segregating toward the gender that makes up a majority of the workforce post-EPSW---suggesting that there are realistic cases under which equilibria feature more segregation toward the majority group.
\Cref{high d result} below considers equilibria with \emph{equitable search} defined as an equilibrium with search intensity $r_i^g=1$ for every $i\in I$ and $g\in\{A,B\}$, and makes two points. First, given small reallocation frictions, i.e. large $d$, a near equal profit condition arises whereby all firms receive similar profit for sufficiently large $t$, approximating the equal profit condition that arises in any equilibrium of our static game. To define what we mean by ``near equal profit''  let $x_A,x_B \in [0,n_C]$ be real numbers such that 
    \begin{align}
    \frac{\Delta^{x_A+n_U}}{x_A+n_U} \cdot \beta \mathbb{E}_A[v] & =    \frac{\Delta^{x_B+n_U}}{x_B+n_U}\cdot \mathbb{E}_B[v], \label{search-equal-profit-condition} \\
    x_A+x_B & = n_C\label{accounting-equality} \end{align}
if such numbers exist.  If  the left-hand side of \eqref{search-equal-profit-condition} is larger than the right-hand side for every $x_A\in[0,n_C],$ then let $x_A=n_C$. If the left-hand side of \eqref{search-equal-profit-condition} is smaller than the right-hand side for every $x_A\in[0,n_C],$ then let $x_B=n_C.$ If $n_A=x_A$, then \eqref{worker-surplus-solution} implies that all firms will earn equal profit from workers born in periods $t\geq T$, where $T$ is the time period after which all firms have segregated in equilibrium. If $x_A$ is not an integer, then it is impossible to have exactly $n_A=x_A$ but if $n_A$ is ``close to'' $x_A$, then all firms will earn approximately equal profit from workers born in periods $t\geq T$.  \Cref{high d result} shows  that for sufficiently large $d$ (i.e. for any sufficiently large values of $d$ that depend on other parameters but not on the selection of a specific equilibrium)  in any equilibrium with equitable search, $n_A$ is ``close to'' $x_A$ (Part \ref{near-equal-profit-claim}).\footnote{We note that $x_A$ is not a function of $d$.} If  $\beta \mathbb{E}_A[v]>\mathbb{E}_B[v]$ also holds, then $A-$group workers are collectively more valuable than $B-$group workers, implying by \eqref{search-equal-profit-condition} that $x_A>x_B,$ and we show that $n_A\geq n_B$ must hold (Part \ref{segregation-pattern-claim}).

\begin{proposition}\label{high d result} There exists $d^*<1$ such that for all $d>d^*$, and in any equilibrium with equitable search:
\begin{enumerate}
    \item $n_A \in [x_A-1,x_A+1]$ and $n_B \in [x_B-1,x_B+1]$.\label{near-equal-profit-claim}
    \item Let $\beta \mathbb{E}_A[v] > \mathbb{E}_B[v].$ Then $n_A \ge n_B$.\label{segregation-pattern-claim}
\end{enumerate}        
         \end{proposition}

\begin{remark}\rm
As previously discussed in \Cref{sec:static model}, we have empirical evidence to support  that $\beta \mathbb{E}_A[v]>\mathbb{E}_B[v]$ holds. Of course, it is likely not the case in reality that workers search across firms with the exact same intensity.  Without the equitable search assumption, the conclusion of Part \ref{segregation-pattern-claim} of \Cref{high d result} holds under more stringent conditions. Specifically,  an essentially identical argument as that presented in the Proof of \Cref{high d result} shows that: For any $r\in(0,1)^{2n}$, there exist $d^*<1$ and $\beta^*>0$ such that for all $d>d^*$ and all $\beta>\beta^*$, any equilibrium with search intensity $r$ yields $n_A>n_B.$
\end{remark}

Another empirically motivated question is to what extent firms that are ``nearly''  segregated upon EPSW enactment will quickly switch to segregation. In our model, the fraction of the employed workforce at each firm at $t=0$ is driven by search intensities. For distinct $g,g'\in\{A,B\}$, define $R^i_g:=r_i^{g}/r_i^{g'}$ to be the ratio of search intensities. Then a firm will be ``nearly'' segregated toward group $g'$ at the beginning of period $0$ if $R^i_g$ is sufficiently low. The following remark presents conditions under which ``nearly'' segregated firms elect to segregate sooner than all other firms. When the discount rate is sufficiently low, firms become myopic, meaning that the cost of remaining desegregated--separating with some existing workers in both groups, increasing the pay of almost all other employed workers, and being unable to hire new workers--exceeds the costs of segregating--separating with a small measure of employed workers from the ``wrong'' group $g$. The proof is straightforward, and therefore omitted.

\begin{proposition}\label{near seg prop}
        There exist $\delta^*\in(0,1)$ and $R^*\in(0,\infty)$ such that if the discount factor is $\delta<\delta^*$ then in any equilibrium with search intensity $r$, any unconstrained firm $i$ with $R^i_g<R^*$ for some $g\in\{A,B\}$ will segregate toward group $g'\neq g$ at $t=0$.
\end{proposition}

\section{Empirical Analysis}\label{empirical section}

\subsection{Institutional Background}

Chile is an OECD country with nearly 20 million inhabitants. The Chilean labor market is relatively concentrated in the formal employment sector; the informal labor market share around the time of the policy was 25\%,  the lowest in Latin America \citep{gasparini}.  Only 10\% of the (formal) workforce is
unionized, and only union members are covered by collective bargaining 
agreements, implying that Chilean firms plausibly have a high degree of wage-setting power. Workers can be fired without cause and without notice at the cost of one month's wages.

The gender wage gap in Chile is similar to that in the United States. Chilean female workers earn 18--23\% less than their male counterparts \citep{chilewagegap}. Female labor force participation was roughly 30\% in 2009, which is lower than in many OECD and other Latin American countries \citep{femalelaborforceparticipation}.

\subsection{EPSW implementation}

In June 2009, Law 20.348 was signed as an amendment to Chile's labor code with the General Secretary of the Chilean Senate declaring, ``The main objective of the initiative is to establish the right to equal remuneration between men and women for the provision of services of similar value.'' We refer throughout to June 2009 as the time of announcement.
The law took effect in November 2009, which we refer to as the time of enactment. 
An important part of the discussion and debate surrounding the law was providing a definition of ``similar work.''
The law specifies that a firm cannot pay a man and a woman different wages for ``arbitrary reasons''; pay differences across genders are allowable only if workers fall into different coarse categories based on skills, qualifications, suitability, responsibility, or productivity.  Firms that do not comply are subject to sizable monetary fines per offense, as we discuss below. The law also establishes a 10\% discount for any other labor fines a firm is subject to if it pays men and women the same wages for ``similar jobs and responsibilities.'' We therefore classify the law for our purposes as (a gender-based) EPSW.\footnote{Guideline 1187/018 published in April 2010 by the Directorate of Labor clarifies that 1) the law does not bind within gender group, and 2) that a firm paying even a single man more than a single woman (or vice versa) despite both performing similar work is in violation of the law.
}

The law has different consequences for firms of different sizes, based on the number of a firm's workers with long-term employment contracts.\footnote{The vast majority of workers in Chile have either long-term contracts (no end date is specified ex ante) or fixed-term contracts (an end date is specified ex ante, although such contracts are automatically transitioned into long-term contracts if the worker continues to be employed beyond the contract end date). EPSW  protections for workers of different contract types within the same firm are identical. Nearly half of firm-months in our dataset contain no workers with fixed-term contracts.} Firms with 10 or more long-term workers 
are required to explicitly have a grievance procedure for gender-based pay discrimination. Workers in firms above this threshold who allege the firm has violated Law 20.348 must receive a sufficient response from the firm within 30 days. If no such response is received, the worker can file a complaint at the Labor Inspection Office or can directly raise the issue with a labor court. Financial penalties also differ per infraction by firm size. Firms with 10--49 long-term workers found to violate the law are subject to a fine of 69-1,384 USD per worker-month of violation, 
while firms with fewer than 10 long-term workers are not subject to a financial penalty.
\footnote{\cite{cruzrau2022} further discuss how the law imposed mandatory transparency guidelines on worker roles in the firm, and additional fines for violations, for firms with at least 200 long-term workers. They show that the disclosure policy reduced the gender wage gap through a bargaining channel. Our analysis avoids firms treated by this additional policy, due to the potential confounding equilibrium effects of transparency policies \citep[for further discussion on potential equilibrium effects, see][]{CPH}.}

Initial evidence suggests that the law was both widely known to workers, and enforced. In a 2013 governmental survey,\footnote{See \url{https://www.evaluaciondelaley.cl/wp-content/uploads/2019/07/ley_20348_igualdad_remuneraciones.pdf}.} 11\% of respondents stated that they know someone who had complained using the law. Through a public-records request, we found that 9,577 complaints were filed by workers alleging violations of Law 20.348, 9,723 inspections were carried out by the government, and that 489 individual firms were punished. The average fine amount was 2,466 USD per firm (each worker-month of unfair pay is a separate violation of the law).

\subsection{Data}
 We study the effects of EPSW using matched worker-firm administrative data from the Chilean unemployment insurance system from January 2005 to December 2013 \citep{datacitation}. We observe a random sample of firms, stratified by size, totaling roughly 4\% of all firms. In our data, an observation is a worker-firm-month. For each observation, the data include worker pay\footnote{Worker monthly pay at each firm is top coded in our data. The threshold for top coding varies over time; in June 2009, the top-code threshold was roughly \$3,550. The share of observations in our data that are top coded in June 2009 is 0.07\%.} and  demographic information including gender, education level, contract status, age, and marital status; additionally, we observe the firm's geographic location and industry code. We observe the entire employment history of each worker ever employed at a sampled firm. We discuss further details of our dataset in  \Cref{filtering}.

Firms in our sample are typically small, with a median of 7 concurrent workers. However, there are outliers. 
Following \cite{BSTW2018,gulyas,duchini,boheim} we only consider firms of similar sizes at the time of policy announcement in order to limit size-based wage dynamics. In our main specifications, we consider firms with at least 6 and no more than 13 total workers at announcement, which account for 41\% of all active firms. 

In Table \ref{tab:stats_samples}, we present descriptive statistics for (column I) the set of firms prior to the size restriction, (column II) the set of firms following the size restriction,  (column III) the set of firms with that will form our control group, and (column IV) the set of firms following the size restriction that are active over our entire period of analysis.

\begin{table}[ht] \centering
\addtolength{\leftskip} {-2cm}
    \addtolength{\rightskip}{-2cm}
    \caption{Descriptive Statistics}
\label{tab:stats_samples}
   \begin{tabular}{@{}lcccc@{}}
    \hline \hline
     & \multicolumn{1}{p{2cm}}{\centering (I)}
        & \multicolumn{1}{p{2cm}}{\centering (II)}
        & \multicolumn{1}{p{2cm}}{\centering (III)}
        & \multicolumn{1}{p{2cm}}{\centering (IV)}  \\
        & \multicolumn{1}{p{2cm}}{\centering All firms}
        & \multicolumn{1}{p{2cm}}{\centering Estimation Sample}
        & \multicolumn{1}{p{2cm}}{\centering Estimation Sample, Control}
        & \multicolumn{1}{p{2cm}}{\centering Estimation Sample, Balanced}  \\ \hline
        \textit{Panel (A): Employees}                         &       &       &       &             \\
        $\quad$Average age in June 2009                     & 32.418   & 32.180   & 32.186   & 32.123          \\
        $\quad$Share with tertiary education                & 0.124   & 0.124   & 0.120   & 0.128          \\
        $\quad$Share male                                   & 0.693   & 0.734   & 0.736   & 0.732          \\
        $\quad$Share married                                & 0.341   & 0.337   & 0.336   & 0.336          \\
        $\quad$Share with residence in Santiago Region      & 0.369   & 0.389   & 0.380   & 0.391          \\
        $\quad$Share in female majority industry-county in June 2009 & 0.174   & 0.154   & 0.155   & 0.155           \\
        $\quad$Number of workers                            & 955,602  & 473,991  & 435,485  & 311,046             \\ \hline
        \textit{Panel (B): Firms}                             &       &       &       &              \\
        $\quad$Average size in June 2009                    & 12.207   & 8.685   & 8.303   & 8.754           \\
        $\quad$Share in Santiago Region                     & 0.484   & 0.553   & 0.539   & 0.567          \\
        $\quad$Share in Agriculture                         & 0.101   & 0.091   & 0.095   & 0.100           \\
        $\quad$Share in Manufacturing                       & 0.088   & 0.132   & 0.126   & 0.165           \\
        $\quad$Share in female majority industry-county in June 2009 & 0.137   & 0.149   & 0.153   & 0.162           \\
        $\quad$Number of firms                              & 25,436  & 6,551  & 5,791  & 3,418             \\
    \hline \hline
    \end{tabular}

    \begin{minipage}{1.3\linewidth} \footnotesize
        \vspace{2mm}
    Notes: This table displays summary statistics for the different samples used in the paper. The unit of our panels is the worker-firm-month. In Panel (A), we display figures about the workers present in our data. In Panel (B), we display figures about the firms in our data. In column I, we display figures corresponding to all firms. In column II, we display figures for the data in column I restricting to (workers of) firms that employed between 6 and 13 workers in June 2009. In column III, we display figures for the data in column I restricting to (workers of) firms that employed between 6 and 9 workers in June 2009. In column IV, we display figures for the data in column II after further restricting to (workers of) firms that are in our data in every month between January 2005 and December 2013. Column II is the sample we use in the majority of our analysis.
    \end{minipage}
\end{table}

\subsection{Identification Strategy: Mapping Theory into Empirics }\label{mapping theory section}
Based on our theoretical findings, we investigate the effect of EPSW on gender segregation within firms and on the differential pay of men and women. We make several conservative choices, which we describe presently, in an attempt to provide a plausible, causal lower bound on the effect of EPSW on our outcomes of interest.

To obtain the causal effect of EPSW on our outcomes of interest, we consider an event-study analysis based on firm size \emph{at the time of policy announcement}; firms are considered ``treated'' if they were subject to EPSW at announcement (i.e. the firm in question employed at least 10 long-term workers in June 2009) and ``control'' otherwise.\footnote{Following our analysis in \Cref{section:search-model}, the effects of EPSW on a firm are predicted to 
be smaller in magnitude for untreated firms; they are unconstrained by EPSW, obviating the benefits of segregation, and thus also avoiding changes in wages for  workers who were hired prior to EPSW. We discuss this point further in \Cref{app_treatment_designation}.} This yields an ``intent-to-treat'' framework as we do not adjust our designation of firm treatment status based on firm choices (e.g. to fire workers to escape the bite of EPSW).

There are two standard assumptions for our approach to yield causal estimates of EPSW. First, we assume the policy is not anticipated prior to its announcement. Since we use the time of policy announcement to designate treatment status, 
``no anticipation'' is plausible in our context. However, the choice to use the time of announcement to determine treatment status serves to attenuate 
our estimates because some firms in our ``treatment'' group will not be bound by EPSW and some firms in our ``control'' group will be in future periods. Therefore, this decision likely contributes to our analysis understating the effect of EPSW on increased segregation rates and wage gaps. We present and discuss supportive evidence for our designation of treatment status in \Cref{app_treatment_designation}. \Cref{app_treatment_designation} also presents ``placebo'' tests which do not find statistically or economically meaningful effects at alternative firm size thresholds, supporting that the observed effects of EPSW around the size threshold specified in the law are plausibly causal. Second, we assume that outcomes of treated versus control units follow parallel trends in the absence of EPSW. As discussed above, treatment status is plausibly as good as random, and we show balance between treatment and control firms across several observable dimensions in \Cref{tab:stats_samples} which makes the parallel trends assumption plausible in our context. In our analysis below, we discuss a partial validation of this assumption by showing a lack of pre-trends between our control and treatment groups.

Our theoretical analyses predict an increase in segregation for firms that are treated by EPSW. Importantly, our theoretical predictions hold within sets of ``similar'' workers, which leaves open the possibility that a firm could be segregated within, but not across, roles. For example, a firm could be segregated for the purposes of EPSW if all of its custodians are men and all of its lawyers are women. We do not observe which workers perform ``similar'' work within a firm. Therefore, we define our segregation measure as complete firm-level gender segregation, i.e. only workers of one gender are employed by the firm. This is intuitively the strongest measure of segregation, implying the gender segregation of every set of similar workers within the firm, but it potentially misses firms that are gender segregated by, but not across, roles. Therefore, the causal effect of EPSW on segregation that we present is plausibly a lower bound. In \Cref{alt_comparison} we discuss how wage effects in  non-fully-segregated firms are consistent with the presence of segregation by role.

We also predict EPSW causes a divergence in wages across groups, depending on features of the local labor market (LLM) in question. An important question is, therefore, how we define LLMs. In our dataset, there are 21 industry codes and 16 geographic regions, which are further disaggregated into 321 counties. We define the LLM a firm operates in by the firm's geographic county and industry code pair. We augment our upcoming analysis by allowing wages to trend differently for workers in different LLMs, capturing aggregate changes in productivity for workers in different sectors of the economy.\footnote{As we discuss in the following pages and in  \Cref{alt samples}, our findings are qualitatively similar under alternative time trends, corresponding to more aggregated designations of labor markets.} 

\Cref{two-sided-EPL-equilibrium} finds that EPSW increases the wage gap in favor of the majority group, and that it increases the log wage gap in favor of the   ``more productive'' group (recall that group $A$ is collectively more productive if $\beta \mathbb{E}_A[v]>\mathbb{E}_B[v]$). We identify the majority gender in a LLM at the time of policy announcement, and we find evidence that the condition $\beta \mathbb{E}_A[v]>\mathbb{E}_B[v]$ holds across LLMs. Specifically, and as we discuss at the end of \Cref{sec:static model}, our upcoming regression analysis predicts the (residual) log wage gap pre-policy---which equals $\mathbb{E}_A[v]/\mathbb{E}_B[v]$, see Observation \ref{observation_pre_period}---is greater than one across male-majority LLMs, and is 0.994 across female-majority LLMs. The desired condition is satisfied in male-majority LLMs, as  $\beta>1$ by virtue of there being more male than female workers, and it is also satisfied in female-majority LLMs where $\beta=1.31$.  In our search model, \Cref{existence-result} finds that EPSW increases the log wage gap in favor of the group for which more firms have segregated. After the introduction of EPSW, 94\% of LLM-by-month units feature more firms segregating toward the majority group in the LLM. Taken together, these figures imply that the conditions of our two theoretical models are met and are consistent with one another: they predict the wage gap moves in favor of the majority group of workers within a LLM after the onset of EPSW.

\subsection{Effect of EPSW on Segregation}

To study the effect of EPSW on gender segregation within the firm, we consider a panel in which an observation is a firm-month. We let $j$ index a firm and let $t$ index a month. We construct a full segregation indicator, $full_{jt}$, that equals 1 in time $t$ if all workers employed by firm $j$ are of the same gender at time $t$, and 0 otherwise. We also construct an indicator for whether the firm employs at least 10 long-term workers in June 2009 (policy announcement date), which we call $above10_j$. Finally, we construct a post-treatment indicator $post_t$ that is 1 if the time $t$ of the observation is from June 2009 or later, and zero otherwise. We estimate difference-in-difference models of the following form:
\begin{align}
full_{jt} = \alpha_j + \alpha_{k(j)t} + \beta^{seg} (above10_j \times post_t) + X_{jt}\Lambda + \epsilon_{jt} \label{emp_eq1}
\end{align}
\noindent where $X_{jt}$ is a vector of controls indicating the share of workers (strictly) younger than the median age in the industry-region, the share of workers with tertiary education, the share of workers that have long-term contracts, and the total wages paid by the firm. $\alpha_j$ is a fixed effect for firm $j$, and  $\alpha_{k(j)t}$ are time-varying fixed effects.  These time-varying fixed effects compare the firm's segregation outcome in a given month to the average level of segregation among firms within the same ``comparison group'': in our baseline specification, time-varying fixed effects are at the level of 1) the firm's year of exit from the data panel by the firm's industry by the firm's region, and separately, 2) the firm's year of exit from the panel by the firm's county.\footnote{We include an interaction with the firm's exit year from the panel to control for potential cohort effects leading firms facing the similar labor market conditions to shut down.} 
We provide results under alternative time trends (i.e. different ``comparison groups'') in \Cref{alt_comparison}, and in \Cref{alt samples} we discuss alternative specifications and data samples.  
Our coefficient of interest is $\beta^{seg}$, and we interpret it as the effect of the policy on the share of gender-segregated firms.

To understand more about the dynamic effects of EPSW, we consider a difference-in-differences model period by period, where we define each period to be a half year (January-June, July-December). We label the period immediately prior to policy announcement as period -1 and omit it as the reference period. Therefore, the set of non-reference periods included in our analysis is $\mathcal{T}:=\{-8,...,-2,0,...,9\}$. By construction, each period (indexed by $\tau$) is composed of six months (indexed by $t$).
\begin{align}
full_{jt} = \alpha_j + \alpha_{k(j)t} + \sum_{\tau\in \mathcal{T} }^{} \beta^{seg}_\tau D_{jt} + X_{jt}\Lambda + \epsilon_{jt} \label{emp_eq2}
\end{align}
\noindent where $D_{jt}$ is an indicator that equals 1 in time period  $t$ if  firm $j$ employs at least 10 long-term workers at policy announcement, and zero otherwise. $\beta^{seg}_\tau$ is the average difference in segregation between treated and control firms in period $\tau$ (relative to period -1).

Our identifying assumption is that parallel trends hold between treated and control firms. 
That is, $\mathbb{E}[\epsilon_{jt}\cdot D_{jt}]=0$ for all $t$. This strategy builds in a partial falsification test, in that we expect coefficient estimates of $\beta^{seg}_\tau$ to be zero for all $\tau<0$.

Panel (A) of Table \ref{emp_table_seg_robust} presents estimates on the effect of EPSW on segregation. Column I presents our baseline results from  \eqref{emp_eq1}. We find a 4.41 percentage point increase in segregation following EPSW, from a baseline of 34.3\% of firms that were fully segregated at  EPSW announcement. Columns II-VI present results on segregation from alternative empirical specifications and alternative sample selections.  Across all specifications, our point estimate of $\beta^{seg}$ is positive, indicating that EPSW led to an increase in segregation; in all but one specification, this increase is statistically significant at conventional levels. 
While we cannot reject the null hypothesis that the coefficient of interest in the ``doughnut hole'' sample  is zero at conventional levels, the p-value for the test of this null  equals $0.175$; also, under the assumption that the coefficient of interest across our baseline and doughnut hole samples are independently distributed, we fail to reject the null hypothesis that the effect of EPSW on the doughnut hole sample is equal to the effect on our baseline sample at conventional significance levels.  Additionally, we present the mean of the pre-treatment estimates (i.e. for $\tau<0$) from \eqref{emp_eq2} to support our parallel trends assumption. Across specifications, these estimates  are small and statistically insignificant.  

Panel (A) of Figure \ref{emp_fig1} displays the estimated coefficients of interest from \eqref{emp_eq2}. 
Prior to EPSW, the coefficient of interest is statistically indistinguishable from 0 in all periods $\tau<0$. In $\tau=0$ segregation rises by 1.75 percentage points (SE = 0.0079) in the treated group compared to the control group and rises to 4.96 percentage points (SE=0.0185) by $\tau=9$. 

To better understand how EPSW affects firm incentives to segregate, we consider its effect on the share of firms that are mostly-but-not-fully segregated. Our model makes strong predictions that firms bound by EPSW will have incentives to fully segregate in equilibrium, but not that there are particular incentives to partially segregate. 
Indeed, \Cref{near seg prop} suggests that treated firms are unlikely to be mostly-but-not-fully segregated post EPSW. That is, firms that would have otherwise had only a small number of workers of the ``wrong'' gender that prevent full segregation may be particularly likely to end their relationship with these workers to achieve full segregation. Therefore, the economic forces present in our model likely indicate a decrease in the share of treated firms that are almost-but-not-fully segregated after EPSW.

\begin{table}[ht]
\centering
\addtolength{\leftskip} {-2cm}
    \addtolength{\rightskip}{-2cm}\caption{Effect of EPSW on Segregation}\label{emp_table_seg_robust}
{
\def\sym#1{\ifmmode^{#1}\else\(^{#1}\)\fi}
\begin{tabular}{l*{7}{c}}
\hline\hline
               &\multicolumn{1}{c}{(I)}&\multicolumn{1}{c}{(II)}&\multicolumn{1}{c}{(III)}&\multicolumn{1}{c}{(IV)}&\multicolumn{1}{c}{(V)}&\multicolumn{1}{c}{(VI)}&\\
            &\multicolumn{1}{c}{Baseline}&\multicolumn{1}{c}{No firm}&\multicolumn{1}{c}{No}&\multicolumn{1}{c}{Doughnut}&\multicolumn{1}{c}{Firm}&\multicolumn{1}{c}{Balanced} \\
            &\multicolumn{1}{c}{}&\multicolumn{1}{c}{FEs}&\multicolumn{1}{c}{controls}&\multicolumn{1}{c}{hole}&\multicolumn{1}{c}{growth}&\multicolumn{1}{c}{sample} \\

\hline
\textit{Panel (A): Dep. var. ``Full Segregation''}\\
\quad ($\hat \beta^{seg}$) Post $\times$ Treated &     $0.0441$  & $0.0438$  & $0.0385$ & $0.0233$        & $0.0394$  & $0.0349$   \\
                      & $(0.0133)$      & $(0.0136)$      & $(0.0134)$     & $(0.0171)$      & $(0.0140)$      & $(0.0161)$      \\\\
\quad Mean Pre-Treatment    &    $-0.0014$           &  $-0.0040$            & $0.0030$    &   $-0.0055$                 &    $-0.0024$        &      $-0.0110$              \\
 &  $(0.0118)$              &   $(0.0113)$     &  $(0.0118)$           &     $(0.0158)$                        & $(0.0127)$            &    $(0.0130)$            \\
\hline
\textit{Panel (B): Dep. var. ``Near Segregation''}\\
\quad  ($\hat \beta^{nearseg}$) Post $\times$ Treated                                                           & $-0.0402$ & $-0.0371$ & $-0.0350$ & $-0.0300$   & $-0.0344$  & $-0.0344$  \\
                      & $(0.0139)$      & $(0.0139)$      & $(0.0140)$     & $(0.0176)$      & $(0.0146)$      & $(0.0165)$      \\\\
                      \quad Mean Pre-Treatment    &   $0.0047$            &      $0.0048$        &   $0.0003$  &    $0.0190$                &    $0.0051$       &       $0.0106$             \\
 &      $(0.0132)$          &  $(0.0123)$      &   $(0.0130)$           &      $(0.0175)$                  &    $(0.0138)$          &    $(0.0152)$             \\
\hline
Number of Firms                 &     5312           &    5312             &   5312           &   4660          &   4578         &     3311        \\
Number of Observations          &     $515361$          &    $515363$  & $515361$         &    $452638$                   &   $457649$         &      $357588$              \\
\hline\hline
\textit{Fixed effects}                     &       &   &   &   &  &          \\
\quad Firm                              & Yes & No  & Yes  & Yes & Yes & Yes  \\
\quad Month$\times$exit year$\times$county              & Yes  & Yes & Yes & Yes & Yes & Yes \\
\quad Month$\times$exit year$\times$region$\times$industry              & Yes  & Yes & Yes & Yes & Yes & Yes \\
\quad Firm-month level controls            & Yes  & Yes  & No & Yes & Yes  & Yes \\
\hline\hline
\end{tabular}
}
 \begin{minipage}{1.31\textwidth} \footnotesize
        \vspace{2mm}
       Notes:  The unit of analysis is the firm-month. In panel (A) the dependent variable is a binary variable that equals 1 if and only if all workers at the firm in question are of a single gender in a given month. In panel (B) the dependent variable is a binary variable that equals 1 if and only if the share of the majority gender of workers at firm $j$ at time $t$ is an element of $[.8,1).$ Across panels, column I presents estimated coefficients  $\hat \beta^{seg}$ and $\hat \beta^{nearseg},$ respectively, for our baseline difference-in-differences regression presented in \eqref{emp_eq1}. Firm-month controls included are: the share of workers younger than the median age in the industry-region, the share of workers with tertiary education, the share of workers that have long-term contracts,  and the total wages paid by the firm. The mean pre-treatment effect is the mean of $\hat \beta^{seg}_\tau$ ($\hat \beta^{nearseg}_\tau$) for $\tau<0$ calculated from \eqref{emp_eq2}. Columns II-VI present the analogous information for alternative sample selections and empirical specifications as described in \Cref{alt samples}: column II removes firm fixed effects, column III removes the vector of controls, column IV drops firms with 9 or 10 workers at announcement from our baseline sample, column V interacts time-varying fixed effects with quartiles of firms growth in the 6 months prior to policy announcement, and column VI drops firms that are not present in every month of our panel. Throughout, standard errors in parentheses are two-way clustered at the firm and month levels.
    \end{minipage}
\end{table}

We reanalyze our firm-based analysis with a different dependent variable: ``near'' segregation. Specifically, we define a firm $j$ to be nearly segregated at time $t$ if the share of workers in the majority gender of its workforce is in the interval $[0.8,1).$ Note that $j$ is classified as ``nearly'' segregated at time $t$ only if it is not fully segregated at time $t$. We select 0.8 as the lower end of the interval for the definition of this outcome variable due to our size restrictions; firms in our sample typically are nearly segregated only in time periods in which they employ either 1 or 2 workers of the non-majority gender; however, our findings are robust to other selections of the lower end of the range. We re-estimate  \eqref{emp_eq1} and \eqref{emp_eq2} with near segregation as the outcome variable and present these results in Panel (B) of Table  \ref{emp_table_seg_robust} and Panel (B) of Figure \ref{emp_fig1.5}, respectively. We refer to the associated coefficients of interest as $\beta^{near seg}$ and $\beta^{near seg}_\tau$, respectively.

The findings presented in Panel (B) of Table \ref{emp_table_seg_robust} support our hypothesis on the effects of EPSW on near segregation. Each column presents a specification corresponding to the same column in Panel (A). Our baseline specification in column I reveals that EPSW lowers the share of nearly segregated firms 4.02 percentage points following EPSW. 
Across all columns, these estimates are similar in magnitude (and opposite in sign) to the increase in the share of fully segregated, treated firms following EPSW (see Panel (A)). All estimates are statistically significant at conventional levels. This ``missing mass'' of firms with near-but-not-full segregation suggests that firms that face relatively lower costs of fully segregating are the most likely to do so within our window of analysis. Panel (B) of Figure \ref{emp_fig1.5} shows the dynamic path of near segregation following EPSW. The pre-treatment coefficients are not statistically different from 0 in any period $\tau< 0$.

\begin{figure}[ht]\caption{Dynamic Path of EPSW's Effect on Segregation \label{emp_fig1} \label{emp_fig1.5}}    
\makebox[\textwidth][c]{
\vspace{3mm}
\centering
\begin{minipage}[c][1\totalheight][t]{0.7\textwidth}\center{\scriptsize{(A): Full Segregation}}
\begin{center}
\includegraphics[width=0.9\textwidth]{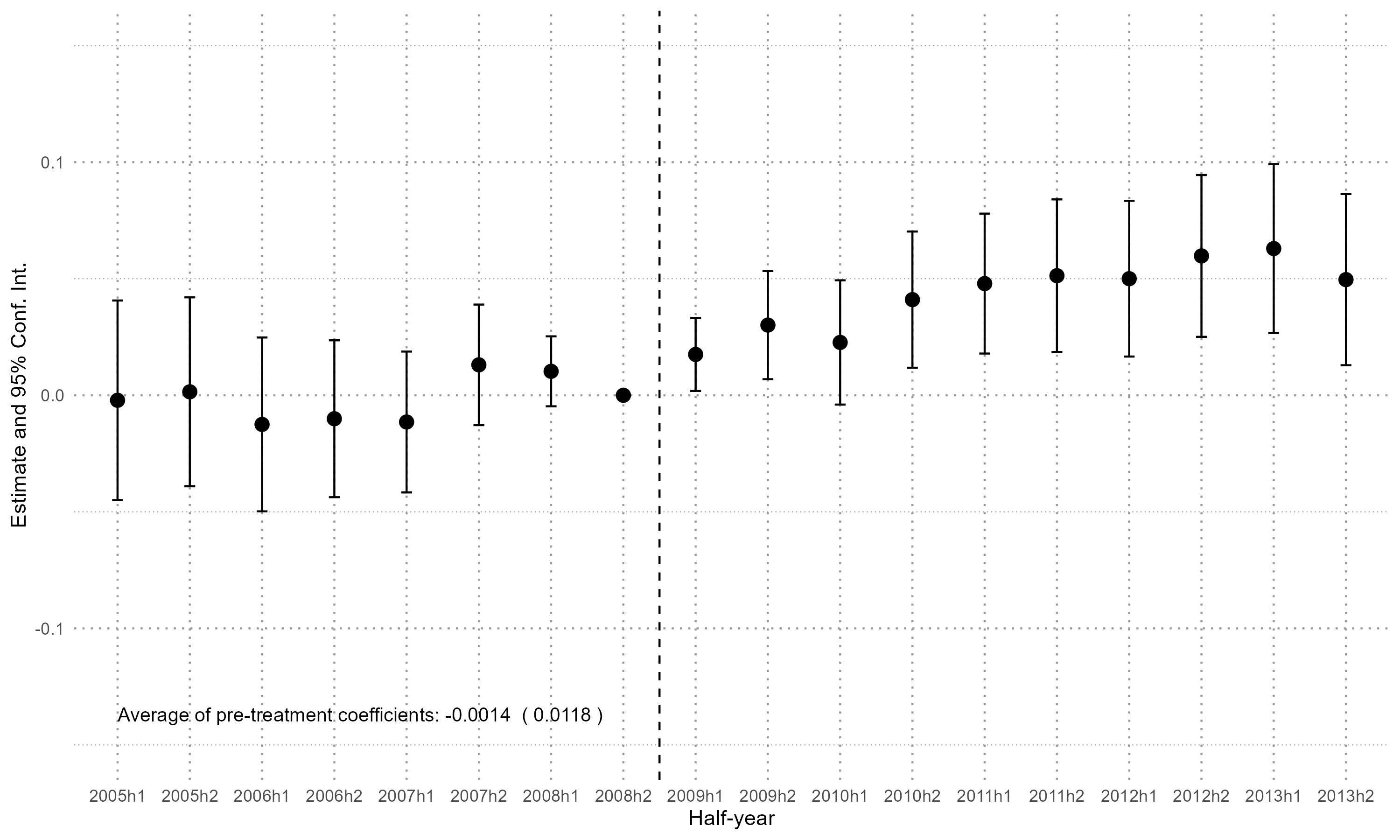}
\end{center}
\end{minipage}\hspace{-1cm}
\begin{minipage}[c][1\totalheight][t]{0.7\textwidth}\center{\scriptsize{(B): Near Segregation}}
\begin{center}
\includegraphics[width=0.9\textwidth]{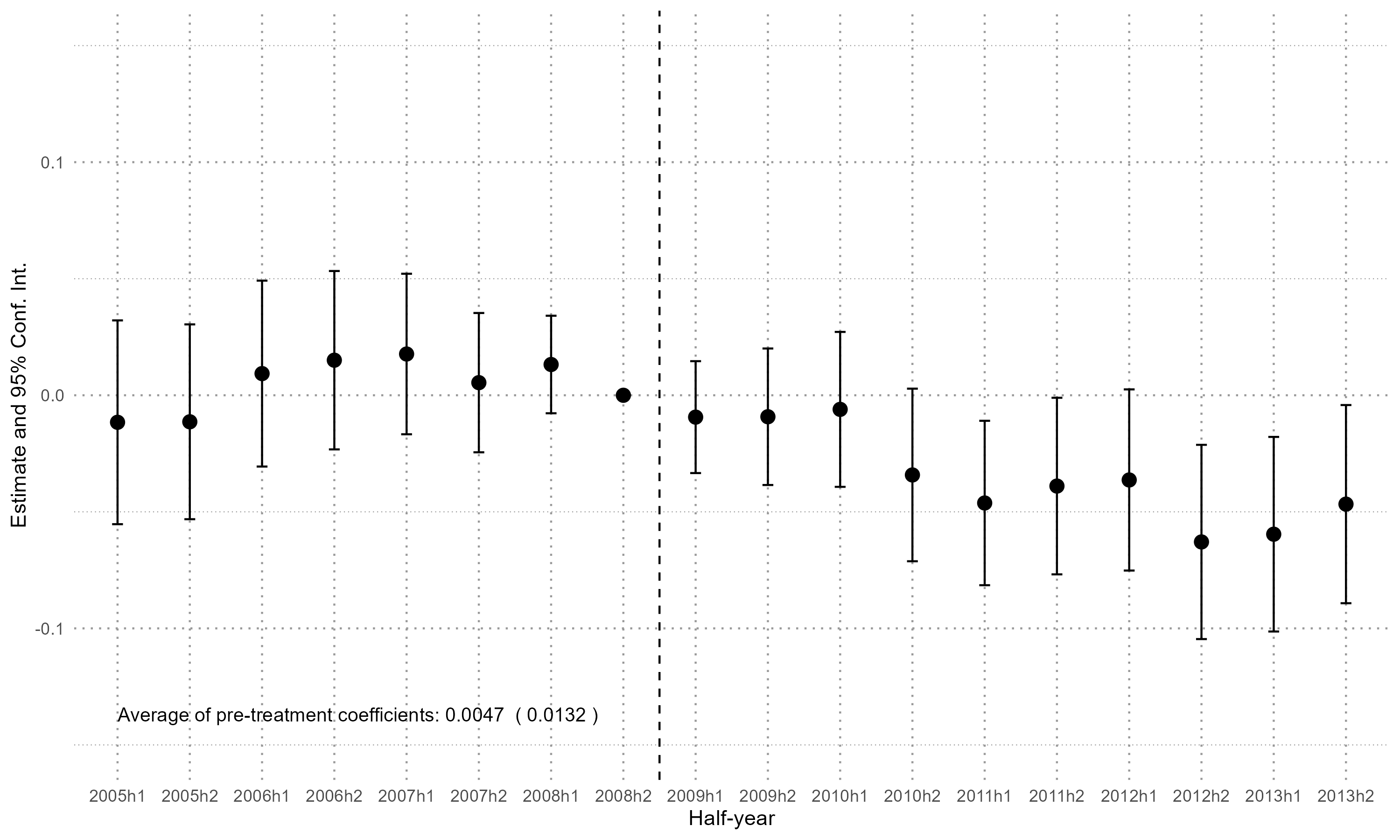}
\end{center}
\end{minipage}}
    \begin{minipage}{.999\textwidth} \footnotesize
        \vspace{2mm}
       Notes: This figure displays the estimated coefficients $\hat \beta_\tau$ for the difference-in-differences regression described in \eqref{emp_eq2}. In Panel (A) the outcome variable is ``full'' segregation. In Panel (B) the outcome variable is ``near'' segregation, i.e. the dependent variable equals 1 if and only if the share of the majority gender of workers at firm $j$ at time $t$ is an element of $[.8,1).$ In both panels, bars depict 95\% confidence intervals.
    \end{minipage}
\end{figure}

\Cref{near seg prop} predicts that firms that are initially nearly segregated are the most likely to segregate within our window of analysis. While near-segregation at announcement is not randomly assigned, we can test to see whether this predicted association is present in our data. To do so, we add an additional independent variable (and its relevant interactions) to \eqref{emp_eq1} which captures whether a firm is nearly segregated at announcement. Specifically, we construct an indicator, $NSA_{j}$, that equals 1 if the share of workers in the majority gender of firm $j$'s workforce \emph{at announcement} is in the interval $[0.8,1),$ and zero otherwise.  We estimate difference models of the following form:
\begin{align}
    full_{jt} = &~  \alpha_j + \alpha_{k(j)t} + \gamma_1 (NSA_j \times post_t) \notag\\
    & + \beta^{seg} (above10_j \times post_t) + \beta^{NSA} (above10_j \times post_t \times NSA_j) + X_{jt}\Lambda + \epsilon_{jt} \label{emp_eqNSA} 
\end{align}
\noindent Our coefficients of interest are $\beta^{seg}$ and $\beta^{NSA}$. We interpret $\beta^{seg}$ as the association between the policy and the share of gender-segregated firms among those that are not nearly segregated at announcement, and $\beta^{seg} + \beta^{NSA}$ as the association between the policy and the share of gender-segregated firms among those that are nearly segregated at announcement.

Table \ref{emp_table_seg_by_near_seg} presents estimates on the association between EPSW and segregation by near segregation status at announcement. Column I presents our baseline results from  \eqref{emp_eqNSA}. We find a 2.78 percentage point increase in segregation following EPSW for firms that are not nearly segregated at announcement, and a 8.71 percentage point increase in segregation for firms that are nearly segregated at announcement. Columns II-VI present results from alternative empirical specifications and alternative sample selections. The same pattern holds across all specifications: the estimated coefficient for firms that are nearly segregated at announcement is roughly three times larger than the estimated coefficient for firms that are not nearly segregated at announcement, suggesting that firms that are nearly segregated prior to EPSW are the most likely to segregate post EPSW.

\begin{table}[ht]
\centering
\addtolength{\leftskip} {-2cm}
    \addtolength{\rightskip}{-2cm}\caption{Effect of EPSW on Segregation, by Nearly Segregated at Announcement}\label{emp_table_seg_by_near_seg}
{
\def\sym#1{\ifmmode^{#1}\else\(^{#1}\)\fi}
\begin{tabular}{l*{7}{c}}
\hline\hline
               &\multicolumn{1}{c}{(I)}&\multicolumn{1}{c}{(II)}&\multicolumn{1}{c}{(III)}&\multicolumn{1}{c}{(IV)}&\multicolumn{1}{c}{(V)}&\multicolumn{1}{c}{(VI)}&\\
            &\multicolumn{1}{c}{Baseline}&\multicolumn{1}{c}{No firm}&\multicolumn{1}{c}{No}&\multicolumn{1}{c}{Doughnut}&\multicolumn{1}{c}{Firm}&\multicolumn{1}{c}{Balanced} \\
            &\multicolumn{1}{c}{}&\multicolumn{1}{c}{FEs}&\multicolumn{1}{c}{controls}&\multicolumn{1}{c}{hole}&\multicolumn{1}{c}{growth}&\multicolumn{1}{c}{sample} \\

\hline
($\hat \beta^{seg}$) Post $\times$ Treated                                                           & $0.0278$      & $0.0330$      & $0.0212$      & $0.0186$      & $0.0268$      & $0.0203$      \\
                                                                               & $(0.0144)$    & $(0.0148)$    & $(0.0145)$    & $(0.0181)$    & $(0.0152)$    & $(0.0175)$      \\
                      ($\hat \beta^{seg}+\hat\beta^{NSA}$) Effect for Firms                                                   & $0.0871$      & $0.0810$      & $0.0836$      & $0.0418$      & $0.0711$      & $0.0782$      \\
           \quad Nearly Segregated at Announcement             & $(0.0267)$    & $(0.0269)$    & $(0.0271)$    & $(0.0355)$    & $(0.0278)$    & $(0.0320)$    \\
 \hline
Number of Firms                 &     5312           &    5312             &   5312           &   4660          &   4578          &     3311        \\
Number of Observations          &     $515361$          &    $515363$  & $515361$         &    $452638$                   &   $457649$         &      $357588$              \\
\hline\hline
\textit{Fixed effects}                     &       &   &   &   &  &          \\
\quad Firm                              & Yes & No  & Yes  & Yes & Yes & Yes  \\
\quad Month$\times$exit year$\times$county              & Yes  & Yes & Yes & Yes & Yes & Yes \\
\quad Month$\times$exit year$\times$region$\times$industry              & Yes  & Yes & Yes & Yes & Yes & Yes \\
\quad Firm-month level controls            & Yes  & Yes  & No & Yes & Yes  & Yes \\
\hline\hline
\end{tabular}
}
 \begin{minipage}{1.23\textwidth} \footnotesize
        \vspace{2mm}
       Notes: The unit of analysis is the firm-month. The dependent variable is a binary variable that equals 1 if and only if all workers at the firm in question are of a single gender in a given month. Column I presents estimated coefficients of interest for our baseline difference regression presented in \eqref{emp_eqNSA}. Firm-month controls included are: the share of workers younger than the median age in the industry-region, the share of workers with tertiary education, the share of workers that have long-term contracts,  and the total wages paid by the firm.  Columns II-VI present the analogous information for alternative sample selections and empirical specifications as described in \Cref{alt samples}: column II removes firm fixed effects, column III removes the vector of controls, column IV drops firms with 9 or 10 workers at announcement from our baseline sample, column V interacts time-varying fixed effects with quartiles of firms growth in the 6 months prior to policy announcement, and column VI drops firms that are not present in every month of our panel. Throughout, standard errors in parentheses are two-way clustered at the firm and month levels. 
    \end{minipage}
\end{table}

\subsection{Effect of EPSW on the Gender Wage Gap}\label{wage_gap_section}

Guided by our theory, we are interested in studying the effect of EPSW on the wage gap between male and female workers. Our theory predicts a relative benefit to men in male majority LLMs and a relative benefit to women in female majority LLMs. We define a LLM as being female majority if the share of female workers across all firms in a particular industry-county pair is greater than 0.5 in June 2009, and otherwise, we define it as a male majority LLM.

To estimate the effects of EPSW on the gender wage gap across LLMs, we consider a panel in which an observation is a worker-firm-month. In order to remove potential confounds from workers who are simultaneously enrolled in formal schooling and those beyond Chile's official retirement age, we only consider workers aged 22-65. We let $i$ index a worker, $j$ index a firm, and $t$ index a month. Let $w_{ijt}$ be the earnings of worker $i$ at firm $j$ in month $t$. We construct an indicator, $male_{i}$, that equals 1 if the worker is a male, and 0 otherwise. We construct indicator $femalemaj_{j}$ that equals 1 if firm $j$ is in a female majority LLM, and 0 otherwise. We estimate difference models of the following form: 
\begin{align}
    \ln w_{ijt} = & ~\alpha_i + \alpha_j + \alpha_{k(ij)t} + \gamma_{1} (above10_j \times post_t) + \psi_{1} (above10_j \times post_t \times femalemaj_{j}) \notag\\
    & + \gamma_2 (male_i \times post_t) + \psi_2 (male_i \times post_t \times femalemaj_{j}) \notag\\
    & + \gamma_3 (above10_j \times male_i) + \psi_3 (above10_j \times male_i \times femalemaj_{j}) \notag\\
    & + \beta^{Mgap}(above10_j \times male_i \times post_t) + \beta^{Fgap}(above10_j \times male_i \times post_t \times femalemaj_{j}) \notag \\
    &+ X_{ijt}\Lambda + \epsilon_{ijt} \label{emp_eq5}
\end{align}

\noindent where $\alpha_i$ is a fixed effect for worker $i$ and  $X_{ijt}$ is a vector of firm-month level controls and worker-firm-month level controls. The firm-month level controls are the share of workers younger than the median age at the industry-region, the share of workers with tertiary education, and the share of workers that have long-term contracts. The worker-firm-month level controls are the number of months at the firm and an indicator for whether the worker's  earnings reach the top-coding threshold. 
$\alpha_{k(ij)t}$ are time fixed effects for workers in set $k(ij)$, where $k(ij)$ is a comparison group of workers to $i$ employed at a comparison group of firms to $j$. Worker comparison groups are defined by equivalence across three binary ``human capital'' dimensions at time $t$ at firm $j$: an indicator for the highest level of educational attainment (high school, tertiary, or neither), an indicator for long-term versus fixed-term contract, and an indicator for the worker's age in decades (i.e. 20-29 years old, 30-39 years old, etc). Firm comparison groups are defined as in \eqref{emp_eq1}. 
Because workers can switch from one LLM to another, the inclusion of these fixed effects controls for composition changes in the distribution of observable characteristics within a LLM caused by EPSW. We provide results controlling for alternative time trends at more aggregated levels in \Cref{alt_comparison}; our estimates are qualitatively and quantitatively similar which suggests that worker substitution across LLMs is not driving our findings.  In \Cref{alt samples} we discuss alternative specifications and data samples.

 Our coefficients of interest are $\beta^{Mgap}$ and $\beta^{Fgap}$. We interpret $\beta^{Mgap}$ as the effect of the policy on the  (percentage) wage gap between male and female workers in male majority LLMs. We interpret $\beta^{Fgap} + \beta^{Mgap}$ as the effect of the policy on the (percentage) wage gap between male and female workers in female LLMs.

To understand more about the dynamic effects of EPSW, we consider a difference model period by period, where we define each period to be a half year (January-June, July-December). We label the period immediately prior to policy announcement as period -1 and omit it as the reference period. Therefore, the set of non-reference periods included in our analysis is $\mathcal{T}:=\{-8,...,-2,0,...,9\}$ By construction, each period (indexed by $\tau$) is composed of six months (indexed by $t$). Let $period_\tau$ be an indicator that equals 1 in period $\tau$, and zero otherwise. 
 \begin{align}
    \ln w_{ijt} = & ~\alpha_i + \alpha_j + \alpha_{k(ij)t} +  \sum_{\tau\in\mathcal{T}}^{} \gamma_{1\tau} (above10_j \times period_\tau) + \sum_{\tau\in\mathcal{T}}^{} \psi_{1\tau} (above10_j \times femalemaj_{j} \times period_\tau)  \notag\\
     & + \sum_{\tau\in\mathcal{T}}^{} \gamma_{2\tau} (male_i \times period_\tau)  + \sum_{\tau\in\mathcal{T}}^{} \psi_{2\tau} (male_i \times femalemaj_{j} \times period_\tau) \notag\\
     & + \gamma_3 (above10_j \times male_i) + \psi_3 (above10_j \times male_i \times femalemaj_{j}) \notag\\
     & + \sum_{\tau\in\mathcal{T}}^{} \beta^{Mgap}_{\tau} D^M_{ijt}  +  \sum_{\tau\in\mathcal{T}}^{} \beta^{Fgap}_{\tau} D^F_{ijt}  \notag\\
     & + X_{ijt}\Lambda + \epsilon_{ijt} \label{emp_eq6}
 \end{align}

 \noindent where $D^M_{ijt}$ is an indicator that equals 1 in time $t$ if firm $j$ employs at least 10 long-term workers at the time of policy announcement, $j$'s LLM is coded as male majority, and $i$ is male, and zero otherwise. $\beta^{Mgap}_\tau$ is the average difference in log wages between men and women in treated versus control firms  in period $\tau$ (relative to period -1)  in male majority LLMs. Similarly, $D^F_{ijt}$ is an indicator that equals 1 in time  $t$ if firm $j$ employs at least 10 long-term workers at the time of policy announcement, $j$'s LLM is coded as female majority, and $i$ is male, and zero otherwise. $\beta^{Fgap}_\tau$ is the average difference in  log wages between men and women in treated versus control firms in period $\tau$ (relative to period -1) in female majority LLMs \emph{relative to male majority LLMs}. Therefore, $\beta^{Mgap}_\tau+\beta^{Fgap}_\tau$ is the average difference in log wages between men and women in treated versus control firms in period $\tau$ (relative to period -1) in female majority LLMs.

Our identifying assumption is that parallel trends hold between treated and control firms, that is, $\mathbb{E}[\epsilon_{ijt}\cdot D^M_{ijt}]=0$ and $\mathbb{E}[\epsilon_{ijt}\cdot (D^M_{ijt}+D^F_{ijt})]=0$ for all $t$ \citep{olden}. This strategy builds in a partial falsification test, in that we expect coefficient estimates of $\beta^{Mgap}_\tau$, and $\beta^{Mgap}_\tau+\beta^{Fgap}_\tau$ to be zero for $\tau<0$.

Table \ref{emp_tableYYYY} presents our estimates on the effect of EPSW on the gender wage gap. Column I presents our baseline results from   \eqref{emp_eq5}. We find that EPSW increases the gender wage gap (in favor of men) by 4.27 percentage points in male majority LLMs, but decreases the gender wage gap (in favor of women) by 6.24 percentage points in female majority LLMs. For reference, men typically outearn women prior to EPSW: the wage gap (controlling for human capital factors, as in  \eqref{emp_eq5}) is 8.58\% (SE=0.0060) in male-majority LLMs and 0.65\% (SE=0.0119) in female-majority LLMs. Columns II-VII present results on the gender wage gap from alternative empirical specifications and alternative sample selections. We discuss these specifications further in \Cref{alt samples}. Across all specifications, our findings of an increase in the gender wage gap in male majority LLMs, and a decrease in the gender wage gap in female majority LLMs, are statistically significant at conventional levels.
Additionally, we present the mean of the pre-treatment estimates (i.e. for
$\tau<0$) from \eqref{emp_eq6} to support our parallel trends assumption. Across specifications, these estimates
are small and statistically insignificant.

\begin{table}[ht]\centering
\addtolength{\leftskip} {-2.5cm}
    \addtolength{\rightskip}{-2.5cm}
    \caption{Effect of EPSW on Gender Wage Gap, by Majority Worker Group}\label{emp_tableYYYY}
{
\def\sym#1{\ifmmode^{#1}\else\(^{#1}\)\fi}
\centerline{\begin{tabular}{l*{6}{c}}
\hline\hline
              &\multicolumn{1}{c}{(I)}&\multicolumn{1}{c}{(II)}&\multicolumn{1}{c}{(III)}&\multicolumn{1}{c}{(IV)}&\multicolumn{1}{c}{(V)}&\multicolumn{1}{c}{(VI)}\\
            &\multicolumn{1}{c}{Baseline}&\multicolumn{1}{c}{No firm}&\multicolumn{1}{c}{No}&\multicolumn{1}{c}{Doughnut}&\multicolumn{1}{c}{Firm}&\multicolumn{1}{c}{Balanced} \\
            &\multicolumn{1}{c}{}&\multicolumn{1}{c}{FEs}&\multicolumn{1}{c}{controls}&\multicolumn{1}{c}{hole}&\multicolumn{1}{c}{growth}&\multicolumn{1}{c}{sample}\\
\hline
($\hat \beta^{Mgap}$) Treated $\times$ Male $\times$ Post            & $0.0427$       & $0.0364$       & $0.0432$       & $0.0409$       & $0.0439$       & $0.0407$     \\
                                             & $(0.0116)$     & $(0.0114)$     & $(0.0116)$     & $(0.0158)$     & $(0.0117)$     & $(0.0144)$\\

($\hat \beta^{Mgap}+\hat \beta^{Fgap}$) Effect in Female                    &   $-0.0624$      & $-0.0475$      & $-0.0575$      & $-0.0840$      & $-0.0622$      & $-0.0511$              \\
\quad Majority LLM  &     $(0.0197)$     & $(0.0206)$     & $(0.0196)$     & $(0.0251)$     & $(0.0230)$     & $(0.0234)$           \\\\

Mean Pre-Treatment                           &   $-0.0155$          &   $-0.0134$         &$-0.0136$  &   $-0.0082$                   &   $-0.0159$        &    $-0.0101$        \\
\quad (Male Majority LLM)                        &  $(0.0163)$                &     $(0.0155)$      & $(0.0166)$  &    $(0.0229)$                &   $(0.0158)$           &    $(0.0201)$           \\

Mean Pre-Treatment                    &    $-0.0012$          &  $-0.0167$          & $-0.0072$  &    $0.0343$                 &  $0.0102$        &  $0.0335$            \\
\quad (Female Majority LLM)                & $(0.0269)$             &  $(0.0332)$           & $(0.0268)$  &   $(0.0319)$               &     $(0.0324)$        & $(0.0284)$                 \\

\hline
Number of Firms                                        &   $ 6424$            &        $ 6424$        & $ 6424$ &  $ 5692$                    &       $ 6153$          &       $ 3351$               \\
Number of Observations                                 &    $5551100$        &     $5551100$       & $5551100$  & $4800898$                &     $5108499$      &     $3563150$           \\
\hline\hline
\textit{Fixed effects}                       &       &   &   &   &  &        \\
\quad Firm                                & Yes & No & Yes & Yes & Yes & Yes  \\
\quad Worker                                 & Yes& Yes & Yes & Yes & Yes & Yes \\  

\quad Month$\times$exit year$\times$county$\times$hum. cap.              & Yes  & Yes & Yes & Yes & Yes & Yes \\
\quad Month$\times$exit year$\times$region$\times$industry$\times$hum. cap.              & Yes  & Yes & Yes & Yes & Yes & Yes \\

\quad Firm-month level controls              & Yes & Yes & No & Yes & Yes & Yes\\
\quad Worker-firm-month level controls            & Yes & Yes & No & Yes & Yes & Yes  \\
\hline\hline
\end{tabular}}
}
 \begin{minipage}{1.38\textwidth} \footnotesize
        \vspace{2mm}
       Notes:   The unit of analysis is the worker-firm-month and the dependent variable is the natural logarithm of the worker's wage at the firm in a given month.
       Column I presents estimated coefficients  $\hat \beta^{Mgap}$ and $\hat \beta^{Fgap}$  for our baseline regression specification presented in \eqref{emp_eq5}. 
       ``Human capital'' fixed effects are comprised of an indicator for the highest level of educational attainment (high school, tertiary, or neither), an indicator for long-term versus fixed-term contract, and an indicator for the worker's age in decades (i.e. 20-29 years old, 30-39 years old, etc).
       Firm-month controls included are:  the share of workers younger than the median age in the industry-region, the share of workers with tertiary education, and the share of workers that have long-term contracts. Worker-firm-month level controls included are the number of months at the firm,  
 and an indicator for reaching the earnings truncation threshold. 
       The mean pre-treatment effects are the mean of $\hat \beta^{Mgap}_\tau$ and $\hat \beta^{Mgap}_\tau+\hat \beta^{Fgap}_\tau$, respectively, for $\tau<0$ calculated from \eqref{emp_eq6}.
     Columns II-VI present the analogous information for alternative sample selections and empirical specifications as described in \Cref{alt samples}: column II removes firm fixed effects, column III removes the vector of controls, column IV drops firms with 9 or 10 workers at announcement from our baseline sample, column V interacts time-varying fixed effects with quartiles of firms growth in the 6 months prior to policy announcement, and column VI drops firms that are not present in every month of our panel. Throughout, standard errors in parentheses are two-way clustered at the firm and month levels. 
    \end{minipage}
\end{table}

Moreover, our empirical findings match our theoretical predictions that the change in the wage gap caused by EPSW is primarily driven by a reduction in average wages of the minority group of workers in a LLM. Specifically, women's average wages in male-majority labor markets fall post EPSW by 3.29 percentage points (SE= 0.0111) in treated versus control firms, and men's average wages in female-majority labor markets fall post EPSW by 4.45 percentage points (SE=0.0178) in treated versus control firms.\footnote{These values are estimates, respectively, of $\gamma_1$ and $\gamma_1+\psi_1+\beta^{Mgap}+\beta^{Fgap}$ in \eqref{emp_eq5}.} There are no statistically significant changes in the average wage of the majority group of workers within a LLM.

Figure \ref{emp_figXX} displays the estimated coefficients of interest from \eqref{emp_eq6}. Panel (A) presents estimates for male majority LLMs. The coefficient of interest is statistically indistinguishable from 0 in all periods $\tau<0$. In $\tau=0$ the wage gap rises (in favor of men) by 1.38 percentage points (SE = 0.0123) in the treated group compared to the control group, and rises by 3.76 percentage points (SE = 0.0187) by $\tau=9$. Panel (B) presents estimates for female majority LLMs. The coefficient of interest is statistically indistinguishable from 0 in all periods $\tau<0$. In $\tau=0$, the wage gap falls (in favor of women) by 2.35 percentage points (SE=0.0210) in the treated group compared to the control group,  and falls by 10.30 percentage points (SE = 0.0313) by $\tau=9$.

\begin{figure}[ht]\caption{Dynamic Path of EPSW's Effect on Gender Wage Gap, by Majority Worker Group \label{emp_figXX}}
\makebox[\textwidth][c]{
\vspace{3mm}
\centering

\begin{minipage}[c][1\totalheight][t]{0.7\textwidth}\center{\scriptsize{(A): Male Majority Local Labor Markets}}
\begin{center}
\includegraphics[width=0.9\textwidth]{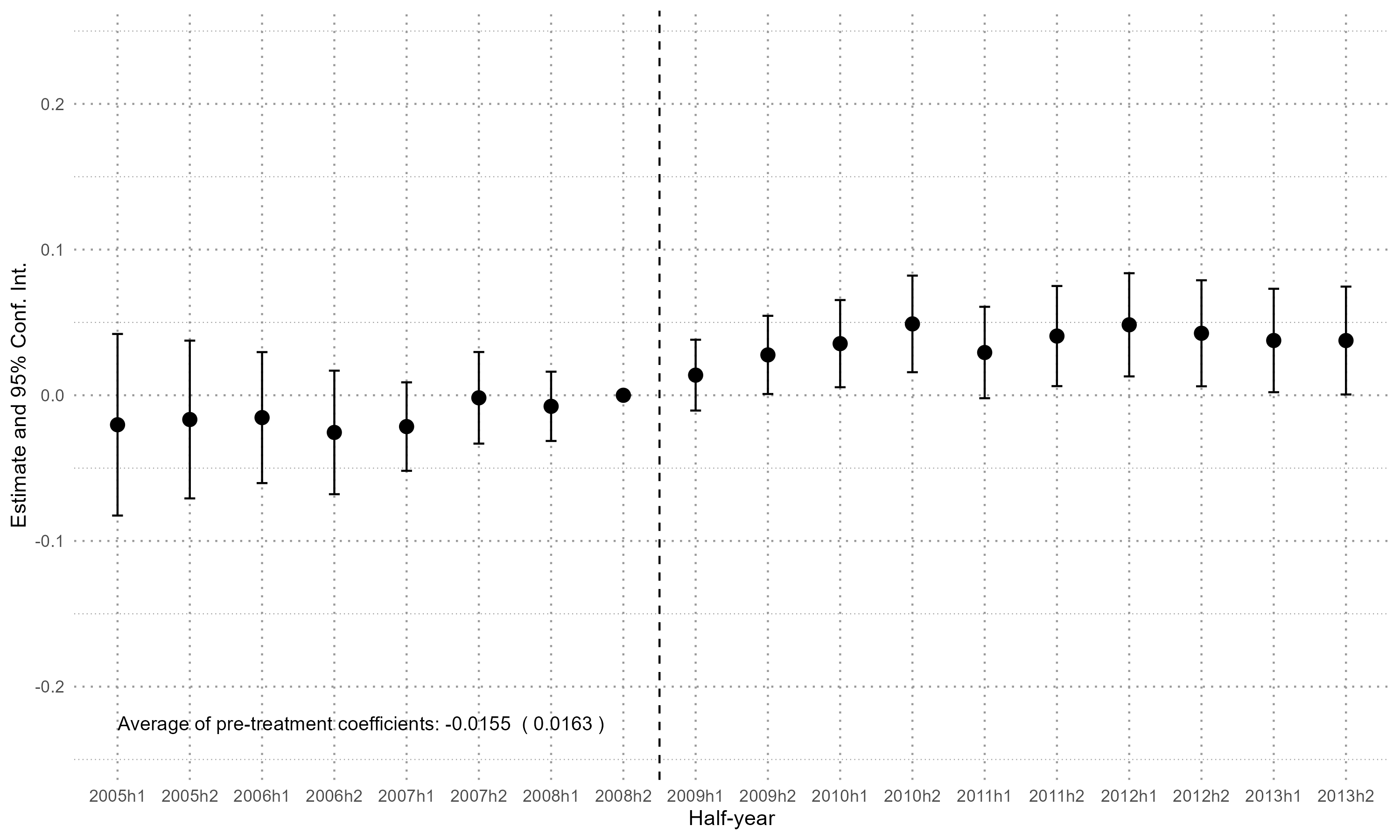}
\end{center}
\end{minipage}\hspace{-1cm}
\begin{minipage}[c][1\totalheight][t]{0.7\textwidth}\center{\scriptsize{(B): Female Majority Local Labor Markets}}
\begin{center}
\includegraphics[width=0.9\textwidth]{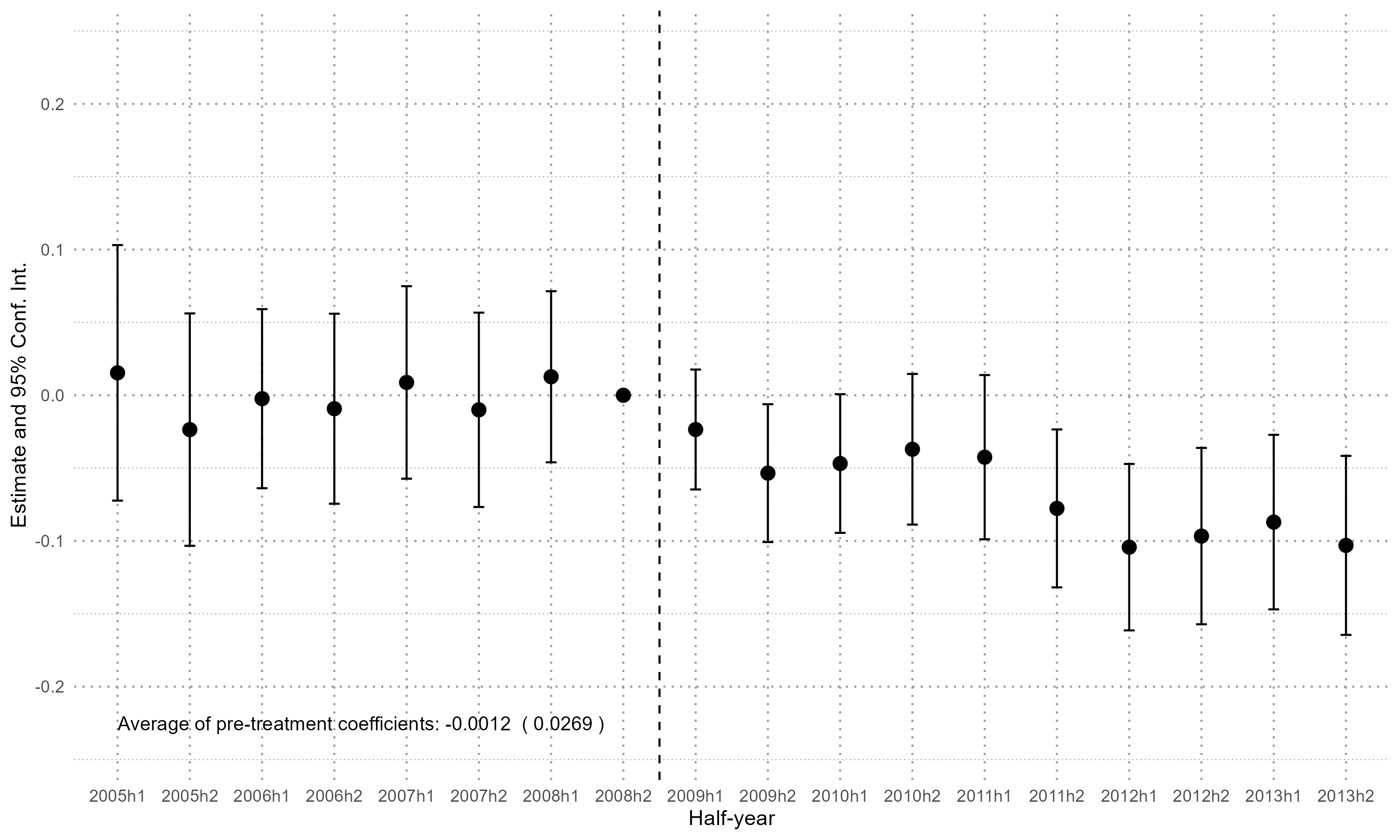}
\end{center}
\end{minipage}\hspace{1cm}}
 \begin{minipage}{.999\textwidth} \footnotesize
        \vspace{2mm}
       Notes: Panel (A) displays estimated coefficients $\hat \beta^{Mgap}_\tau$ for the regression described in \eqref{emp_eq6}. Panel (B) displays estimated coefficients $\hat \beta^{Mgap}_\tau+\hat \beta^{Fgap}_\tau$ for the  regression described in \eqref{emp_eq6}. 
       Bars depict 95\% confidence intervals. 
    \end{minipage}
\end{figure}

One potentially interesting policy-related question is the overall effect of EPSW on the gender wage gap across all LLMs, given that the majority of workers are employed in male-majority LLMs.  We find that EPSW increases the overall gender wage gap (in favor of men) by 2.74 percentage points (SE=0.0102).

\section{Conclusion}
We find that the equilibrium effects of Equal Pay for Similar Work (EPSW) policies dominate their direct effects. Therefore, imposing these policies may lead to unintended outcomes. Our modeling demonstrates that EPSW targeted specifically to equalize pay across protected classes of workers leads to firms segregating their workforce in equilibrium to avoid the bite of the policy. Although discriminatory forces may lead to pay gaps across groups without EPSW, segregation caused by EPSW results in the minority group of workers in a labor market receiving even lower relative wages. Our empirical evaluation of Chile's 2009 EPSW---which prohibited unequal pay for similar work across genders---supports these predictions. The policy caused gender segregation within firm to rise, and the gender wage gap to rise. Importantly, the rise in the wage gap only occurred in male-majority local labor markets; in local labor markets with majority female workers, the wage gap closed. Both of these findings are as predicted by our theoretical analysis.

Many important questions remain in understanding the equilibrium impacts of EPSW, and equity-related labor market policies in general. One difficulty is understanding the role of complementary policies; a benefit of studying the Chilean labor market is the relative dearth of alternative anti-discrimination policies prior to and contemporaneously with the enactment of EPSW. For this reason, further theoretical study of labor-market policies may be particularly fruitful. One particular avenue for further research is understanding firm incentives--especially those of large firms, for which segregation may be especially costly--to heterogenize on-the-job responsibilities and duties to evade EPSW.

\singlespacing
\bibliographystyle{plainnat} 

\bibliography{eplbib}
\onehalfspacing
\newpage
\appendix
\part*{\hspace{-4.9mm}\LARGE{Supplemental Appendix: ``Equal Pay for \emph{Similar} Work''}}

\noindent ~~~~~\large{\textbf{Diego Gentile Passaro, Fuhito Kojima, and Bobak Pakzad-Hurson}}

\normalsize
\bigskip

\noindent \Cref{proofs_app} presents proofs omitted from the main text. \Cref{microfoundation app}  presents microfoundations and discussions regarding our theoretical analyses. \Cref{emperical app} presents additional empirical results and descriptions.

\setcounter{equation}{0}
\renewcommand{\theequation}{A\arabic{equation}}

\setcounter{proposition}{0}
\renewcommand{\theproposition}{A\arabic{proposition}}

\setcounter{lemma}{0}
\renewcommand{\thelemma}{A\arabic{lemma}}

\setcounter{remark}{0}
\renewcommand{\theremark}{A\arabic{remark}}
\section{Proofs}\label{proofs_app}

\subsection*{Proof of \Cref{no_EPL_equilibrium_prop}}

Consider any outcome $O=\{(f_g^i(v),w_i^g(v))\}_{v \in [0,1], i =1,2, g=A,B}$ in which the following hold for almost all $v\in[0,1]$ and all $g\in\{A,B\}$:
\begin{enumerate}
    \item $f_g^1(v)+f_g^2(v) = f_g(v)$, and 
    \item for all $i\in \{1,2\}$, $w_i^g(v)=v$ if $f_g^i(v)>0$.
\end{enumerate}
We establish the desired result through three lemmas regarding these two enumerated conditions. \Cref{noepsw_lemma1} shows that any outcome satisfying these two enumerated conditions is a core outcome, \Cref{noespw_lemma2} shows that there are uncountably many non-equivalent outcomes satisfying these two conditions (implying that there are uncountably many core outcomes), and \Cref{noespw_lemma3} shows that there are no core outcomes that fail either of the two enumerated conditions.\footnote{\cite{bertrand} study Bertrand competition under a wage monotonicity constraint. Their analysis and proof differ due to the discrepancy in setting.}

\begin{lemma}\label{noepsw_lemma1}
Any outcome $O$ satisfying the two enumerated conditions in the Proof of \Cref{no_EPL_equilibrium_prop} is a core outcome.
\end{lemma}
\begin{proof}[Proof of \Cref{noepsw_lemma1}]
 Suppose not for the sake of contradiction. Then there are a firm $j$ and a distinct outcome (for firm $j$)  $\tilde O_j:=\{(\tilde f_g^j(v),\tilde w_j^g(v))\}_{v \in [0,1], g=A,B}$ that blocks $O.$ In order for $\tilde O_j$ to block $O$ it must be that $\pi^{\tilde O_j}_j > \pi_j^{O_j}$. However, it is also the case that
\begin{align*}
     \pi^{O_j}_j&= \beta \int_0^1 [v-w_j^A(v)]f^j_A(v)dv +  \int_0^1 [v-w_i^B(v)]f^j_B(v)dv \\
& = 0\\
& \geq \beta \int_0^1 [v-\tilde w_j^A(v)]\tilde f^j_A(v)dv +  \int_0^1 [v-\tilde w_i^B(v)]\tilde f^j_B(v)dv \\
& =  \pi^{\tilde {O}_j}_j
    \end{align*}
\noindent The second equality follows because, by the construction of $O$, either $f^j_g(v)=0$ or $w_j^g(v)=v$ for almost all $v$ and all $g$, therefore, the integrand is almost always equal to zero. The inequality follows because of the following exhaustive cases for almost all $v$ and all $g$, corresponding, respectively, to Conditions 1-4 of the definition of block:

\begin{itemize}
\item Suppose $\tilde w_j^g(v) \ge  w_j^g(v)$ and $\tilde w_j^g(v) >  w_{-j}^g(v)$, then it must be that $\tilde w_j^g(v) \ge v$ since $\max\{w_j^g(v),w_{-j}^g(v)\}=v$, which makes the integrand weakly negative,

\item Suppose $\tilde w_j^g(v) \ge  w_j^g(v)$ and $\tilde f_g^j(v)+f_g^{-j}(v) \le f_g(v)$. If $\tilde f_g^j(v)=0$ then the integrand is weakly negative. If $\tilde f_g^j(v)>0$ then it must be that $f_g^{-j}(v) < f_g(v)$, and by the construction of $O$ that  $f_g^j(v)+f_g^{-j}(v) = f_g(v)$, it must be that $f_g^j(v)>0.$ Therefore, it must be that $ w_j^g(v)=v$, and the requirement that $\tilde w_j^g(v) \ge  w_j^g(v)$ makes the integrand weakly negative.

\item Suppose $\tilde w_j^g(v) >  w_{-j}^g(v)$ and $\tilde f_g^j(v)+f_g^{j}(v) \le f_g(v)$. If $\tilde f_g^j(v)=0$ then the integrand is weakly negative. If $\tilde f_g^j(v)>0$ then it must be that $f_g^{j}(v) < f_g(v)$, and by the construction of $O$ that  $f_g^j(v)+f_g^{-j}(v) = f_g(v)$, it must be that $f_g^{-j}(v)>0.$ Therefore, it must be that $w_{-j}^g(v)=v$, and the requirement that $\tilde w_j^g(v) >  w_{-j}^g(v)$ makes the integrand strictly negative.

\item Suppose $\tilde f_g^j(v)+f_g^j(v)+f_g^{-j}(v) \le f_g(v)$ then it must be that $\tilde f_g^j(v)=0$ since by the construction of $O$ it is the case that $f_g^j(v)+f_g^{-j}(v) = f_g(v)$. Therefore, the integrand is weakly negative.

\end{itemize}
\noindent These cases reveal a contradiction with the premise that $\pi^{\tilde O_j}_j > \pi_j^{O_j}$. Therefore, $O$ is a core outcome. 
\end{proof}

\begin{lemma}\label{noespw_lemma2}
There exist a continuum of (non-equivalent) core outcomes satisfying the two enumerated conditions in the Proof of \Cref{no_EPL_equilibrium_prop}.
\end{lemma}
\begin{proof}[Proof of \Cref{noespw_lemma2}]

For any $v^*\in[0,1]$ there is an outcome $O$ such that $f_g^1(v)=f_g(v)$ and $w_1^g(v)=v$ for all $g$ and $v\leq v^*$, and $f_g^2(v)=f_g(v)$ and $w_2^g(v)=v$ for all $g$ and $v> v^*$. 
\end{proof}

 It remains to show that any core outcome must be such that all workers are hired and all workers receive a wage equal to their productivity.
\begin{lemma}\label{noespw_lemma3}
There exist no core outcomes which do not satisfy the two enumerated conditions in the Proof of \Cref{no_EPL_equilibrium_prop}.
\end{lemma}
\begin{proof}[Proof of \Cref{noespw_lemma3}]
 Suppose for contradiction that there is a core outcome $O=\{(f_g^i(v),w_i^g(v))\}_{v \in [0,1], i =1,2, g=A,B}$ such that there exists a set with positive measure $V$ where the following fails for some $g\in\{A,B\}$ and all $v\in V$: 
\begin{enumerate}
    \item $f_g^1(v)+f_g^2(v) = f_g(v)$, and 
    \item for all $i\in \{1,2\}$, $w_i^g(v)=v$ if $f_g^i(v)>0$.
\end{enumerate}

Throughout, we assume that the above conditions are violated for $g=B$, and we construct a blocking outcome $\tilde O_j=\{(\tilde f_g^j(v),\tilde w_j^g(v))\}_{v \in [0,1],g=A,B}$, where $(\tilde f_A^j(v),\tilde w_j^A(v))=( f_A^j(v), w_j^A(v))$ (or $\tilde O_{-j}=\{(\tilde f_g^{-j}(v),\tilde w_{-j}^g(v))\}_{v \in [0,1],g=A,B}$ where $(\tilde f_A^{-j}(v),\tilde w_{-j}^A(v))=( f_A^{-j}(v), w_{-j}^A(v))$), i.e. we do not change the outcome for $A-$group workers for either firm. The argument in which the violations occur for $g=A$ is analogous, where terms related to the firm's profit must be multiplied by $\beta$.

We show a contradiction by considering six exhaustive cases: By countable additivity of measure, the set of productivities that fails one of the above two enumerated points  has positive measure if and only if at least one of the sets in the following six cases has a positive measure. 

First, suppose there exist a firm $j$ and a subset of productivities $V \subset [0,1]$ with positive measure such that $w_j^g(v)>v$ for all $v \in V$. We proceed by showing that firm $j$ can fire workers of type $(g,v)$, $v\in V$ and increase its profit. Let $\tilde O_j$ be defined as follows: 
\begin{align*}
\tilde f_g^j(v) & := \begin{cases} 
f_g^j(v) & \text{if }  v\notin V,  \\
0 & \text{if }  v\in V.
\end{cases}
~~~~~~ \tilde w_j^g(v) := \begin{cases} 
w_j^g(v) & \text{if }  v\notin V,  \\
0 & \text{if }  v\in V.
\end{cases}
\end{align*}
\noindent $\tilde O_j$ blocks $O$ as $j$'s profit increases and Condition 4 of the definition of block is satisfied for all $v\in V$ (i.e. the workers in $V$  are fired) and for all $v\in [0,1]\setminus V$ Condition 2 of the definition of block is satisfied (i.e. there is no change in the hiring or wages of workers with $v\in[0,1]\setminus V$).  

Therefore, we proceed with the assumption that for each firm $j$, $w_j^g(v)\leq v$ for almost all $v$. 

Second, suppose there exist a firm $j$ and a subset of productivities $V$ with positive measure such that  $f_g^1(v)+f_g^2(v) < f_g(v)$ and $w_j^g(v)< v$ for all $v\in V$. We proceed by showing that firm $j$ can hire unemployed workers of type $(g,v)$, $v\in V$ at the going wage and increase its profit. Let $\tilde O_j$ be defined as follows: 
 \begin{align*}
\tilde f_g^j(v) :=  \begin{cases} 
f_g(v)-f_g^{-j}(v) & \text{if }  v\in V,\\
f_g^j(v) & \text{if }  v\notin V.
\end{cases}
~~~~~~~~~~~\tilde w_j^g(v) :=   
w_{j}^g(v)  \text{ for all $v$.}
\end{align*}
\noindent $\tilde O_j$ blocks $O$ as firm $j$'s profit increases  as some previously unemployed workers are hired at a wage strictly less than their productivity while all existing workers at $j$ continue to be employed at the same wage as before, and Condition 2 of the definition of a block is satisfied for all $v\in[0,1]$ (i.e. no worker receives a wage cut and no workers are poached from firm $-j$).

Third, suppose there exists a firm $j$ and a subset of productivities $V$ with positive measure such that 
 $f_g^1(v)+f_g^2(v) < f_g(v)$ and $w_j^g(v)= v$ for all $v\in V$. We proceed by showing that firm $j$ can fire its current workers of type $(g,v)$, $v\in V$ and hire the unemployed workers of the same type at a lower wage to increase its profit. Let $\tilde O_j$ be defined as follows:
  \begin{align*}
\tilde f_g^j(v) :=  \begin{cases} 
f_g(v)-f_g^{1}(v)-f_g^{2}(v) & \text{if } v\in V, \\
f^j_g(v) & \text{otherwise.}
\end{cases}
~~~~~~~~~~~\tilde w_j^g(v) :=  \begin{cases} 
0 & \text{if } v\in V, \\
w_g^j(v) & \text{otherwise.}
\end{cases}
\end{align*}
\noindent $\tilde O_j$ blocks $O$. To see that $\tilde O_j$ blocks $O$, note that Condition 4 of the definition of block is satisfied for all $v\in V$ by construction, and Condition 2 of the definition of block is satisfied for all $v\in [0,1]\setminus V$ (i.e. there are no changes in employment or wages for workers with $v\in [0,1]\setminus V$). Therefore, it remains only to show that $\tilde O_j$ yields firm $j$ a strictly higher profit than does $O_j$, and this fact follows because $j$ earns no profit in $O_j$ from worker types $(g,v)$, $v\in V$ because $w_j^g(v)=v$, and because $j$ earns strictly positive profit in $\tilde O_j$ from worker types $(g,v)$, $v\in V$ because $\tilde f_g^j(v)=
f_g(v)-f_g^{1}(v)-f_g^{2}(v)>0$ (where the inequality follows from the ongoing assumption of this case) and $\tilde w_j^g(v)=0<v$ for all $v\in V\setminus\{0\}.$

The previous two cases exhaust the possibility of a core outcome in which  $f_g^1(v)+f_g^2(v) < f_g(v)$ for any $g$ and a subset of productivities with positive measure. Therefore, we proceed with the assumption that $f_g^1(v)+f_g^2(v) = f_g(v)$ for almost all $v$.

Fourth, suppose that there exist $j$ and a set $V$ of productivities with positive measure such that $w_j^g(v)< v$ and $f_g^j(v)= f_g(v)$ for all $v\in V$. We proceed by showing that firm $-j$ can poach workers of type $(g,v)$, $v\in V$ from firm $j$ at marginally higher wages and increase its profit. Let $\tilde O_{-j}$ be defined as follows:
  \begin{align*}
\tilde f_g^{-j}(v) :=  \begin{cases} 
f_g(v) & \text{if } v\in V, \\
f^{-j}_g(v) & \text{otherwise.}
\end{cases}
~~~~~~~~~~~\tilde w_{-j}^g(v) :=  \begin{cases} 
(w^g_j(v)+v)/2 & \text{if } v\in V, \\
w^g_{-j}(v) & \text{otherwise.}
\end{cases}
\end{align*}
\noindent $\tilde O_{-j}$ blocks $O$. To see that $\tilde O_{-j}$ blocks $O$, note that Condition 1 of the definition of block is satisfied for all $v\in V$ by construction, and Condition 2 of the definition of block is satisfied for all $v\in [0,1]\setminus V$ (i.e. there are no changes in employment or wages for workers with $v\in [0,1]\setminus V$). Therefore, it remains only to show that $\tilde O_{-j}$ yields firm $-j$ a strictly higher profit than does $O_{-j}$, and this fact follows because $-j$ earns no profit in $O_{-j}$ from worker types $(g,v)$, $v\in V$ because $f_g^{-j}(v)=0$, and because $-j$ earns strictly positive profit in $\tilde O_{-j}$ from worker types $(g,v)$, $v\in V$ which follows from $w_{-j}^g(v)=(w_j^g(v)+v)/2<v$ and $f_g^{-j}(v)>0$  for all $v\in V\setminus\{0\}.$

Fifth, suppose that there exist $j$ and $V \subset [0,1]$ with positive measure such that $0\le w_{-j}^g(v)\leq w_j^g(v)<v$ and $f_g^j(v)\in(0, f_g(v))$ for all $v \in V$. Then, there exists $\varepsilon>0$ and $V' \subset V$ with positive measure such that  $0\le w_{-j}^g(v)\leq w_j^g(v)<v-\varepsilon$ and $f_g^j(v)\in(0, f_g(v)-\varepsilon)$ for all $v \in V'.$ We proceed by showing that firm $j$ can retain workers of type $(g,v)$, $v\in V'$, and poach all workers of the same type from firm $-j$ by marginally increasing its wages and increase its profit.  For a constant $\varepsilon'>0$, consider $\tilde O_{j}$ where for all $v$:
\begin{align*}
\tilde f_g^{j}(v) & := \begin{cases} 
f_g(v) & \text{ if }  v \in V', \\
f_g^j(v) & \text{otherwise.}
\end{cases}
~~~~ \tilde w_{j}^g(v)  := \begin{cases} 
0 & \text{if } \tilde f_g^j(v)=0,  \\
w_j^g(v)+\varepsilon' & \text{otherwise.}
\end{cases}
\end{align*}
$\tilde O_{j}$ blocks $O_{j}$ for the following reasons: Condition 1 of the definition of block is satisfied for all $v\in V'$ since $w_{-j}^g(v)\leq w_j^g(\bar v)< \tilde w_{j}^g(v)$ for all $v\in V'$ by construction, and Condition 2 of the definition of the block is satisfied for all $v\notin V'$. To see that firm $j$'s profit increases, first note that firm $j$ benefits from hiring additional workers from $V'$, which results in an additional profit of at least $(\varepsilon-\varepsilon')\varepsilon \mu(V').$ Meanwhile, the firm may lose from paying more for existing workers, but the associated loss is bounded from above by $\varepsilon' \beta.$ Therefore, for any sufficiently small $\varepsilon',$ firm $j$'s profit increases, as desired.

Cases 4 and 5 exhaust the possibility of a core outcome in which there exists a set $V$ of positive Lebesgue measure such that $\max\{w^g_1(v),w^g_2(v)\}<v$ for almost all $v\in V$. Therefore, we proceed with the assumption that for almost any $v\in[0,1]$ there exists a firm $j$ such that $w^g_j(v)=v$.

Sixth, suppose there exist a set $V$ of positive Lebesgue measure and a firm $j$ such that $0\leq w^g_{-j}(v)<w^g_j(v)=v$ and $f_g^{-j}(v)\in(0, f_g(v))$ for all $v\in V.$ We proceed by showing that firm $j$ can fire all workers of type $(g,v)$, $v\in V$ and poach all workers of the same type from firm $-j$ by paying  higher wages than $-j$ does and increase its profit. Let $\tilde O_j$ be defined as follows:
  \begin{align*}
\tilde f_g^j(v) :=  \begin{cases} 
f_g^{-j}(v) & \text{if } v\in V, \\
f^j_g(v) & \text{otherwise.}
\end{cases}
~~~~~~~~~~~\tilde w_j^g(v) :=  \begin{cases} 
(w^g_{-j}(v)+v)/2 & \text{if } v\in V, \\
w_g^j(v) & \text{otherwise.}
\end{cases}
\end{align*}
\noindent $\tilde O_j$ blocks $O$. To see that $\tilde O_j$ blocks $O$, note that Condition 3 of the definition of block is satisfied for all $v\in V$ by construction, and Condition 2 of the definition of block is satisfied for all $v\in [0,1]\setminus V$ (i.e. there are no changes in employment or wages for workers with $v\in [0,1]\setminus V$). Therefore, it remains only to show that $\tilde O_j$ yields firm $j$ a strictly higher profit than does $O_j$, and this fact follows because $j$ earns no profit in $O_j$ from worker types $(g,v)$, $v\in V$, because $w_j^g(v)=v$, and because $j$ earns strictly positive profit in $\tilde O_j$ from worker types $(g,v)$, $v\in V$ because $\tilde f^j_g(v)=f^{-j}_g(v)>0$ (where last inequality follows from the ongoing assumption of this case) and $\tilde w_j^g(v)<v$ for all $v\in V\setminus\{0\}.$

As these six cases are exhaustive and none of them admits a core outcome, we have completed the argument that any core outcome must be such that all workers are hired and all workers receive a wage equal to their productivity. 
\end{proof}

\subsection*{Proof of \Cref{generic_segregation}}
We show the desired conclusions via multiple steps. First, we show that if firm (without loss of generality) 1 hires almost no $B-$group workers, then it hires almost all $A-$group workers and firm 2 hires almost all $B-$group workers. Therefore, the only possible core outcomes that do not feature segregation involve each firm hiring from both groups at a common wage. Second, we show that if there exists a core outcome  featuring common wages, it must be that the two firms pay different wages. Third, we show that generically there exists no core outcome in which each firm pays common wages and hires from both groups.  Throughout, it suffices to assume $w_i^g(v)\leq v$ for all $v\in[0,1]$, all $i\in\{1,2\}$, and all $g\in\{A,B\}$ by the Individual Rationality Condition.\\

\noindent \textbf{Step 1:} 
We show that in any core outcome $O$ if $f_B^1(v)=0$ for almost all $v$, then it must be that $f_A^1(v)=f_A(v)$ for almost all $v$ and $f_B^2(v)=f_B(v)$ for almost all $v$. We show this in several exhaustive cases.

First, we will show that if $f_B^1(v)=0$ for almost all $v$, then it must be that there exists a set $V\subset[0,1]$ of positive Lebesgue measure such that  $f_A^1(v)>0$ for all $v\in V$. Suppose not toward a contradiction so that  $f_B^1(v)=0$ for almost all $v$ and $f_A^1(v)=0$ for almost all $v$. Then,  $\pi_1^O=0$. By the Equal Profit Condition it must be that $\pi^O_2=0$. 
There are two possibilities to consider. First, it may be that there exist a group $g\in\{A,B\}$ and a set $V$ with positive Lebesgue measure such that $f_g^2(v)<f_g(v)$ for all $v\in V$. Without loss of generality let $g=A$. Then firm 1 can block outcome $O$ via $\tilde O_1$ where
\begin{align*}
\tilde f_A^1(v) & := 
 f_A(v)-f_A^2(v)-f_A^{1}(v) \text{ for all } v  ~~~~~~~~~~  \tilde w_1^A(v)  := 0 \text{ for all } v, 
 \\ \tilde f_B^1(v) & := 0 \text{ for all } v ~~~~~~~~~~~~~~~~~~~~~~~~~~~~~~~~~~~
\tilde w_1^B(v)  := 0 \text{ for all } v.
\end{align*}
which yields a positive profit, contradiction. Second, it may be that $f_A^2(v)=f_A(v)$ and $f_B^2(v)=f_B(v)$ for almost all $v$. By EPSW, it must be that firm 2 pays a common wage to almost all workers: there exists $w_2\geq 0$ such that $w_2=w_2^A(v)=w_2^B(v)$ for almost all $v$. Moreover, by the ongoing assumption that $w_i^g(v)\leq v$ for all $v\in[0,1]$, all $i\in\{1,2\}$, and all $g\in\{A,B\}$, it must be the case that $w_2=0$. But then for any $w^*\in(0,1)$ firm 1 can block $O$ via $\tilde O_1$ where for each $g\in\{A,B\}$:
    \begin{align*}
\tilde f_g^1(v) & := \begin{cases} 
 f_g(v)-f_g^1(v) & \text{ if } v\geq w^*,  \\
 0 & \text{ otherwise.}
\end{cases}~~~~ \tilde w_1^g(v)  := \begin{cases} 
0 & \text{if } \tilde f_g^1(v)=0, \\
w^* &  \text{otherwise.}
\end{cases}
\end{align*}
\noindent which yields positive profit, contradiction.

Second, we  will prove that if $f_B^1(v)=0$ for almost all $v$ and there exists a positive Lebesgue measure set $V$ such that  $f_A^1(v)>0$ for all $v\in V$, then there does not exist any pair of sets $V_A$ and $V_B$ with positive Lebesgue measure such that $f_A^2(v)>0$ for all $v\in V_A$ and $f_B^2(v)>0$ for all $v\in V_B$. To show this,  assume for contradiction that there exist  sets $V_A$ and $V_B$ with positive Lebesgue measure such that $f_A^2(v)>0$ for all $v\in V_A$ and $f_B^2(v)>0$ for all $v\in V_B$.  Then due to the  EPSW, it must be that firm 2 pays a common wage $w_2$ to almost all workers it hires, i.e. $w^A_2(v)=w^B_2(v)=w_2$ for almost all $v$ such that $f^2_A(v)>0$ and $f^2_B(v)>0,$ respectively. By our ongoing assumption that  $w_2^g(v)\leq v$ for all $v\in[0,1]$ and all $g\in\{A,B\}$, it must be the case that $f^2_A(v)=0$ for all $v\leq w_2$. This implies that firm 2 is earning a positive profit from the $A-$group workers it hires as $w_2<v$ for almost all $v\in V_A$. By an argument similar to those made in the proof of Case 4 of \Cref{noespw_lemma3}, firm 1 can block $O$ by ``poaching'' some subset of these workers and increase its profits.

Third, we will establish that if $f_B^1(v)=0$ for almost all $v$ and there exists a  set $V$ with positive Lebesgue measure such that $f_A^1(v)>0$ for all $v\in V$ then $f_A^2(v)=0$ for almost all $v$. To show this, by the previous case, it suffices to consider for contradiction that $f_B^2(v)=0$ for almost all $v$ and that there exists a positive Lebesgue measure set $V_A$ such that $f_A^2(v)>0$ for  all $v\in V_A$. The conclusion of \Cref{no_EPL_equilibrium_prop} applies, and the two firms ``Bertrand'' compete away profits from $A-$group workers such that they both earn zero profits, i.e. $w_1^A(v)=v$ for almost all $v$ such that $f^1_A(v)>0$ and $w_2^A(v)=v$ for almost all $v$ such that $f^2_A(v)>0$. But then an arbitrary firm $j$ can block via outcome $\tilde O_j$:
\begin{align*}\tilde f_A^j(v) & := 0 \text{ for all } v ~~~~~~~~~~~~~~~~~~~~~~~~~~~\tilde w_j^A(v)  := 0  \text{ for all } v\\
\tilde f_B^j(v) & := f_B(v)-f_B^{-j}(v) \text{ for all } v ~~~~~~~~~\tilde w_j^B(v)  := 0  \text{ for all } v
\end{align*}
The ongoing assumption that $f_B^{1}(v)=0$ for almost all $v$ and $f_B^{2}(v)=0$ for almost all $v$ demonstrates that $\tilde O_j$ indeed blocks $O$.

Fourth, we will show that if $f_B^1(v)=0$ for almost all $v$, there exists a  set $V$ with positive Lebesgue measure such that $f_A^1(v)>0$ for all $v\in V$, and $f_A^2(v)=0$ for almost all $v$ then it must be that $f_A^1(v)=f_A(v)$ for almost all $v$. This claim is shown via the second and third cases in the proof of \Cref{noespw_lemma3}.\\

\noindent \textbf{Step 2:} Step 1 implies the only remaining possible core outcomes in which firms do not completely segregate involves each firm hiring both groups of workers, i.e. there exist  sets $V^1_A, V^2_A, V^1_B, V^2_B$ with positive Lebesgue measure such that $f^1_A(v)>0$ for all $v\in V^1_A,$ $f^2_A(v)>0$ for all $v\in V^2_A,$ $f^1_B(v)>0$ for all $v\in V^1_B,$ and $f^2_B(v)>0$ for all $v\in V^2_B$. Suppose for contradiction that there exists such a non-segregation core outcome $O$. We have argued that in any such core outcome, each firm $i\in\{1,2\}$ must pay a common wage $w_i$ to almost every worker it hires. Without loss of generality, let $w_1\leq w_2$. We first claim that, if $w_1< w_2$:

\begin{enumerate}

\item For almost all $v<w_1$, all workers of  productivity $v$ are unemployed: for almost all $v<w_1$, $f^1_A(v)=f^1_B(v)=f^2_A(v)=f^2_B(v)=0$,
\item for almost all $v\in (w_2,1]$, all workers of productivity $v$ are hired by firm 2, i.e. $f^2_A(v)=f_A(v)$ and 
$f^2_B(v)=f_B(v)$, and
\item for almost all $v\in [w_1,w_2]$, all workers of productivity $v$ are hired by firm 1, i.e. $f^1_A(v)=f_A(v)$ and $f^1_B(v)=f_B(v)$.
\end{enumerate}

Point 1 follows from our previous individual rationality argument that no firm hires a positive measure of workers at wage higher than productivity in any core outcome. Point 2 is demonstrated with the following argument. Individual rationality implies  firm 2 will hire almost no workers with productivity $v\in[w_1,w_2]$. If there exists a  set $V\subset[w_2,1]$ with  positive measure such that  $f^2_A(v)<f_A(v)$ or $f^2_B(v)<f_B(v)$ for all $v\in V$, then firm 2 can block outcome $O$ via $\tilde O_2$ such that for all $g\in\{A,B\}$:
    \begin{align*}
\tilde f_g^2(v) & := \begin{cases} 
 f_g(v) & \text{ if } v\in (w_2,1],  \\
 0 & \text{ otherwise.}
\end{cases}~~~~~ \tilde w_2^g(v) := \begin{cases} 
0 & \text{if } \tilde f_g^2(v)=0, \\
w_2 &  \text{otherwise.}
\end{cases}
\end{align*}
\noindent which clearly blocks $O$ as firm 2's profit increases (from hiring additional workers at a positive marginal profit) and Condition 1 of the definition of block is satisfied for all $v \in (w_2,1]$. Point 3 follows from a similar argument as the proof of point 2.

We now claim that in any core outcome $O$  it must be that $w_1<w_2$ (recall that we have  assumed $w_1\leq w_2$). To see this, suppose for contradiction that $w_1=w_2$. If $w_1=w_2<1$,  then at least one firm $i\in\{1,2\}$ receives total profit $\pi^O_i<\int_{w_1}^{1}(v-w_1)[\beta f_A(v)+f_B(v)]dv$. But because $F_A$ and $F_B$ are atomless, for sufficiently small $\epsilon>0$, firm $i$ would receive $\int_{w_1+\epsilon}^{1}(v-w_1-\epsilon)[\beta f_A(v)+f_B(v)]dv>\pi^O_i$ by instead setting wage $w'=w_1+\epsilon$ and hiring all workers with productivity strictly greater than $w_1+\epsilon$. This is a contradiction to the assumption that $O$ is a core outcome. If $w_1=w_2=1$, then $\pi^O_1=\pi^O_2=0$ and so either firm could block $O$ by  setting any wage $w'\in(0,1)$ and hiring all workers with productivities $v\in[w',1)$, which yields profit $\int_{w'}^{1}(v-w')[\beta f_A(v)+f_B(v)]dv>0$. This contradicts the assumption that  $O$ is a core outcome. As these two cases are exhaustive, it must be the case that $w_1<w_2$ in any core outcome in which $w_1\leq w_2$.\\

\noindent \textbf{Step 3:} We now show that generically, any outcome $O$ in which  $f_A^1(v)=f_B^1(v)=f_A^2(v)=f_B^2(v)=0$ for almost all $v<w_1$, $f_A^1(v)=f_A(v)$ and $f_B^1(v)=f_B(v)$ for almost all $v\in[w_1,w_2]$, and $f_A^2(v)=f_A(v)$ and $f_B^2(v)=f_B(v)$ for almost all $v>w_2$ is not a core outcome.  We consider two exhaustive cases.

First, suppose that  $0<w_1<w_2$. Let $\pi_{i,g}$ 
represent the profit  firm $i\in\{1,2\}$ receives from all workers in group $g\in\{A,B\}$ in outcome $O$, or more formally,
$$
\pi_{i,A}:=\beta \int_{0}^{1}(v-w_i(v))f_A^i(v)dv, ~~~~~~~
\pi_{i,B}:= \int_{0}^{1}(v-w_i(v))f_B^i(v)dv.
$$
It must be the case that $\pi_{1,A}+\pi_{2,A}\geq\pi_{1,B}+\pi_{2,B}$, or vice versa. Without loss of generality, assume $\pi_{1,A}+\pi_{2,A}\geq\pi_{1,B}+\pi_{2,B}$. We claim that firm 1 can block outcome $O$ via outcome $\tilde O_1$ where for some small $\epsilon'>0$:
    \begin{align*}
\tilde f_A^1(v) & := 
 f_A(v) ~~~~~ \text{ for all } v\in [0,1]
~~~~~~\tilde w_1^A(v)  := \begin{cases} 
0 & \text{if } v<w_1, \\
w_1 &  \text{if } v\in[w_1,w_2], \\
w_2+\epsilon' &  \text{if } v>w_2.
\end{cases}\\
\tilde f_B^1(v) & := 
 0 ~~~~~~~~~~~ \text{ for all } v\in [0,1]~~~
~~~ \tilde w^B_1(v)  := 
 0 ~~~~~~~~~~~ \text{ for all } v\in [0,1]
\end{align*}
 For sufficiently small $\epsilon'$,  $\pi^{\tilde O_1}_{1}>\pi_{1,A}+\pi_{2,A}\geq \pi_1$, where the last inequality comes from the assumption that $\pi_{1,A}+\pi_{2,A}\geq\pi_{1,B}+\pi_{2,B}$ and the Equal Profit Condition $\pi_1=\pi_2=\frac{1}{2}(\pi_{1,A}+\pi_{2,A}+\pi_{1,B}+\pi_{2,B})$. For almost all $A-$group workers, Condition 1 of the definition of block is satisfied, and for almost all $B-$group workers, Condition 4 of the definition of block is satisfied. Therefore, $\tilde O_1$ blocks $O$.

Second, suppose that $0=w_1<w_2$. The argument is analogous to case 1 if $\pi_{1,A}+\pi_{2,A}>\pi_{1,B}+\pi_{2,B}$ or $\pi_{1,A}+\pi_{2,A}<\pi_{1,B}+\pi_{2,B}$.  
Below, we argue that the condition $\pi_{1,A}+\pi_{2,A}=\pi_{1,B}+\pi_{2,B}$ is non-generic on the space of distributions:

First, we will show the openness of the set of distributions $F_A,F_B$ such that $\pi_{1,A}+\pi_{2,A}>\pi_{1,B}+\pi_{2,B}$ or $\pi_{1,A}+\pi_{2,A}<\pi_{1,B}+\pi_{2,B}$. To do so, assume that the former inequality holds at $F_A,F_B$ (the other case is analogous), where the core outcome wages of firms 1 and 2 are $0$ and $w_2$, respectively. Let $\epsilon>0$ be small enough that \begin{align}
\pi_{1,A}+\pi_{2,A}>\pi_{1,B}+\pi_{2,B}+\epsilon.\label{equation:epsilon-inequality}
\end{align} 
By the Portmanteu Theorem,\footnote{See Theorem 2.8.1 (c) of \cite{Ash}. 
 } for any $\tilde w_2 \in [0,1]$ it follows that 
\begin{align}
\beta \int_0^{\tilde w_2} v f'_A(v)dv \to \beta \int_0^{\tilde w_2} v f_A(v)dv, \notag 
\end{align}
as $(F'_A,F'_B)$ with densities $f'_A,f_B'$ converges in weak$^*$ topology to $(F_A,F_B).$ 
Thus, there is a neighborhood of $(F_A,F_B)$ such that, for any $(F'_A,F'_B)$ in that neighborhood, we have
\begin{align}
\beta \int^{\tilde w_2}_0 v f'_A(v)dv \in \left ( \beta \int^{\tilde w_2}_0 v f_A(v)dv-\frac{\epsilon}{8},  \beta \int^{\tilde w_2}_0 v f_A(v)dv+\frac{\epsilon}{8}\right ).\label{inequality2}
\end{align}
 Next, again by Portmanteu Theorem, it follows that
for any $\gamma< w_2$ (where  $w_2$ is the wage paid by firm 2 when the distributions of worker productivities are given by $(F_A,F_B)$),
$$
\beta \int_0^{w_2-\gamma} v f'_A(v)dv+\int_0^{w_2-\gamma} v f'_B(v)dv \to \beta \int_0^{w_2-\gamma} v f_A(v)dv+\int_0^{w_2-\gamma} v f_B(v)dv 
$$
and
$$
\hspace{-17mm}\beta \int_{w_2-\gamma}^1 (v-(w_2-\gamma)) f'_A(v)dv+\int_{w_2-\gamma}^1 (v-(w_2-\gamma)) f'_B(v)dv \to \beta \int_{w_2-\gamma}^1 (v-(w_2-\gamma)) f_A(v)dv+\int_{w_2-\gamma}^1 (v-(w_2-\gamma)) f_B(v)dv,
$$
as $(F'_A,F'_B)$  converges in weak$^*$ topology to $(F_A,F_B)$. 
Because we know that
$$
\beta \int_0^{w_2-\gamma} v f_A(v)dv+\int_0^{w_2-\gamma} v f_B(v)dv
<\beta \int_{w_2-\gamma}^1 (v-(w_2-\gamma)) f_A(v)dv+\int_{w_2-\gamma}^1 (v-(w_2-\gamma)) f_B(v)dv$$
 from the construction of outcome $O$ at $(F_A,F_B)$, it follows that
there exists a neighborhood of $(F_A,F_B)$ such that, for any pair of distributions $(F'_A,F_B')$  in that neighborhood, we have 
$$
\beta \int_0^{w_2-\gamma} v f'_A(v)dv+\int_0^{w_2-\gamma} v f'_B(v)dv <\beta \int_{w_2-\gamma}^1 (v-(w_2-\gamma)) f'_A(v)dv+\int_{w_2-\gamma}^1 (v-(w_2-\gamma)) f'_B(v)dv,
$$
and, with an analogous argument, that 
$$
\beta \int_0^{w_2+\gamma} v f'_A(v)dv+\int_0^{w_2+\gamma} v f'_B(v)dv >\beta \int_{w_2+\gamma}^1 (v-(w_2+\gamma)) f'_A(v)dv+\int_{w_2+\gamma}^1 (v-(w_2+\gamma)) f'_B(v)dv.
$$
These inequalities imply that if firm 2 pays $w'_2$ to all workers it employs and firm 1 pays $w'1=0$ to all workers it employs, both given distributions $(F'_A,F'_B)$, such that the Equal Profit Condition is satisfied, then $w'_2 \in (w_2-\gamma,w_2+\gamma)$.
Moreover, by monotonicity of the expression $
\beta \int_0^{w_2'} v f'_A(v)dv$ in $w_2'$, it follows that 
$
\beta \int_0^{w_2'} v f'_A(v)dv  \in \left ( \beta \int_0^{w_2-\gamma} v f'_A(v)dv, \beta \int_0^{w_2+\gamma} v f'_A(v)dv \right )$ for any distribution $(F'_A,F'_B)$ in the  neighborhood. Take a sufficiently small $\gamma>0$ such that   
\begin{align}
    \beta \int_0^{w_2-\gamma} v f_A(v)dv>\beta \int_0^{w_2} v f_A(v)dv-\frac{\epsilon}{8} \text{ and } \beta \int_0^{w_2+\gamma} v f_A(v)dv <\beta \int_0^{w_2} v f_A(v)dv+\frac{\epsilon}{8}.\label{lemmna-inequality}
\end{align} 
Moreover, we apply \eqref{inequality2} to $\tilde w_2=w_2-\gamma$ and $\tilde w_2=w_2+\gamma$ to obtain 
\begin{align}
\beta \int^{w_2-\gamma}_0 v f'_A(v)dv \in \left ( \beta \int^{w_2-\gamma}_0 v f_A(v)dv-\frac{\epsilon}{8},  \beta \int^{w_2-\gamma}_0 v f_A(v)dv+\frac{\epsilon}{8}\right ),\notag \\    
\beta \int^{w_2+\gamma}_0 v f'_A(v)dv \in \left ( \beta \int^{w_2+\gamma}_0 v f_A(v)dv-\frac{\epsilon}{8},  \beta \int^{w_2+\gamma}_0 v f_A(v)dv+\frac{\epsilon}{8}\right ) \label{inequality3}. \end{align}
Thus, by \eqref{lemmna-inequality} and \eqref{inequality3} it follows that $\pi'_{1,A} \in (\pi_{1,A}-\frac{\epsilon}{4},\pi_{1,A}+\frac{\epsilon}{4})$. 
Similarly we can establish that $\pi'_{i,g}$ for $i \in \{1,2\},g \in \{A,B\}$, each satisfies $\pi'_{i,g} \in (\pi_{i,g}-\frac{\epsilon}{4},\pi_{i,g}+\frac{\epsilon}{4})$.
This and \eqref{equation:epsilon-inequality} imply $\pi'_{1,A}+\pi'_{2,A}>\pi'_{1,B}+\pi'_{2,B},$
completing the proof of openness.

Second, we will show the denseness of the set of distributions $F_A,F_B$ such that $\pi_{1,A}+\pi_{2,A}>\pi_{1,B}+\pi_{2,B}$ or $\pi_{1,A}+\pi_{2,A}<\pi_{1,B}+\pi_{2,B}$. To do so, suppose for a given $F_A,~F_B$ satisfying our regularity conditions there is a core outcome that does not feature segregation, i.e. $\pi_{1,A}+\pi_{2,A}=\pi_{1,B}+\pi_{2,B}.$ We show that for any $\gamma_1>0$ there exist distributions $F'_A,~F'_B$ satisfying our regularity conditions such that $\vert F'_A(v)-F_A(v)\vert \leq \gamma_1$ and $\vert F'_B(v)-F_B(v)\vert \leq \gamma_1$ for all $v\in[0,1]$ and for which there does not exist a non-segregation core outcome.\footnote{Note that our pointwise notion of ``closeness'' corresponds to that of the weak-* topology on the set of distributions satisfying our regularity conditions.} We show this by modifying the distributions such that the same wages by firms 1 and 2, $0$ and $w_2$, respectively, yield equal profit, but show that aggregate profit derived from $A-$group workers no longer equals that derived from $B-$group workers.

Take any $\gamma\in(0,\min\{\underline f_A,\underline f_B,\gamma_1\})$ and any $w^*\in(0,w_2)$ such that $w_2-w^*<1-w_2$. Define $F'_A$ and $F'_B$ via their respective densities: 
\begin{align*}
\hspace{-7mm}f'_A(v) & := \begin{cases} 
f_A(v) & \text{ if }  v <w^*,\\
f_A(v)-\frac{\gamma}{\beta} & \text{ if }  v \in[w^*,w_2),\\
f_A(v)+\frac{\gamma}{\beta} & \text{ if }  v \in[w_2,w_2+(w_2-w^*)),\\
f_A(v) & \text{ otherwise}
\end{cases},~~~~~~~~~
f'_B(v):= \begin{cases} 
f_B(v) & \text{ if }  v <w^*,\\
f_B(v)+\gamma & \text{ if }  v \in[w^*,w_2),\\
f_B(v)-\gamma & \text{ if }  v \in[w_2,w_2+(w_2-w^*)),\\
f_B(v) & \text{ otherwise}
\end{cases}.
\end{align*}

Noting there exists a $\beta$ measure of $A-$group workers, the construction of $F'_A$ and $F'_B$ imply 

\begin{align*}
\hspace{-10mm}\pi_{2,A}+\pi_{2,B}= \beta \int_{w_2}^1 [v-w_2]f_A(v)dv+\int_{w_2}^1 [v-w_2]f_B(v)dv 
    = \beta \int_{w_2}^1 [v-w_2]f'_A(v)dv+\int_{w_2}^1 [v-w_2]f'_B(v)dv
    = \pi'_{2,A}+\pi'_{2,B},
\end{align*}

\begin{align*}
    \pi_{1,A}+\pi_{1,B}= \beta \int_{0}^{w_2} vf_A(v)dv+\int_{0}^{w_2} vf_B(v)dv 
    = \beta \int_{0}^{w_2} vf'_A(v)dv+\int_{0}^{w_2} vf'_B(v)dv= \pi'_{1,A}+\pi'_{1,B},
\end{align*}

\noindent where $\pi'_{i,g}$ represents firm $i$'s profit from $g-$group workers it hires given distribution $F'_g$. The above two equations imply that $\pi'_i:=\pi'_{i,A}+\pi'_{i,B}=\pi^{O_i}_i \textrm{ for } i\in\{1,2\}.$
It is also the case that 
\begin{align*}
\pi_{1,A}+\pi_{2,A}-(\pi'_{1,A}+\pi'_{2,A})= \beta \int_{w_2}^{w_2+(w_2-w^*)} [v-w_2]\frac{\gamma}{\beta}dv+\beta\int_{w^*}^{w_2} v\frac{\gamma}{\beta}dv
= \gamma(w_2-w^*)w_2 >0,
\end{align*}where the final inequality comes from the assumptions that $w_2>w^*>0$. Similarly it can be shown that  
$$\pi_{1,B}+\pi_{2,B}-(\pi'_{1,B}+\pi'_{2,B})<0.$$
The  two preceding equations imply that 
$\pi'_{1,A}+\pi'_{2,A}<\pi_{1,A}+\pi_{2,A},$ and 
$\pi_{1,B}+\pi_{2,B}<\pi'_{1,B}+\pi'_{2,B}.$
The initial condition that $\pi_{1,A}+\pi_{2,A}=\pi_{1,B}+\pi_{2,B}$ further implies that 
$\pi'_{1,A}+\pi'_{2,A}<\pi'_{1,B}+\pi'_{2,B}$ which completes the claim. \qed

\subsection*{Proof of \Cref{two-sided-EPL-equilibrium}}

\noindent \textbf{Proof of Part \ref{continuous-equilibria}:} We first constructively show that there exist a continuum of core outcomes $O(\gamma,\gamma')$ with group-based segregation, parameterized by the following class of wage functions for $\gamma, \gamma' \in [0,1]$:
\begin{align*}
 w_1(v) := \begin{cases} 
v & \text{if } v<\gamma, \\
\gamma &  \text{otherwise.}
\end{cases}~~~~~  w_2(v) := \begin{cases} 
0 & \text{if }  v<\gamma', \\
v &  \text{otherwise.}
\end{cases}
\end{align*}
\noindent with firms 1 and 2 hiring all $A-$ and $B-$group workers, respectively.
We show that these wage functions permit a core outcome if $\gamma \ge \gamma'$ and 
\begin{align}
\pi^{O(\gamma,\gamma')}_1=\beta \stackrel[\gamma]{1}{\int}\left(v-\gamma \right) f_A(v)dv=\stackrel[0]{\gamma'}{\int} v f_B(v)dv=\pi^{O(\gamma,\gamma')}_2. \label{equal-profit-condition}
\end{align}

To see that the above wage functions permit a core outcome $O(\gamma,\gamma')$, first note that there is no firm $i$ and outcome $\tilde O_i$ with group-based segregation that blocks $O(\gamma,\gamma')$. Suppose without loss of generality that in outcome $\tilde{O}_i$, firm $i$ hires a positive measure of $A-$group workers only with associated wage function $\tilde w_i(\cdot)$. By the definition of block, for almost every $v$, $\tilde f^i_A(v)>0$ only if $\tilde w_i(v)\geq w_1(v)$, implying that $\pi^{\tilde O_i}_i\leq \pi^{O(\gamma,\gamma')}_1=\pi^{O(\gamma,\gamma')}_2$. Therefore, $\tilde {O}_i$ is not a block.

Thus, consider a firm $i$ and a potential blocking outcome $\tilde O_i$ in which $i$ employs positive measures of both $A-$ and $B-$group workers. Under EPSW, such a  firm pays a common wage $w$ to almost all workers. There are two exhaustive cases:
\begin{itemize}
\item Consider $w \le \gamma$. By the definition of block, firm $i$ does not benefit from hiring $A-$group workers, as almost all $A-$group workers with $v\geq w$ are paid $w_1(v)=\min\{v,\gamma\} \ge w$ in outcome $O(\gamma,\gamma').$ Therefore, $\pi^{\tilde O_i}_i$ is upper bounded by the case in which firm $i$ hires only $B-$group workers. It has been shown in the preceding paragraph, however, that there exists no such block.

\item Consider $w > \gamma$. By the definition of block, firm $i$ does not benefit from hiring $B-$group workers, as almost all $B-$group workers with $v\geq w$ are paid $w_2(v)=v$ in outcome $O(\gamma,\gamma').$ Therefore, $\pi^{\tilde O_i}_i$ is upper bounded by the case in which firm $i$ hires only $A-$group workers. It has been shown in the preceding paragraph, however, that there exists no such block.

\end{itemize}

The preceding arguments demonstrate that any outcome $O(\gamma,\gamma')$ with $\gamma\geq \gamma'$ and $\pi^{O(\gamma,\gamma')}_1=\pi^{O(\gamma,\gamma')}_2$ is a core outcome. To see that there is a continuum of pairs $\gamma,\gamma'$ that satisfy these conditions, note that if $\gamma=\gamma'=1$, then the right-hand side of \eqref{equal-profit-condition} is strictly positive while the left-hand side is zero, and hence the former is strictly larger than the latter. Because of continuity of the left-hand side in $\gamma$, the left-hand side is strictly smaller than the right-hand side for any $\gamma$ that is sufficiently close to $1$. Now, noting continuity of the right-hand side in $\gamma'$ and the fact that it is equal to zero for $\gamma'=0$, by intermediate value theorem, it follows that there exists $\gamma'$ such that  \eqref{equal-profit-condition} holds with equality. This concludes the proof of part 1. \qed\\

\noindent \textbf{Proof of Part \ref{gap_increases_profit}:}  Recall that $AW^O_A$ and $AW^O_B$ denote the average wages for  $A-$ and $B-$group workers, respectively.
In addition, denote $TS_A:=\beta\mathbb{E}_A(v)$ and $TS_B:=\mathbb{E}_B(v)$ to be the total surpluses created by $A-$ and $B-$group workers (together with the firms hiring them), respectively.
Then, note that $\pi^O_1=TS_A - \beta AW^O_A$ and $\pi^O_2=TS_B - AW^O_B$. The Equal Profit Condition tells us
\begin{align}
 & TS_A - \beta AW^O_A  = TS_B - AW^O_B  \notag \\
\Longleftrightarrow &
\label{eq:gap_widens}
TS_A - TS_B + (1-\beta)AW^O_A= AW^O_A - AW^O_B.
\end{align}

\noindent \emph{Proof of the ``only if'' part:}
Suppose that $AW^O_A - AW^O_B \le AW^{O'}_A - AW^{O'}_B$ for core outcomes $O$ and $O'$. Then \eqref{eq:gap_widens} implies that $TS_A - TS_B+ (1-\beta)AW^O_A \le TS_A - TS_B+ (1-\beta)AW^{O'}_A$. Note that $TS_A - TS_B$ is a constant. Since $\beta>1$, this implies that $AW^O_A \ge AW^{O'}_A$. Since firm 1 hires almost all $A-$group workers in both $O$ and $O'$, this implies that $\pi^O_1 \le \pi^{O'}_1$. Then  $\pi^O_2 = \pi^O_1 \le \pi^{O'}_1 = \pi^{O'}_2$, where the equalities follow from the Equal Profit Condition. \qed\\

\noindent \emph{Proof of the ``if'' part:}
Suppose $\pi^{O}_1 \le \pi^{O'}_1$ for two core outcomes $O$ and $O'$ (which implies $\pi^O_2 = \pi^O_1 \le \pi^{O'}_1 = \pi^{O'}_2$ by the Equal Profit Condition). Since firm 1 hires almost all $A-$group workers in both $O$ and $O'$, this implies that $AW^{O}_A \ge AW^{O'}_A$. Since $\beta>1$, this
implies that $TS_A - TS_B+ (1-\beta)AW^O_A \le TS_A - TS_B+ (1-\beta)AW^{O'}_A$. Thus, by \eqref{eq:gap_widens},  $ AW^{O}_A - AW^{O}_B \le AW^{O'}_A - AW^{O'}_B.$ \qed \\

\noindent \textbf{Proof of Part \ref{gap-is-larger-with-EPL}:} Consider the core outcome in which $w_1(v)=w_2(v)=v$ for all $v \in [0,1]$ and all $A-$ and $B-$group workers are hired by firms 1 and 2, respectively. The wage gap in this core outcome is the same as in any core outcome without EPSW, and both firms' profits are zero. It is straightforward to see that in any non-equivalent core outcome firm profits are strictly positive. Thus, applying the conclusion of Part \ref{gap_increases_profit} completes the claim. \qed\\

\noindent \textbf{Proof of Part \ref{ratio_increases_profit}:}  We first show that if $\beta \mathbb{E}_A[v]>\mathbb{E}_B[v]$ then the analogue of Part \ref{gap_increases_profit} holds. 
Let $\pi^O:=\pi^O_1=\pi^O_2$ be the profit of each firm in core outcome $O$, where the equality follows from the Equal Profit Condition. Then, 
$AW^O_A=(\beta \mathbb{E}_A[v]-\pi^O)/\beta$ and $AW^O_B=\mathbb{E}_B[v]-\pi^O$ because every worker is hired in any core outcome. So, the wage ratio is
$AW^O_A/AW^O_B=(\beta \mathbb{E}_A[v]-\pi^O)/[\beta(\mathbb{E}_B[v]-\pi^O)]$ if $AW^O_B \neq 0$.  By individual rationality, it must be the case that $\pi^O\geq 0$ and by the Equal Profit Condition and the assumption that $\beta \mathbb{E}_A[v]>\mathbb{E}_B[v]$, it must be the case that $\pi^O\leq \mathbb{E}_B[v].$ Note that if $\pi^O=0$ (which is equivalent to $AW_B^O=\mathbb{E}_B[v]$) then $AW^O_A/AW^O_B=\mathbb{E}_A[v]/\mathbb{E}_B[v]$, and if $\pi^O=\mathbb{E}_B[v]$ (which is equivalent to $AW_B^O=0$) then by definition the wage ratio is $\infty$. Therefore, it remains only to show that the wage ratio is strictly increasing in $\pi^O$ over the interval $(0,\mathbb{E}_B[v])$. To see that this is true, consider two core outcomes $O$ and $O'$ such that $\pi^O>\pi^{O'}$. $AW^O_A/AW^O_B>AW^{O'}_A/AW^{O'}_B$ if and only if $(\beta \mathbb{E}_A[v]-\pi^O)/(\mathbb{E}_B[v]-\pi^O)>(\beta \mathbb{E}_A[v] - \pi^{O'})/(\mathbb{E}_B[v]- \pi^{O'})$. Cross multiplying reveals this latter condition is equivalent to $(\mathbb{E}_B[v]-\pi^{O'})(\beta \mathbb{E}_A[v]-\pi^O)>(\beta \mathbb{E}_A[v]-\pi^{O'})(\mathbb{E}_B[v]-\pi^O),$ which is in turn equivalent to $\beta \mathbb{E}_A[v]>\mathbb{E}_B[v]$, which is true by assumption, thus demonstrating the desired claim. 

 To see that the analogue of Part \ref{gap-is-larger-with-EPL} holds, consider the core outcome $O$ in which $w^O_1(v)=w^O_2(v)=v$ for all $v \in [0,1]$ and all $A-$ and $B-$group workers are hired by firms 1 and 2, respectively. Then $\pi^O=0$, and we have established above that the wage ratio in this core outcome is the same as in any core outcome without EPSW. Now, it is straightforward to see that  in any non-equivalent core outcome firm profits are strictly positive. Therefore, the fact that  the wage ratio is strictly increasing in the profit in any core outcome completes the claim. \qed

\subsection*{Proof of \Cref{existence-result}}

\textbf{Proof of Equilibrium Existence:}
Fix any search intensity $r$ satisfying the previously imposed regularity conditions. In order to show the existence of an equilibrium comporting with $r$, we consider another game which ``automates'' the search stage of our original game and is payoff equivalent for firms to our original game. Let  $\mathcal{G}$ denote the original game, where we suppress model primitives. Let $\mathcal{G}'_r$ be a game of perfect information defined as follows:
\begin{itemize}
    \item The set of players in $\mathcal{G}'_r$ is the set of firms in $\mathcal{G}$, i.e.  workers are not players in $\mathcal{G}'_r$. The same set of firms are ``unconstrained'' in both games.
    \item The set of time periods in $\mathcal{G}'_r$ is $\{-\underline t,\underline t+1,...,-1,0,1,2,...\}$.
    \item Firms have no actions available to them in $\mathcal{G}'_r$ in periods $t< 0.$ In each period $t\geq 0$, unconstrained firms have no actions available to them, and each of the other firms can select from the following options: segregate for group $A$, segregate for group $B$, or desegregate at any endogenously selected wage $w\in[0,1]$. As in $\mathcal{G}$, the decision to segregate is irreversible, i.e. once a firm has segregated for a group $g$, then that firm continues to be segregated for group $g$ in all subsequent periods. Moreover, the segregation decision is made sequentially, following the same order as in $\mathcal{G}$.
    \item Given any history of play up to and including some arbitrary time period $t$ in $\mathcal{G}'_r$, period-$t$ payoffs for each firm equal the period-$t$ payoff each firm receives in $\mathcal{G}$ given the same history of (de)segregation decisions and assuming search, bargaining, and exogenous departure rates follow according to $r$, \eqref{worker-surplus-solution}, and $d$, respectively.
    \item All players in $\mathcal{G}'_r$ discount future payoffs according to the same geometric rate $\delta$.
\end{itemize}
Let $s'$ be a subgame perfect equilibrium of $\mathcal{G}'_r$. We claim that there is a corresponding equilibrium $s$ in $\mathcal{G}$ in which the firms make the same (de)segregation decision following a common history of (de)segregation decisions in $s$ and $s'$, and, both on and off path,  workers search according to intensity $r$. Clearly such a strategy profile $s$ satisfies our regularity conditions, and so it remains only to show that there is no decision node at any $t\geq 0$ at which any player can improve her payoff by deviating from $s$. To see this, note that by construction no firm in $\mathcal{G}$ has an incentive to deviate at any (de)segregation decision node, because no firm in $\mathcal{G}'_r$ has any incentive to deviate at any decision node by the equilibrium hypothesis. Also, no worker in $\mathcal{G}$ has an incentive to deviate her search behavior, as the selected bargaining path has been shown to be optimal for each worker.  By the one-shot deviation principle, $s$ then satisfies subgame perfection in $\mathcal{G}$, as desired.

It therefore remains only to show that there exists a pure-strategy subgame perfect equilibrium $s'$ in game $\mathcal{G}'_r$, which follows from \cite{har85perfect}: $\mathcal{G}'_r$ is a game of perfect information and clearly satisfies all the regularity conditions required in \cite{har85perfect}, and therefore admits at least one pure-strategy subgame perfect equilibrium. \qed\\

\noindent\textbf{Proof of Part \ref{claim-segregation}:} To show Part \ref{claim-segregation},  suppose for contradiction that, for any $t\geq 0$, there exists a pure-strategy equilibrium in which there exists a firm $i$ that is desegregated at time $t$.  Let $\pi_i=\frac{r_i^A}{\sum_{i' \in I}r_{i'}^A}\Delta^n\beta\mathbb{E}_A[v]>0$ be the  per-period profit firm $i$ receives from hiring $A$-group workers at wages determined by \eqref{worker-surplus-solution} born in the current period when all other constrained firms are also segregated for $A.$\footnote{\label{derivation footnote} Noting that $W^v$ as denoted in \eqref{worker-surplus-solution} is equal to zero because all firms are assumed to be segregated for group $A$, the (infinitesimal) discounted expected lifetime profit that a worker of type $(v,g)$ accrues to firm $i$ is
 $\Delta^{ n}  V^v,$
and therefore the (unique, as shown in the discussion following \eqref{worker-surplus-solution}) per-period wage for such a worker is $w=\left(1-\Delta^{ n}\right )  v $
because
$\sum_{\tau=1}^{\infty} \delta^{\tau-1} (1-d)^{\tau-1} w =\left (1-\Delta^{ n}\right )\sum_{\tau=1}^\infty \delta^{\tau-1} (1-d)^{\tau-1}v= \left (1-\Delta^{ n}\right )  V^v,$ where the first equality follows by substituting in the specified value for $w$ and the second inequality follows from \eqref{V-definition}. Note that the final term is the residual surplus accruing to the worker from \eqref{worker-surplus-solution}.
 Therefore, $\pi_i=\frac{r_i^A}{\sum_{i' \in I}r_{i'}^A}\Delta^n\beta\mathbb{E}_A[v]>0.$
 }
 Also, let $\hat V$ be any finite upper bound on the present discounted value for a firm when it hires all workers at wage zero in each period, starting at the present period indefinitely into the future.
 Let $\underline T$ be a period such that firm $i$'s profit from existing workers at period $\underline T$ is $\varepsilon \in (0,\pi_i);$ note that because desegregated firms hire no new workers on path and because workers exit the market independently with probability $d>0$ in each period, such a period $\underline T$ must exist.  Let $T_i^1>\underline T$ be such that 
  \begin{align}
     \sum_{t=\underline T}^{\infty} \delta^{t-\underline T} (1-d)^{t-\underline T} \pi_i > \delta^{T^1_i-\underline T} \hat V+\sum_{t=\underline T}^{T^1_i-1} \delta^{t-\underline T} (1-d)^{t-\underline T} \varepsilon .\label{rough-inequality}
 \end{align}
 
Because $\pi_i>\varepsilon$ and $\hat V$ is finite, such $T_i^1$ exists. The left-hand side of \eqref{rough-inequality} is a lower bound of the continuation value at period $\underline T$ of segregating for $A$ at period $\underline T$, while the right-hand side is an upper bound of the continuation value at period $\underline T$ from playing any strategy in which the firm desegregates at least until period $T_i^1.$ 
Thus, \eqref{rough-inequality} implies firm $i$ has a profitable deviation, leading to a contradiction with the hypothesis that there exists an equilibrium in which $i$ remains desegregated after time $T_i^1$. Noting that there are finitely many firms, our argument further establishes the existence of a uniform bound $T^1$ such that all firms are segregated by the end of time $T^1$ in any equilibrium, as desired. \qed 
\\

 \noindent \textbf{Proof of Part \ref{wage-ratio-result} and that every newly born worker is hired in equilibrium:} 
We first show that every newly born worker is hired in equilibrium. To see this, note that \eqref{worker-surplus-solution} implies that every newly born worker immediately reaches a bargaining agreement with the first firm that she bargains with in any time $t$.
 
Next, we show the claim in Part 2 holds. Let $T^1$ be the time period identified in the proof of Part \ref{claim-segregation} such that all firms are segregated by the end of time $T^1$ in any equilibrium. For any time $t\geq T^1$,  the ``outside option'' $W^v=0$ for all newly born workers with any productivity $v$ because there are no desegregated firms.  Then \eqref{worker-surplus-solution} implies that the  surplus  generated to a worker of type $(g,v)$ who is born at any $t\geq T^1$ is given by 
\begin{align}
\left (1-\Delta^{ n_g+n_U} \right ) V, 
\end{align}
and the worker's wage at each period (see \Cref{derivation footnote}) is 
\begin{align}
\left (1-\Delta^{n_g+n_U} \right ) v.
\end{align}
Thus, the average wage of workers in group $g$ is 
\begin{align}
\left (1-\Delta^{n_g+n_U} \right ) \mathbb{E}_g[v] + \epsilon_{g,t},
\end{align}
where $\epsilon_{g,t}$ denotes the contribution 
to the average wage of the workers of group $g$ who are hired at earlier periods and have not exited the market by period $t.$ Note that, because of exogenous exit of workers, $\epsilon_{g,t}$ converges to zero as $t \to \infty$ for each group $g.$ Consider the case in which $n_A>n_B$ (the other cases of interest follow a nearly identical argument). The wage ratio converges to 
\begin{align}
\frac{1-\Delta^{n_A+n_U}}{1-\Delta^{n_B+n_U}} \cdot \frac{\mathbb{E}_A[v]}{\mathbb{E}_B[v]}>\frac{\mathbb{E}_A[v]}{\mathbb{E}_B[v]},
\end{align}
where the inequality follows because $n_A>n_B$. Noting that this inequality is strict and that the right-hand side is equal to  the wage ratio before EPSW (see Observation \ref{observation_pre_period}), this inequality implies that there is some $T\geq T^1$ such that the wage ratio is strictly more in favor of group $A$ for any $t\geq T$ compared to any period $t'<0$. \qed

\subsection*{Proof of \Cref{high d result}}

\noindent \textbf{Proof of Part \ref{near-equal-profit-claim}:}
Suppose for contradiction that $n_A>x_A+1$ (the case with $n_A<x_A-1$ follows a symmetric argument, and is therefore omitted). Then, $n_A$ and $n_B$ satisfy $n_A-1>x_A$ and $n_B+1<x_B$. Noting that the left- and right-hand sides of \eqref{search-equal-profit-condition} are decreasing in variables $x_A$ and $x_B$, respectively,  \eqref{search-equal-profit-condition} implies that 
\begin{align} \label{intermediate-equal-profit-condition}
    \frac{\Delta^{n_A+n_U-1}}{n_A+n_U-1} \cdot \beta \mathbb{E}_A[v] & <     \frac{\Delta^{x_B+n_U+1}}{n_B+n_U+1} \cdot \mathbb{E}_B[v].
\end{align}
Considering the left-hand side of the previous equation and noting $n_A>n_A-1$, it follows that 
    \begin{align}
    \frac{\Delta^{n_A+n_U}}{n_A+n_U} \cdot \beta \mathbb{E}_A[v] & <     \frac{\Delta^{x_B+n_U+1}}{n_B+n_U+1} \cdot \mathbb{E}_B[v].\label{almost-payoff-inequality}
\end{align}
 Let 
\begin{align}
\Pi:=\sum_{t=0}^\infty \left [ \delta^{t} (1-d)^t\right ] \left [ \sum_{t'=-\underline{t}}^{-1} (1-d)^{t'+\underline{t}} \beta \mathbb{E}_A[v] \right ]. \label{big pi definition}
\end{align}

    We note that $\Pi$ is the discounted sum of any firm's payoff (evaluated at $t=0$) from hiring and keeping all $A-$group workers who are born  before EPSW at wage zero. Therefore, $\Pi$ gives an upperbound on the discounted sum of any firm's payoffs (evaluated at any period post EPSW)  from hiring and keeping all $A-$group workers who are born and matched to the firm before EPSW.
Noting that $\Pi$ converges to zero as $d\to 1$, \eqref{almost-payoff-inequality} implies that there exists $d^* \in (0,1)$ such that, for any $d>d^*$ and $t_0 \in \mathbb N,$
   \begin{align}
   \Pi + \sum_{t'=t_0}^\infty \sum_{t=t'}^\infty \left [
 \delta^{t-t_0} (1-d)^{t-t'} \frac{\Delta^{n_A+n_U}}{n_A+n_U}\cdot \beta \mathbb{E}_A[v] \right ]
 <  \sum_{t'=t_0}^\infty \sum_{t=t'}^\infty \left [
 \delta^{t-t_0} (1-d)^{t-t'}\frac{\Delta^{x_B+n_U+1}}{n_B+n_U+1} \cdot \mathbb{E}_B[v] \right]. \label{long-inequality}
\end{align}
Now, let $T$ be the final time period during which any firm segregates for group $A$ on equilibrium path--note that since the equilibrium is in pure strategies, such a  period $T$ is well defined. Then, the payoff (evaluated at $T$) of any firm $i$ that segregates for $A$ at time $T$   is at most as high as the left-hand side of \eqref{long-inequality} with $t_0=T$. Meanwhile, if firm $i$ deviates and segregates for group $B$ at time $T$, then $i$'s continuation payoff (evaluated at $T$) following this deviation is at least as high as the right-hand side of \eqref{long-inequality} with $t_0=T$.\footnote{To see this, note that $n_A-1$ firms other than $i$ segregate for $A$  by the end of time $t_0$. So, in any period following $i$'s deviation, there are at most $n_B+1$ firms segregating for $B$.} Therefore, \eqref{long-inequality} implies that $i$ is strictly better off by deviating, a contradiction to the equilibrium assumption.
 \qed\\

\noindent \textbf{Proof of Part \ref{segregation-pattern-claim}:}
Suppose $n^t_A<n^t_B$ for large $t$ for contradiction. Assume, without loss of generality, $t$ is large enough that every firm has segregated by then on the equilibrium path. Consider the last period at which a firm segregates for $B$. For this firm,  the payoff from newly hired workers in each period is strictly higher if it deviates to segregating for $A$ rather than segregating for $B$,  because both the firm's surplus from each worker and the volume of workers  reaching a bargaining agreement with it are larger if it is segregated for $A$ than when it is segregated for $B$. Because of the assumption that $d \in (0,1)$ is sufficiently large, this means that this  firm has an incentive to deviate to segregating for $A$ in this period, contradicting the equilibrium hypothesis. \qed

\setcounter{equation}{0}
\renewcommand{\theequation}{B\arabic{equation}}

\setcounter{proposition}{0}
\renewcommand{\theproposition}{B\arabic{proposition}}

\setcounter{lemma}{0}
\renewcommand{\thelemma}{B\arabic{lemma}}

\setcounter{remark}{0}
\renewcommand{\theremark}{B\arabic{remark}}

\section{Microfoundations and Theoretical Discussions}\label{microfoundation app}
In this appendix, we present several microfoundations and additional discussions pertaining to our static model in \Cref{sec:static model} and to our search model in \Cref{section:search-model}.

\subsection{Non-cooperative Game Formulation}\label{section-noncooperative-games} We describe a non-cooperative game played by firms and workers which is an analogue to our model in \Cref{sec:static model}. We show that the pure-strategy  subgame perfect Nash equilibrium outcomes of this game are equivalent to the core of the static game analyzed in the main text. We then discuss how some of the mathematical machinery we produce in this appendix is applicable to the search model presented in \Cref{section:search-model}.

The set of players are composed of two firms $1,2$ and a continuum of workers. The set of workers are given by $\{A,B\} \times [0,1]^2$. A worker is identified by a tuple  $(g,v,\i) \in \{A,B\} \times [0,1]^2$, where $g$ is the group that the worker belongs to, $v$ is her productivity, and $\i$ is an index. For each $g$, we assume that there is measure $\mu_g$ that is given as a product measure of $\mu^\text{p}_g$ and $\mu^\text{w}_g.$ More specifically, let $\mu^{\text{p}}_g$ be the Lebesgue measure on Lebesgue  $\sigma$-algebra $\mathcal{B}^{\text{p}}$ on $[0,1]$, representing the measure of productivity.\footnote{Formally, we define a Borel measure $\tilde \mu^\text{p}_g$ such that
 $
 \tilde \mu^\text{p}_g([0,x])=F_g(x)$  for all $x \in [0,1]$,
which exists and is unique \citep[Proposition 25, Section 20.3]{royden}.  Let $\mu^\text{p}_g$ be the unique measure defined on the Lebesgue measurable sets and coincides with $\tilde \mu^\text{p}_g$ on Borel measurable sets: such $\mu^\text{p}_g$ exists and is unique due to the Caratheodory Extension Theorem and the Hahn Extension Theorem \citep[see][Theorems 7.3 and 7.2']{stokeylucas}.}
Let $\mu^{\text{w}}_g$ be a measure on a $\sigma$-algebra $\mathcal{B}^{\text{w}}$ on $[0,1]$, representing the measure of workers. 
We assume that   $\mu^{\text{w}}_g$ is non-atomic. The  density function associated with $\mu_g$  is given by $f_g(v,\i)=f^\text{p}_g(v) \times f^\text{w}_g(\i)$, where  $f^\text{p}_g(v)$ is associated with measure $\mu^{\text{p}}_g$ and represents the density of $g-$group workers with productivity $v$ while $f^\text{w}_g(\i)$ is associated with measure $\mu^{\text{w}}_g$ and represents the density of workers whose indices are $\i$.

We first provide a ``microfoundation'' to our approach to work with distribution functions over productivities in our game.  Let $X=[0,1]\times[0,1]$. For $Y\subseteq X$ and $v\in [0,1]$, let $Y_v=\{t\in [0,1]: (v,t)\in Y\}$.  
Let $\mu_g$ denote the product measure induced from $\mu_g^{\text{p}}$ and  $\mu_g^{\text{w}}$ with the corresponding $\sigma$-algebra $\mathcal{B}$. 
By the product measure theorem, there exists a unique product measure induced from the two measures. 
Theorem 7.14 of \citet{stokeylucas} implies that for any $Y\in \mathcal{B}$ and $v\in [0,1]$, it holds that $Y_v \in \mathcal{B}^{\text{w}}$.
We say that $h:[0,1]\rightarrow \mathbb{R}$ is $\mathcal{B}^{\text{p}}$-measurable if
$\{v\in [0,1]: h(v)\leq a\} \in \mathcal{B}^{\text{p}} \text{ for all } a\in \mathbb{R}.$

The following proposition shows how we can interpret  $\mu_g^w(\cdot)$ as the ``density'' of $g-$group workers with productivity $v$.

\begin{proposition}\label{koji-proposition}
Let $h:[0,1]\rightarrow \mathbb{R}_{\geq 0}$ be a $\mathcal{B}^{\text{p}}$-measurable function such that $h(v)\leq c$ for all $v\in [0,1]$. Then, there exists a subset $Y^*\subseteq X$ such that $Y^*\in \mathcal{B}$ and $\mu_g^{\text{w}}(Y^*_v)=h(v) \text{ for all } v\in [0,1].$
\end{proposition}
\begin{proof}
 We say that a function mapping $[0,1]$ to $\mathbb{R}$ is \emph{simple} if its range is a finite set. The following useful result is stated as Theorem 7.5 of \citet[p. 180]{stokeylucas}.

\begin{lemma}\label{intermediate-proposition-2}
 Suppose that $h:[0,1]\rightarrow \mathbb{R}_{\geq 0}$ is $\mathcal{B}^{\text{p}}$-measurable. Then, there exists a sequence of $\mathcal{B}^{\text{p}}$-measurable simple functions $(h^n)_{n=1}^\infty$ such that
\begin{align} 
&0\leq h^n(v) \leq h^{n+1}(v)\leq h(v) \text{ for all } v\in [0,1] \text{ and } n=1, 2, \dots, \label{r2-1} \\
&h^n(v) \rightarrow h(v) \text{ (as $n\rightarrow \infty$) }\text{ for all } v\in [0,1]. \label{r2-2} 
\end{align}
\end{lemma}

For $A\subseteq [0,1]$, let $\chi_A:[0,1]\rightarrow \{0,1\}$ denote the \emph{indicator function} for $A$, i.e. 
\begin{align*}
\chi_A(v)=\begin{cases} 1 & \text{ if } v\in A, \\
                                  0 & \text{ otherwise. }
              \end{cases}
\end{align*}

By \Cref{intermediate-proposition-2}, there exists a sequence of $\mathcal{B}^{\text{p}}$-measurable simple functions $(h^n)_{n=1}^\infty$ that satisfies \eqref{r2-1} and \eqref{r2-2}. 
For each $n=1, 2, \dots$, since $h^n$ is a simple function, its range consists of a finite number of reals; let $k(n)$ denote the number. Then, there exists a sequence of reals $(a^{n,r})_{r=1}^{k(n)}$ and a sequence of mutually disjoint $\mathcal{B}^{\text{p}}$-measurable sets $(A^{n,r})_{r=1}^{k(n)}$ such that
\begin{align*}
h^n=\sum_{r=1}^{k(n)}a^{n,r}\cdot \chi_{A^{n,r}}. 
\end{align*} 

We next state the following mathematical result, which is a known modification of Sierpi\'{n}ski's theorem on non-atomic measures.\footnote{See \cite{sierpinski} for the original result, and see \url{https://en.wikipedia.org/wiki/Atom_(measure_theory)} for the modified result we state.}
\begin{lemma}\label{intermediate-proposition-1}
 Suppose that $\mu_g^{\text{w}}$ is a non-atomic measure. Let $c=\mu_g^{\text{w}}([0,1])$. Then, 
there exists a function $\phi: [0, c] \rightarrow \mathcal{B}^{\text{w}}$ such that 
\begin{align}
&\mu_g^{\text{w}}(\phi(t))=t \text{ for all } t\in [0,c], \text{ and } \label{p2}\\
&\phi(t)\subseteq \phi(t') \text{ for all } t,t'\in [0,c] \text{ with } t\leq t'. \label{p3}  
\end{align}
\end{lemma}

By \Cref{intermediate-proposition-1}, there exists a function $\phi$ that satisfies \eqref{p2} and \eqref{p3}. 
For each $n=1, 2, \dots$, and $r=1, 2, \dots, k(n)$, by \eqref{r2-1} and $h(v)\leq c$ for all $v\in [0,1]$, it holds that $a^{n,r}\leq c$. We define 
$
\bar{A}^{n,r}=\phi(a^{n,r}) \text{ for all } n=1, 2, \dots, \text{ and } r=1, 2, \dots, k(n),$ and we define $
Y^*=\cup_{n=1}^\infty \cup_{r=1}^{k(n)}(A^{n,r}\times \bar{A}^{n,r}).$ 
It holds that $A^{n,r}\in \mathcal{B}^{\text{p}}$ (because $h^n$ is $\mathcal{B}^{\text{p}}$-measurable) and $\bar{A}^{n,r}=\phi(a^{n,r})\in \mathcal{B}^{\text{w}}$. Since $\mu_{h}$ is a product measure, $A^{n,r}\times \bar{A}^{n,r}\in \mathcal{B}$. 
Since $Y^*$ is obtained by taking union of measurable sets countably many times, we have $Y^*\in \mathcal{B}$. 

Fix an arbitrary $v\in [0,1]$. Recall that, for any $n=1, 2, \dots$, the subsets $(A^{n,r})_{r=1}^{k(n)}$ are mutually disjoint. Thus,  
there exists a unique sequence $(A^{n,r(n)})_{n=1}^\infty$ such that
\begin{align}
&v\in A^{n,r(n)} \text{ for all } n=1,2, \dots. \label{7}
\end{align}
By \eqref{r2-1}, we have  $a^{n,r(n)}=h^n(v)\leq h^{n+1}(v)=a^{n+1, r(n+1)}$ for all $n=1, 2, \dots$. 
Together with \eqref{p3},
\begin{align}
 \bar{A}^{n,r(n)}=\phi(a^{n,r(n)})\subseteq \phi(a^{n+1, r(n+1)})=\bar{A}^{n+1, r(n+1)} \text{ for all } n=1, 2, \dots. 
 \label{r1-3} 
\end{align}
We obtain 
\begin{align*}
\mu_g^{\text{w}}(Y^*_v)&=\mu_g^{\text{w}}(\cup_{n=1}^\infty \bar{A}^{n,r(n)}) =\lim_{n\rightarrow \infty} \mu_g^{\text{w}}(\bar{A}^{n,r(n)}) =\lim_{n\rightarrow \infty} a^{n,r(n)} =\lim_{n\rightarrow \infty} h^n(v) =h(v), 
\end{align*}
where the first equality follows from \eqref{7}, the second equality follows from \eqref{r1-3} and the monotone convergence theorem of a measure, the third equality follows from \eqref{p2}, and the fourth equality follows \eqref{7}, and the fifth equality follows from \eqref{r2-2}. 
\end{proof}

Now we proceed to describe our non-cooperative game.  In the first period, each firm $i$ simultaneously announces measurable sets of workers  $\S^A_i$ and $\S^B_i$ of groups  $A$ and $B$, respectively, to which it makes job offers, as well as a measurable function $w_i^g$ on $\S^g_i$ for each $g\in\{A,B\}$ where $w_i^g(v,\i)$ is the wage that the firm makes to worker $(g,v,\i)$. We assume that $w_i^g(v,\i)= w_i^g(v',\i')$ if $v = v'$.  Each worker observes the identity of the firm that made an offer to her (if any) and the associated wage offered to her and chooses to accept one of the offers or stay unassigned and receive a wage of zero. Then, each firm $i$ is matched to the workers from $ \S^A_i \cup  \S^B_i$ who accepted its offers and pays wages to the hired workers according to the offer it made to them.

We assume that each worker is only interested in monetary transfers, so the worker's payoff is equal to the wage paid to her if she accepts an offer from a firm and zero otherwise. Next, we describe the firm's payoff. First consider the case where there is no EPSW. Let $\tilde \S^A_i$ and $\tilde \S^B_i$ be the sets  of workers from groups $A$ and $B$, respectively, who accepted firm $i$'s offer. If $\tilde \S^A_i$ and $\tilde \S^B_i$ are both measurable, then the firm $i$ obtains a payoff of
$$
\beta \int_{\tilde \S_i^A} [v-w_i^A(v,\i)]d\mu_A +  \int_{\tilde \S_i^B} [v-w_i^B(v,\i)]d\mu_B.
$$
If at least one of $\tilde \S^A_i$ and $\tilde \S^B_i$ is nonmeasurable, then the firm's payoff is $-1.$ 

When there is a EPSW, we modify the firm's payoff such that  it receives a payoff of $-1$ if there are sets of workers with positive measure  $\tilde \S^g_i \subseteq g \times [0,1]^2$ and $\tilde \S^{g'}_i \subseteq g' \times [0,1]^2$ hired by it such that $g \neq g'$ and $w_i^g(v,\i) \neq w_i^{g'}(v',\i')$ for all $(g, v,\i) \in \tilde \S^g_i$ and $(g', v',\i') \in \tilde \S^{g'}_i$.

 Given an action profile, the associated \emph{outcome} is defined as $O:=(O_1,O_2)$, $O_i=\{(f_g^i(v),w_i^g(v))\}_{v \in [0,1],  g=A,B}$ for $i=1,2,$ where $f_g^i(v)$ is the density of workers of group $g$ with productivity $v$ who are hired by firm $i$ in equilibrium, and $w_i^g(v)$ be the wage paid to those workers if $f_g^i(v)>0$ and zero otherwise.

Our solution concept is pure-strategy  subgame perfect Nash equilibrium.

\begin{proposition}
    \label{equivalence-claim}
The set of pure-strategy subgame perfect Nash equilibrium outcomes of the noncooperative game without EPSW and with  EPSW, respectively, coincides with the core of the cooperative game without EPSW and with EPSW, respectively.
\end{proposition}
\begin{proof} 
 We first prove that for any core outcome of the cooperative game $O:=\{(f_g^i(v),w_i^g(v))\}_{i \in \{1,2\}, v \in [0,1], g=A,B}$, there exists a subgame perfect Nash equilibrium in the noncooperative game whose outcome is $O.$
To show this, note by \Cref{no_EPL_equilibrium_prop} that, for almost all $i,g$ and $v,$ we have $f_g^1(v)+f_g^2(v)=f_g(v)$ and $w_i^g(v)=v$ if $f_g^i(v)>0$ in any outcome in the core.

For each $g$, let $Y^*_g \subseteq \{g\} \times [0,1]^2$ be a measurable set of workers containing, for each $v$, $f_g^1(v)$ share of workers--note that such sets exist thanks to \Cref{koji-proposition}. Consider the following strategy profile in the non-cooperative game:

\begin{enumerate}
    \item Both firms make job offers to all workers $(g,v,t)\in\{g\} \times [0,1]^2$ where the wage offer for each worker is equal to her productivity $v$.
    \item All workers accept an offer from the firm whose wage offer is the highest and weakly positive, if any. If they receive the same wage offers from both firms, then the workers in set $Y^*_A \cup Y^*_B$ accept an offer of firm 1 while all other workers accept an offer from firm 2. 
\end{enumerate}
It is straightforward to verify that the above strategy profile results in outcome $O.$ To see that this is a subgame perfect Nash equilibrium, first observe that by definition of the workers' strategies, each worker is  clearly maximizing their payoffs against all possible strategies of other players. For the firms, there is no other strategy that results in a strictly higher payoff than the payoff of zero under the prescribed strategy profile because all workers are offered a wage that is equal to their own productivity from both firms under the prescribed strategy profile. These observations show that the specified strategy profile constitutes a subgame perfect Nash equilibrium, as desired.

We next prove that any pure-strategy subgame perfect Nash equilibrium outcome of the noncooperative game is in the core of the cooperative game. To show this, assume that the outcome $O:=\{(f_g^i(v),w_i^g(v))\}_{v \in [0,1], i =1,2, g=A,B}$ is an outcome associated with a strategy profile of  the noncooperative game and is not in the core of the cooperative game. Then we construct an outcome $\tilde O_j=\{(\tilde f_g^j(v),\tilde w_j^g(v))\}_{v \in [0,1],g=A,B}$ such that the associated strategy profile of the noncooperative game yields a profitable deviation for at least one of the firms. We take $(\tilde f_A^j(v),\tilde w_j^A(v))=( f_A^j(v), w_j^A(v))$ (or $\tilde O_{-j}=\{(\tilde f_g^{-j}(v),\tilde w_{-j}^g(v))\}_{v \in [0,1],g=A,B}$ where $(\tilde f_A^{-j}(v),\tilde w_{-j}^A(v))=( f_A^{-j}(v), w_{-j}^A(v))$), i.e. we do not change the outcome for $A-$group workers for either firm. The argument in which the deviation involves workers in group $g=A$ is analogous, where terms related to the firm's profit must be multiplied by $\beta$. To show a profitable deviation,  we consider exhaustive cases as in \Cref{noespw_lemma3}. 

First, suppose there exist a firm $j$ and a subset of productivities $V \subset [0,1]$ with positive measure such that $w_j^g(v)>v$ for all $v \in V$.   
Consider a deviation by firm $j$ in which it continues making the same wage offers to all workers with productivities $v\notin V$ and makes no offers to any workers with $v\in V.$  This will be a profitable deviation, because the firm now fails to hire workers in $V$ at wages strictly greater than their productivities, while retaining the rest of its workforce at the same wages. Therefore, the original strategy profile is not a subgame perfect Nash equilibrium.

Other cases, as enumerated in the proof of \Cref{noespw_lemma3},
can be treated in a similar manner.
\end{proof}

Finally, we note that an analogous analysis to the one offered here also establishes that the search model in \Cref{section:search-model} using density functions has a measure-theoretic foundation. Specifically, each worker is associated with a type $(g,v,\i, t)$ where $t$ is the period at which the worker enters the market. The worker distribution is given as a measure on a sigma algebra defined over the set of worker types, and the wage function is measurable with respect to the same sigma algebra.

\subsection{Capacity constraints} \label{remark: capacity constraint}
     Our search model does not explicitly place any constraints on the measure of workers each firm can hire in each period. One relevant consideration may be whether segregation and wage setting decisions are artificially driven by ``unlimited'' hiring capacities. We do not believe this is the case.
     
     First, note that the unconstrained firms mechanically hire at least a fixed fraction of the workforce, as determined by the vector $r$ of search intensities. Letting $I^U$ represent the set of unconstrained firms, no firm $i\notin I^U$ ever hires more than $\frac{r_i^g}{\sum_{i' \in I^U\cup\{i\}}r_{i'}^g}$ fraction of workers of group $g$ in any time period $t$ (i.e. even if $i$ is the only firm segregated toward group $g$). Fixing $r_i^g$ for each $i\notin I^U$ and each $g\in\{A,B\},$ increasing $r_{i'}^g$ for each $i'\in I^U$ and each $g\in\{A,B\}$ mechanically lowers the maximum hiring rate of any firm $i\notin I^U$. Of note is that increasing search intensities as described above decreases the measure of workers hired by $i$ in a given period, regardless of the segregation patterns of the remaining firms. 
     
     Second, consider a more standard capacity interpretation: hiring is constrained in equilibrium if and only if ``too many'' workers seek positions at a firm.  Specifically, let $\kappa=(\kappa_i^g)_{i \in I, g \in \{A,B\}} \in (0,\infty)^{2n}$ and consider equilibria in which each firm $i$ is to hire no more than $\kappa_i^g$ measure of $g-$group workers in each period on path. We assume that unconstrained firms jointly have excess capacity:  $\sum_{i \in U}\kappa_{i}^A\geq \beta$ and  $\sum_{i \in I^U}\kappa_{i}^B\geq 1.$ To accommodate hiring limits, the search intensity at any time $t$ (i.e. the measure of workers of each $(g,v)$ type who initially search for each firm $i\in I^g_t$) is proportionally determined by weights $r^g_i$ such that no firm is initially searched for by more than $\kappa^g_i$ $g-$group workers.\footnote{Formally, for each firm $i\in I^A_t$, the fraction of $A-$group workers that first search at firm $i$ is $\hat r^A_i(x^A):=\min\{\kappa^A_i/\beta,r^A_i x^A\}$ where $x^A$ solves $\sum_{i \in  I^A_t} \hat r^A_i(x^A)=1$ and for each firm $i\in I^B_t$, the fraction of $B-$group workers that first search at firm $i$ is $\hat r^B_i(x^B):=\min\{\kappa^B_i,r^B_ix^B\}$ where $x^B$ solves $\sum_{i \in  I^B_t} \hat r^B_i(x^B)=1$. Notice that, for each $g\in\{A,B\}$, there is a unique solution for $x^g$: the function $\sum_{i \in I^g_t}\hat r^g_i(x^g)$ equals 0 at $x^g=0$, and is continuous and strictly increasing in $x^g$ until the capacity constraints of all firms bind, i.e. the minimum in the definition of $\hat r^g_i(\cdot)$ is achieved by the first element for each group $g$ and firm $i$.} Our results are qualitatively similar given such a modification to our model.

\subsection{Alternating offer microfoundation}\label{nash_bargaining}
Our search model considers that the total surplus accruing to each of the firm and a given worker is split in an exogenous manner given the worker's  ``outside option,'' leading to a ``Nash-in-Nash'' bargaining protocol. One question may be, given the finding of \cite{binmore1986nash},  whether any equilibrium  presented above has an analogue as the equilibrium of a game with alternating-offer bargaining. In this appendix, we show that this is indeed true.

   Formally, we alter Stage 2 of the timing of the game presented in \Cref{basic setup} in the following manner. Fix an exogenous order over unemployed workers.\footnote{Formally, let $\mathcal W^t$ denote the set of newly born workers at  period $t$. Within each period $t$, there are uncountably many ``moments'' of time comprising the set $[0,1]$.  Let there be a bijection $\varphi:[0,1] \to \mathcal W^t$ such that at each moment $y \in [0,1]$ of period $t$, the worker $\varphi(y) \in \mathcal W^t$ engages in search and bargaining according to the specified protocol.} Sequentially, each of these workers bargains with firms according to the following protocol, where we denote the relevant worker's type by $(g,v) \in \{A,B\} \times [0,1]$:
\begin{enumerate}
\item The worker decides which of the firms to bargain with.

\item The worker and firm engage in an alternating-offer bilateral bargaining game, with the firm being the first proposer. If the proposer's  offer is rejected, then there is probability $1-\delta' \in (0,1)$ that the bargaining breaks down, and with the complementary probability bargaining continues with the other player becoming the proposer. Additionally, if the worker's and the firm's strategies are such that the proposer's offer is rejected in all future rounds of bargaining, then we also say that the bargaining breaks down without an agreement being reached.

\item If bargaining breaks down with a particular firm, then the worker selects another eligible firm with which to bargain (the worker is not eligible to match with a firm she has previously bargained with). If there are no eligible firms remaining, the worker permanently remains unmatched and obtains a payoff of zero.
\end{enumerate}

In post EPSW periods $t\geq 0$, we assume that any firm $i$ that is desegregated at some wage $w$ must always offer workers wage $w$ in the current period, and will accept a worker's offer in the current period if and only if it is $w$.

All other aspects of the game are as before.

In general, a wide range of bargaining outcomes can be consistent with subgame perfect Nash equilibrium in this alternating-offer game, but we view some of them as being unrealistic. 
First, our model contains a continuum of infinitesimally small workers. This implies that firms can ``throw their weight around'' in potentially unrealistic ways because the surplus obtained through bargaining with any single worker is always zero. Suppose, for instance,  a firm's strategy were to propose a wage of 0 and accept only a wage of 0 when bargaining with any worker, and a   worker's strategy were to always offer and accept a wage of 0. If each worker made a strictly positive impact on firm profit, it would be farfetched for an equilibrium to contain such a strategy: if the worker were to propose a sufficiently small positive wage,  the firm would rather accept that wage offer and obtain positive surplus from this worker at the present round rather than risking a breakdown of the bargaining with that worker. However, our  continuum model makes the firm willing to forgo the infinitesimal surplus from bargaining with any one worker, so the above behavior does not formally violate sequential rationality. Second, some equilibria may not feature stationarity. Specifically, strategies employed by players may depend on the specific bargaining outcome struck between a (``small'') worker and a single firm in the past.

To avoid such nuisances, we impose a refinement of subgame perfect Nash equilibria as follows. Let 
$x$ denote the equilibrium offer that a worker makes to the firm as the worker's own share, and $y$ denote the equilibrium offer that a firm proposes to the worker as the worker's share. 
Denote the total surplus between the firm and the worker by  $V^v$, namely as defined in \eqref{V-definition}, and let $\bar V^v$ denote the worker's continuation value if the bargaining breaks down with the present firm. In standard alternating offer bargaining between two ``non-negligible'' parties, sequential rationality requires that $x$ and $y$ satisfy
\begin{align}\label{bargaining-refinement}
    y  =\delta'x+(1-\delta') \bar V^v,   \qquad \qquad 
    V^v-x =\delta'(V^v-y),
\end{align}
and the firm and the worker accept the offers if and only if the offers are at least $x$ and $y$ specified by the above equations, respectively. 
We require the same: any equilibrium strategy profile is such that the bargaining between any firm and worker follows \eqref{bargaining-refinement} and the acceptance rules as described above both on and off path.
Solving these equations, we obtain 
\begin{align*}
x= \frac{1}{1+\delta'}V^v+\frac{\delta'}{1+\delta'}\bar V^v,\qquad \qquad 
y= \frac{\delta'}{1+\delta'}V^v+\frac{1}{1+\delta'}\bar V^v.
\end{align*}

By the same logic as in our base game, it is optimal for the worker to first bargain with all segregated or unconstrained firms, and then bargain with a desegregated firm offering surplus $W^v$ if $D^v_t\neq \emptyset$. Therefore, applying the above acceptance and offer rules along the optimal bargaining path yields worker surplus of
\begin{align}\label{worker-surplus-solution_alternating}
V_1^v=\frac{1}{(1+\delta')^{m}} W^v + \left (1-\frac{1}{(1+\delta')^{m}} \right ) V^v.
\end{align}

Comparing \eqref{worker-surplus-solution_alternating} to \eqref{worker-surplus-solution} reveals an obvious similarity. Let $\Delta\in(1/2,1)$ denote the share of the available surplus that accrues to a segregated or unconstrained firm when bargaining with a worker in our base game. Let $\bar \Delta(\delta'):=\frac{1}{1+\delta'}$ for any $\delta'\in(0,1)$. Note that $\bar \Delta(\cdot)$ is continuous and strictly decreasing, the limit of the function as $\delta'\to 0$ is 1, and the limit of the function as $\delta'\to 1$ is 1/2. These observations together imply that there is a unique value of  $\delta'$ for which $\bar \Delta(\delta')=\Delta$, that is, there is a unique value of the bargaining friction $\delta'$ in our alternating-offer game that yields the same equilibrium outcomes as does the exogenous split $\Delta\in(1/2,1)$ in our ``Nash-in-Nash'' game.  

We note that the previous paragraph establishes the equivalence of the set of equilibrium outcomes of the two games for $\Delta\in(1/2,1)$. There is nothing in the ``Nash-in-Nash'' procedure that suggests a surplus split $\Delta\in(1/2,1)$ is more viable than a split $\Delta\in(0,1/2)$. By a similar logic, any split $\Delta\in(0,1/2)$ can be obtained in our alternating-offer game by setting the first mover in any bargaining situation to be the worker instead of the firm, and appropriately selecting $\delta'$. 

\begin{remark}
    For any $\Delta\in(0,1)\setminus\{1/2\}$, there is a unique value of  $\delta'\in(0,1)$ such that the set of equilibria (with our imposed regularity conditions) of our ``Nash-in-Nash'' bargaining game with exogenous split $\Delta$ corresponds to the set of equilibria of the alternating-offer game with negotiation breakdown probability $\delta'$, and the firm as first proposer if $\Delta>1/2$ and the worker as first proposer if $\Delta<1/2.$
\end{remark}

\section{Empirical appendix}\label{emperical app}

\subsection{Designation of treatment status}\label{app_treatment_designation}

We define a treated firm as a firm that employs at least 10 long-term workers at the time of EPSW announcement. In this section, we discuss several reasons lending validity to this choice.

First, as firm size is endogenous, a potential concern is that manipulations in size at the time of EPSW announcement could affect our analysis. As discussed by \cite{mccrary2008}, a discontinuity in the share of firms with fewer than 10 long-term workers at the time of announcement suggests that firms may strategically alter their workforce quickly to avoid EPSW, which would mean our designation of treatment would not be as good as random.  Panel (A) of Figure \ref{mccrary} investigates this concern by plotting the share of firms in our sample with strictly fewer than 10 long-term workers across time, and overlays a separate best-fit polynomial for the time period before versus after policy announcement. As can be seen visually, there is a small increase in the share of firms with fewer than 10 long-term workers around the time of policy announcement. However, there is no statistically significant discontinuity; the p-value for the test of the null hypothesis that the share of firms with fewer than 10 long-term workers is equal in May and June 2009 is $0.43.$ 

Second, would an alternative time, such as policy enactment instead of policy announcement, be more appropriate to denote the ``post'' period? We believe not. As seen in panel (A) of Figure \ref{mccrary}, there is one notable time interval over which the share of firms with strictly fewer than 10 long-term workers increases, and this is centered around policy announcement. This indicates EPSW announcement likely led to anticipatory firm responses. Note that firm size responses are consistent with firms attempting to avoid the bite of the policy, suggesting policy announcement was salient to firms (but as discussed in the previous paragraph, the lack of discontinuity around the announcement date allows us to proceed with our difference-in-difference analysis). No such change is discernible around the time of policy enactment.\footnote{We believe that earlier anticipation of the policy is unlikely because EPSW was introduced to the Chilean senate only in May 2009. Moreover, in unreported results, we show that our findings are qualitatively and quantitatively similar if we change our analysis to designate policy time to be April 2009.}

Third, a concern may be that our definition of treatment using time of announcement may not affect the probability that a firm is bound by EPSW in subsequent months. Panel (B) of Figure \ref{mccrary} replicates Panel (A) but generates separate series by firm treatment status. While there is a mechanical mean reversion due to our definition of treatment (i.e. control firms all have fewer than 10 long-term workers at EPSW announcement, and treated firms all have at least 10 long-term workers at EPSW announcement), our treatment variable is positively correlated with a firm being affected by EPSW in future time periods. In all months after announcement, strictly fewer than 25\% of firms in our control group are bound by EPSW while strictly more than 50\% of firms in our treatment group are bound by EPSW. Panel (B) of Figure \ref{mccrary} also suggests that the magnitudes of our estimates are likely conservative, because some control firms are bound by EPSW and some treated firms are not bound by EPSW in every time period after announcement. An alternative empirical approach would be to instrument for these shares across the two groups. However, such an analysis would require additional assumptions on how a firm's specific history of long-term workers translates into policy bite.\footnote{For example, consider a firm $j$ that employs strictly fewer than 10 long-term workers at some time $t$, but at least 10 long-term workers for time periods $t-10,\dots, t-1$. Consider another firm $i$ that employs strictly fewer than 10 long-term workers in times $t-10,\dots,t-1$, but at least 10 long-term workers in time period $t$. It is not obvious, given potential wage rigidities, which of these two firms is more affected by EPSW in period $t$.} Our approach, while conservative, avoids such ad hoc assumptions.

\begin{figure}[ht]\caption{Firm size and treatment status over time \label{mccrary}}
\makebox[\textwidth][c]{
\vspace{3mm}
\centering

\begin{minipage}[c][1\totalheight][t]{0.7\textwidth}\center{\scriptsize{(A): Share with fewer than 10 long-term workers}}
\begin{center}
\includegraphics[width=0.9\textwidth]{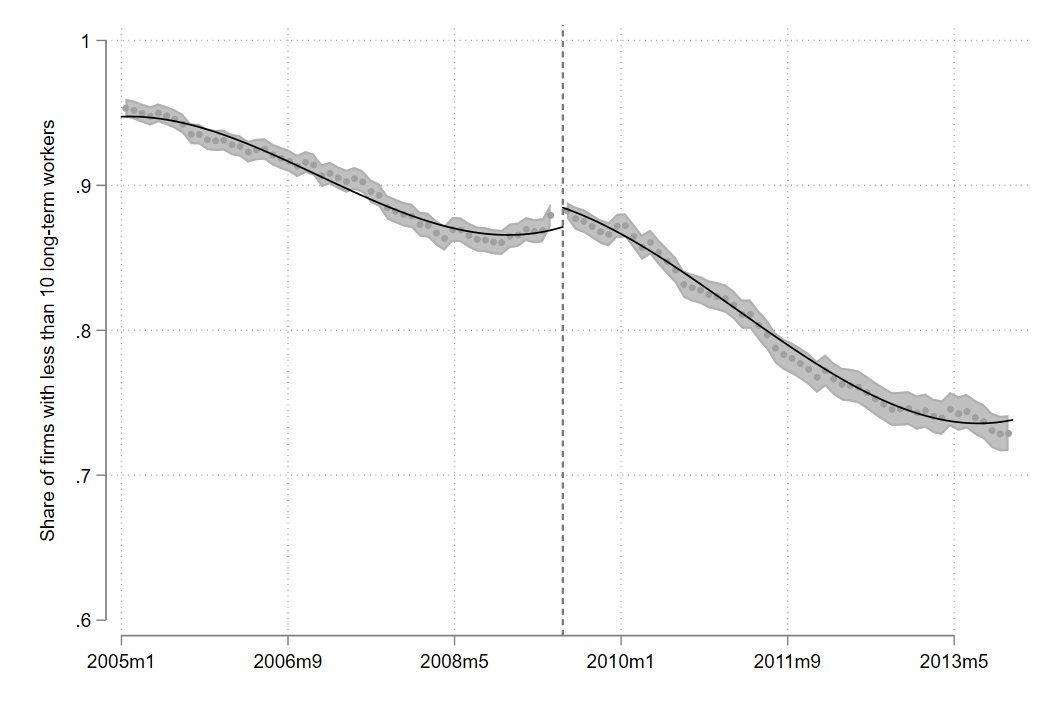}
\end{center}
\end{minipage}\hspace{-1cm}
\begin{minipage}[c][1\totalheight][t]{0.7\textwidth}\center{\scriptsize{(B): Share affected by EPSW, by treatment status}}
\begin{center}
\includegraphics[width=0.9\textwidth]{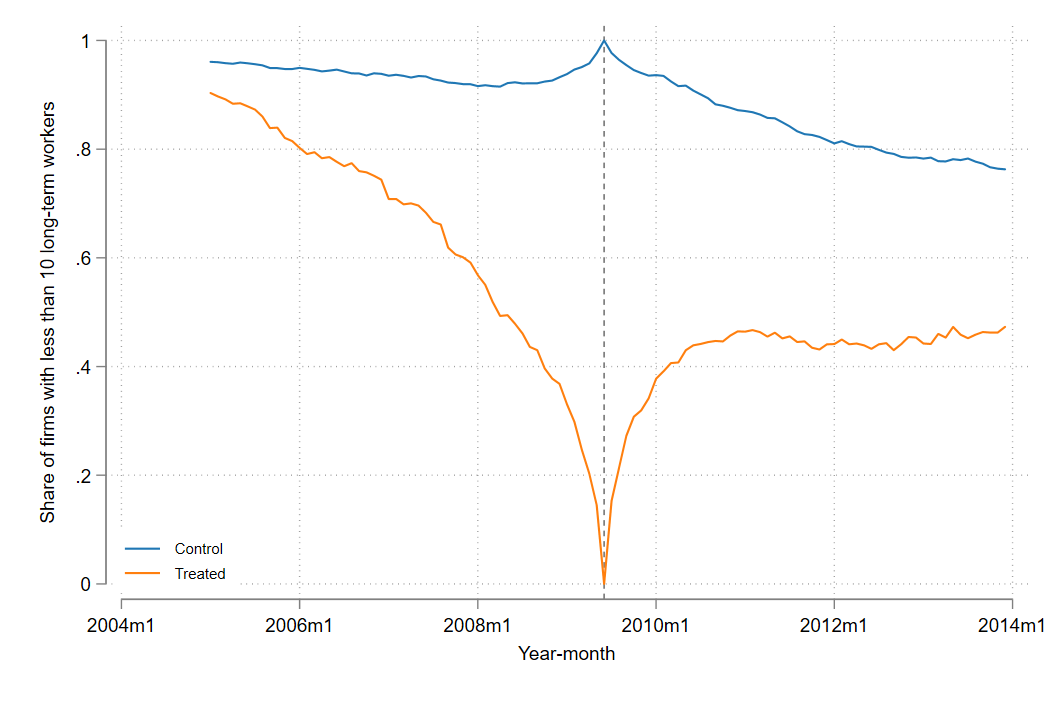}
\end{center}
\end{minipage}\hspace{1cm}}
 \begin{minipage}{.999\textwidth} \footnotesize
        \vspace{2mm}
       Notes: Both panels display size information about firms in our baseline estimation sample.
       Panel (A) presents binned means (across time, where each bin corresponds to a month-year) for the share of firms with fewer than 10 long-term workers. The black curve  is a 3rd-degree polynomial, fitted on each side of the cutoff. The shaded gray region  represents 95\% confidence bands for the local means, computed as in \cite{cattaneo}. Plots were implemented using rdplot from the rdrobust STATA package (See \url{https://rdpackages.github.io/rdrobust/}). Panel (B) displays the monthly share of firms that employ fewer than 10 long-term workers by treatment status, where treatment is defined as having at least 10 long-term workers in June 2009.
    \end{minipage}
\end{figure}

Fourth, is the comparison between control and treatment firms in our empirical specifications plausibly revealing a lower bound on the causal effect of EPSW? We believe it is. Our model in \Cref{section:search-model} finds that EPSW causes treated firms to segregate but not control firms. That is, there are no equilibrium spillovers from treated to control firms (or vice versa) on the outcome of segregation. Our model also finds that EPSW causes treated firms to make larger contributions to shifting the wage gap in favor of the majority group, as control firms retain more of their pre-EPSW workforce at existing wages. That is, the equilibrium effects serve only to attenuate the presented empirical results on wage inequality, as EPSW affects the wages set by both treated and control firms in the same direction, but it affects treated firms to a larger degree. 
In other words, our presented estimates on the effect of EPSW on the wage gap are smaller in magnitude than the difference between wage gaps between a particular market in which EPSW was never introduced and that same market under the counterfactual in which all firms were treated by EPSW.

Fifth, one may worry that there is an alternative reason for the effects that we observe. We are unaware of any contemporaneous labor-market policy that differentially affects firms in our treatment and control groups. However, it is possible that other labor-market forces could be differentially affecting firms with more long-term workers. To assuage this concern, we run several ``placebo'' tests. Specifically, we rerun our analyses in \eqref{emp_eq1} and \eqref{emp_eq5} but instead suppose the treatment cutoff for number of long-term workers at a firm varies from 10. If the true cause for the effects we observe in our baseline specification is the policy itself, we should expect little to no effects around these placebo cutoffs for firms with the same actual treatment status according to the real policy. For each cutoff $c$ we consider, we construct a sample of firms that have between $[c-4,c+3]$ workers at the time of EPSW announcement, and firms are considered treated if and only if they employ at least $c$ long-term workers at announcement. All other details are as in our baseline specification. Results are presented below in Tables \ref{tab:placebo_full_seg} and \ref{tab:placebo_het_wage_gap}. 
Under most of the alternative
cutoffs, the effects are insignificant and/or the point estimates are reversed in sign from our baseline estimates,\footnote{The sole exception is in column IV of \ref{tab:placebo_het_wage_gap}. The effect in male-majority LLMs is close to zero and statistically insignificant, while the effect in female-majority LLMs is statistically significant at the 5\% level. Given that there is no significant effect on segregation in this sample (see Table \ref{tab:placebo_full_seg}) and no significant effect on wages in male-majority LLMs, it is likely that the effect in female-majority LLMs is spurious.} suggesting that the effects we observe in our analysis are not due to other factors that differentially affect firms with different numbers of long-term workers.

\begin{table}[ht]
\centering
\addtolength{\leftskip} {-2.5cm}
    \addtolength{\rightskip}{-2.5cm}
    \caption{Effect of EPSW on Segregation--Placebos}
{
\def\sym#1{\ifmmode^{#1}\else\(^{#1}\)\fi}
\begin{tabular}{l*{5}{c}}
\hline\hline
            &\multicolumn{1}{c}{(I)}&\multicolumn{1}{c}{(II)}&\multicolumn{1}{c}{(III)}&\multicolumn{1}{c}{(IV)}&\multicolumn{1}{c}{(V)}\\
            &\multicolumn{1}{c}{Baseline - 6 to 13}&\multicolumn{1}{c}{14 to 21}&\multicolumn{1}{c}{16 to 23}&\multicolumn{1}{c}{18 to 25}&\multicolumn{1}{c}{20 to 27}\\
\hline
($\hat{\beta}^{seg}$) Post $\times$ Treated             &   $0.0441$ &    $-0.0067$    &  $0.0089$      & $-0.0167$ &  $0.0121$    \\
                                  &  $(0.0133)$         & $(0.0219)$  &   $(0.0232)$    &  $(0.0271)$     & $(0.0292)$    \\
\hline
Number of Firms                             &        $5312$         &      $1811$         &         $1496$         &         $1265$         &         $1097$          \\
Number of Observations                         &   515361	& 181298 &	150478	 &127427 &	110409              \\
\hline\hline
\textit{Fixed effects}                     &       &   &   &   &          \\
\quad Firm                              & Yes & Yes  & Yes  & Yes & Yes   \\
\quad Month$\times$exit year$\times$county              & Yes  & Yes & Yes & Yes & Yes \\
\quad Month$\times$exit year$\times$region$\times$industry              & Yes  & Yes & Yes & Yes & Yes  \\
\quad Firm-month level controls            & Yes  & Yes  & Yes  & Yes & Yes \\
\hline\hline
\end{tabular}
}
    \begin{minipage}{1.09\textwidth} \footnotesize
        \vspace{2mm}
      Notes: In this table we display estimated coefficients for the difference-in-differences regression described in \eqref{emp_eq1}. Column I presents our baseline specification for our balanced sample with firms that employed between 6 to 13 workers at June 2009, and corresponds to column I of Table \ref{emp_table_seg_robust}. Subsequent columns are identical except they consider placebo cutoffs in the number of long-term workers $c\in\{18,20,22,24\}$. For each cutoff $c$ we consider, we construct a sample of firms that have between $[c-4,c+3]$ workers at the time of EPSW announcement. For each cutoff $c$, firms are considered treated if they employ at least $c$ long-term workers at EPSW announcement, and are considered control otherwise. Throughout, standard errors in parentheses are two-way clustered at the firm and month levels. 
    \end{minipage}
    \label{tab:placebo_full_seg}
\end{table}

\begin{table}[ht]\centering
\addtolength{\leftskip} {-2.5cm}
    \addtolength{\rightskip}{-2.5cm}
    \caption{Effect of EPSW on Gender Wage Gap, by Majority Worker Group--Placebos}
{
\def\sym#1{\ifmmode^{#1}\else\(^{#1}\)\fi}
\begin{tabular}{l*{5}{c}}
\hline\hline
            &\multicolumn{1}{c}{(I)}&\multicolumn{1}{c}{(II)}&\multicolumn{1}{c}{(III)}&\multicolumn{1}{c}{(IV)}&\multicolumn{1}{c}{(V)}\\
           &\multicolumn{1}{c}{Baseline - 6 to 13}&\multicolumn{1}{c}{14 to 21}&\multicolumn{1}{c}{16 to 23}&\multicolumn{1}{c}{18 to 25}&\multicolumn{1}{c}{20 to 27}\\
\hline
($\hat \beta^{Mgap}$) Treated $\times$ Male $\times$ Post          &   $0.0427$	& $0.0209$	& $0.0039$	& $-0.0022$	& $-0.0439$     \\
                                             &   $(0.0116)$          &  $0.0175$ &	$0.0169$ &	$0.0199$	& $0.0181$            \\

($\hat \beta^{Mgap}+\hat \beta^{Fgap}$) Effect in Female                    &    $-0.0624$          &      $-0.0541$	&$0.0053$&	$-0.1054$	&$-0.0561$        \\
\quad Majority LLM                   
                        &    $(0.0197)$     &      $0.0372$ &	$0.0297$	& $0.0353$	& $0.0423$      \\

\hline
Number of Firms                                        &        $6424$        &        $2286$    & $1916$    & $1648$     & $1436$      \\
Number of Observations                                 &     $5551100$       &     $3892193$	& $3626862$	& $3433948$	& $3252092$     \\
\hline\hline
\textit{Fixed effects}                                   &       &   &   &   &  \\
\quad Firm                                          & Yes & Yes & Yes & Yes & Yes \\
\quad Worker                                             & Yes & Yes & Yes & Yes & Yes  \\
\quad Month$\times$exit year$\times$county$\times$hum. cap.         & Yes  & Yes & Yes & Yes & Yes  \\ 
\quad Month$\times$exit year$\times$region$\times$industry$\times$hum. cap.         & Yes  & Yes & Yes & Yes & Yes \\
\quad Firm-month level controls                                 & Yes  & Yes  & Yes  & Yes & Yes  \\
\quad Worker-firm-month level controls                              & Yes  & Yes  & Yes  & Yes & Yes \\
\hline\hline
\end{tabular}
}
    \begin{minipage}{1.22\textwidth} \footnotesize
        \vspace{2mm}
        Notes: In this table we display estimated coefficients for the difference-in-differences regression described in \eqref{emp_eq5}. Column I presents our baseline specification for our  sample with firms that employed between 6 to 13 workers at June 2009, and corresponds to column I of Table \ref{emp_tableYYYY}. Subsequent columns are identical except they consider placebo cutoffs in the number of long-term workers $c\in\{18,20,22,24\}$. For each cutoff $c$ we consider, we construct a sample of firms that have between $[c-4,c+3]$ workers at the time of EPSW announcement. For each cutoff $c$, firms are considered treated if they employ at least $c$ long-term workers at EPSW announcement, and are considered control otherwise. Throughout, standard errors in parentheses are two-way clustered at the firm and month levels. 
    \end{minipage}
    \label{tab:placebo_het_wage_gap}
\end{table}

\subsection{Alternative firm comparison groups}\label{alt_comparison}

In this section, we re-estimate our baseline specifications with different fixed effects. Primarily, we consider time-varying fixed effects that alter those in our segregation and wage gap analyses by altering the firm comparison groups. As discussed in \Cref{empirical section}, these time-varying fixed effects control for composition changes across different ``sectors'' of the labor market, as defined by the firm comparison group. Therefore, if our results are quantitatively similar across various specification of ``sector,'' we interpret this as evidence that composition changes across sectors of the labor market  are not driving the findings we present in the main body.

Table \ref{emp_table_seg} presents estimates from \eqref{emp_eq1}. The specifications in columns I--IV differ in the firm comparison groups $k(j)$, and column V recreates the results of our baseline specification presented initially in \Cref{emp_table_seg_robust}, column I. 
Across specifications, we find a 3.9--4.6 percentage point increase in segregation due to EPSW.

\begin{table}[ht]
\centering
\addtolength{\leftskip} {-2cm}
    \addtolength{\rightskip}{-2cm}
\caption{Effect of EPSW on Segregation--Alternative Time Trends }\label{emp_table_seg}
{
\def\sym#1{\ifmmode^{#1}\else\(^{#1}\)\fi}
\begin{tabular}{l*{5}{c}}
\hline\hline
            &\multicolumn{1}{c}{(I)}&\multicolumn{1}{c}{(II)}&\multicolumn{1}{c}{(III)}&\multicolumn{1}{c}{(IV)}&\multicolumn{1}{c}{(V)}\\
\hline
($\hat \beta^{seg}$) Post $\times$ Treated &     $0.0401$  & $0.0398$  & $0.0459$  & $0.0389$  & $0.0441$\\
            &    $(0.0118)$      & $(0.0124)$      & $(0.0124)$      & $(0.0141)$      &  $(0.0133)$                 \\
\\
Mean Pre-Treatment &  $-0.0032$           &   $-0.0065$      &   $-0.0025$        &  $0.0037$   &  $-0.0014$             \\
 &   $(0.0104)$           &        $(0.0108)$     &      $(0.0111)$         &    $(0.0129)$          &     $(0.0118)$              \\
\hline
Number of Firms       &        $6542$          & $6326$          & $5612$          & $4599$          & $5312$               \\
Number of Observations&      $602426$        & $584252$        & $536079$        & $449306$        & $515361$                \\
\hline\hline
\textit{Fixed effects}                           &       &   &   &   &             \\
\quad Firm                               & Yes & Yes  & Yes  & Yes & Yes  \\
\quad Month$\times$exit year                           & Yes & No  & No  & No & No  \\
\quad Month$\times$exit year$\times$industry               & No  & Yes & No & No & No  \\
\quad Month$\times$exit year$\times$county                 & No  & No & Yes & No & Yes  \\
\quad Month$\times$exit year$\times$county$\times$industry             & No  & No & No & Yes & No  \\
\quad Month$\times$exit year$\times$region$\times$industry               & No  & No & No & No & Yes  \\
\quad Firm-month level controls            & Yes  & Yes  & Yes  & Yes & Yes  \\
\hline\hline
\end{tabular}
}
 \begin{minipage}{1.06\textwidth} \footnotesize
        \vspace{2mm}
       Notes:  This table displays estimated coefficient  $\hat \beta^{seg}$ for the difference-in-differences regression described in \eqref{emp_eq1}. The unit of the panel is the firm-month and the dependent variable is a binary variable that indicates whether all workers at the firm in question are of a single gender in a given month. 
       Each column presents  time-varying fixed effects  corresponding to a different comparison group of firms $k(j)$ for each firm $j$. In column I, all firms are included in the comparison group. In column II firms with the same industry code are included in the comparison group. In column III firms in the same geographic county are included in the comparison group. In column IV firms in both the same county and industry code are included in the comparison group. Column  V is our baseline specification (see column I of Table \ref{emp_table_seg_robust}). 
       The mean pre-treatment effect is the mean of $\hat \beta^{seg}_\tau$ for $\tau<0$ calculated from \eqref{emp_eq2}. Throughout, standard errors in parentheses are two-way clustered at the firm and month levels. 
    \end{minipage}
\end{table}

Table \ref{emp_table3} presents estimates from \eqref{emp_eq5}. The specifications in columns I--IV differ in the firms included in comparison group $k(ij)$, and column VI recreates the results of our baseline specification presented initially in \Cref{emp_tableYYYY}, column I. Across specifications, we find a 3.5-4.7 percentage point increase in the wage gap due to EPSW in male-majority LLMs, and a 5.0-7.5 percentage point decrease in the wage gap due to EPSW in female-majority LLMs. Column V presents our baseline specification, but replaces the separate worker and firm fixed effects--$\alpha_i$ and $\alpha_j$, respectively--with a worker$\times$firm fixed effect $\alpha_{ij}.$ By the construction of our triple difference specification, a worker-firm pair contributes to the estimates of the main coefficients of interest only if the pair are matched both before and after June 2009 (and only in the time periods in which the pair is matched). In other words, workers who change jobs due to EPSW are not driving our presented estimates. It is moreover the case that workers of firms that are fully segregated either before or after the policy do not contribute to the main coefficient estimates. Our findings resemble those in our baseline, but are smaller in magnitude; there is a smaller absolute change in the wage gap in both male- and female-majority LLMs. Recalling that EPSW is predicted to increase gender segregation \emph{by ``role'' but not across ``roles,''} within a firm, these findings are consistent with further gender segregation within roles among firms that employ workers of both genders post EPSW (because firms that are desegregated by role within firm can contribute to the wage gap in the ``opposite'' direction, see \Cref{section:search-model}).

\begin{table}[ht]\centering
\addtolength{\leftskip} {-2cm}
    \addtolength{\rightskip}{-2cm}
    \caption{Effect of EPSW on Gender Wage Gap, by Majority Worker Group\\ \centering Alternative Fixed Effects }\label{emp_table3}
{
\def\sym#1{\ifmmode^{#1}\else\(^{#1}\)\fi}
\centerline{\begin{tabular}{l*{6}{c}}
\hline\hline
            &\multicolumn{1}{c}{(I)}&\multicolumn{1}{c}{(II)}&\multicolumn{1}{c}{(III)}&\multicolumn{1}{c}{(IV)}&\multicolumn{1}{c}{(V)}&\multicolumn{1}{c}{(VI)}\\
\hline
($\hat \beta^{Mgap}$) Treated $\times$ Male $\times$ Post          &     $0.0349$       & $0.0403$       & $0.0356$      & $0.0467$ & $0.0407$ & $0.0427$                    \\
                                                                                                     & $(0.0123)$     & $(0.0123)$     & $(0.0118)$       &
                                                                                            $(0.0121)$   & $(0.0119)$ 
                                                                                                     & $(0.0116)$\\
($\hat \beta^{Mgap}+\hat \beta^{Fgap}$) Effect in Female                 &    $-0.0564$      & $-0.0583$      & $-0.0499$       &  $-0.0747$   & $-0.0506$ & $-0.0624$         \\
\quad Majority LLM          
&      $(0.0204)$     & $(0.0201)$     & $(0.0201)$          &$(0.0225)$ & $(0.0184)$ & $(0.0197)$       \\
\\
Mean Pre-Treatment                       &    $-0.0182$         &   $-0.0244$          &   $-0.0130$             &  $-0.0266$              & $-0.0143$            &  $-0.0155$           \\
\quad (Male Majority LLM)
&      $(0.0176)$         &     $(0.0164)$           &   $(0.0161)$            &    $(0.0165)$             &     $(0.0165)$             &    $(0.0163)$       \\
Mean Pre-Treatment                           &  $-0.0223$         &   $-0.0165$         &  $-0.0209$         &  $0.0028$          &   $-0.0015$             & $-0.0012$         \\
\quad (Female Majority LLM)       
&    $(0.0240)$          &        $(0.0249)$       &    $(0.0255)$         & $(0.0301)$            &  $(0.0296)$          &  $(0.0269)$       \\
\hline
Number of Firms                                        &        $6436$    & $6436$    & $6433$ & $6424$   &  $6421$  & $6424$                         \\
Number of Observations                                 &     $6045315$ & $5885503$ & $5677415$ &$5181609$ & $5500022$  & $5551100$     \\
\hline\hline
\textit{Fixed effects}                       &       &   &   &   &             \\
Firm                                & Yes & Yes & Yes & Yes& No & Yes  \\
Worker                                 & Yes & Yes & Yes & Yes & No & Yes  \\
Worker$\times$firm                                 & No & No & No & No & Yes & No \\
Month$\times$exit year$\times$hum. cap.                & Yes  & No & No & No& No & No  \\ 
Month$\times$exit year$\times$industry$\times$hum. cap.              & No & Yes & No & No& No & No  \\
Month$\times$exit year$\times$county$\times$hum. cap.       & No  & No & Yes & No & Yes & Yes  \\
Month$\times$exit year$\times$county$\times$industry$\times$hum. cap.         & No  & No & No & Yes & No & No \\ 
 Month$\times$exit year$\times$region$\times$industry$\times$hum. cap.         & No  & No & No & No & Yes & Yes \\ 
Firm-month level controls              & Yes  & Yes  & Yes  & Yes & Yes & Yes \\
Worker-firm-month level controls           & Yes  & Yes  & Yes  & Yes & Yes & Yes \\
\hline\hline
\end{tabular}}
}
 \begin{minipage}{1.34\textwidth} \footnotesize
        \vspace{2mm}
       Notes:  This table displays estimated coefficients  for the regression described in \eqref{emp_eq5}. In particular, we present estimates of $\hat \beta^{Mgap}$ and $\hat \beta^{Fgap}$. The unit of the panel is the worker-firm-month and the dependent variable is the natural logarithm of monthly earnings. 
       Each column presents time-varying fixed effects corresponding to a different  comparison group of firms $k(ij)$ for each worker $i$ and each firm $j$. All columns include worker ``human capital'' comparison groups defined by equivalence across three binary dimensions at time $t$ at firm $j$: an indicator for tertiary education, an indicator for long-term versus fixed-term contract, and an indicator for being above median age in the particular industry-region in which firm $j$ operates. Columns differentially include firm comparison groups. In column I, all firms are included in the comparison group. In column II firms with the same industry code are included in the comparison group. In column III firms in the same geographic county are included in the comparison group. In column IV firms in both the same county and industry code are included
in the comparison group.  Columns V and VI have time-varying fixed effects at the same level as in our baseline specification in column I of Table \ref{emp_tableYYYY}; Column V removes firm and worker fixed effect and replaces them with a fixed effect for each worker-firm pair, while column VI is precisely our baseline specification.
The mean pre-treatment effects are the mean of $\hat \beta^{Mgap}_\tau$ and $\hat \beta^{Mgap}_\tau+\hat \beta^{Fgap}_\tau$, respectively, for $\tau<0$ calculated from \eqref{emp_eq6}. Throughout, standard errors in parentheses are two-way clustered at the firm and month levels. 
    \end{minipage}
\end{table}

\subsection{Alternative empirical specifications}\label{alt samples}

\Cref{empirical section} of the paper discusses alternative empirical specifications and samples as robustness checks. Specifically, results from each of these robustness specifications is presented in a separate column in each of Tables \ref{emp_table_seg_robust} and \ref{emp_tableYYYY}. In this section, we describe  these robustness specifications.

\paragraph{Removing firm fixed effects}

The results from this specification are presented in column II (``No firm FEs'') of the aforementioned tables. The sample of firms and workers used is the same as in our main analysis, but we remove ``$\alpha_j$'' firm fixed effects and add ``$\delta above10_j$'' to each of the specifications presented in \eqref{emp_eq1}, \eqref{emp_eq2}, \eqref{emp_eq5}, and \eqref{emp_eq6}. 
This specification removes the restriction of identifying effects only through within-firm variation, and also allows for between-firm comparisons. In our firm-month panel, corresponding to our segregation results, the presence of firm fixed effects accounts for the possibility that firms sort differentially into segregation based on unobservables. Because such sorting does not contribute to the treatment effects in our baseline specification with firm fixed effects, the comparison to the specification without firm fixed effects provides evidence on the relative importance of sorting. In the worker-firm-month panel, corresponding to our wage gap results, the presence of firm fixed effects additionally accounts for the possibility that \emph{workers} sort differentially into firms based on firm unobservables. Therefore, the comparison to the model without firm fixed effects provides evidence on the relative importance of worker sorting.

\paragraph{Removing controls}

The results from this specification are presented in column III (``No controls'') of the aforementioned tables. The sample of firms and workers used is the same as in our main analysis, but we remove the vector of firm-month level controls $X_{jt}$ from \eqref{emp_eq1} and \eqref{emp_eq2}, and we remove the vector of worker-firm-month level controls $X_{ijt}$ from \eqref{emp_eq5} and \eqref{emp_eq6}. 
Because this specification removes observable characteristics as covariariates, stability of the coefficient estimates of interest between this and our baseline specifications indicates that plausibly-related characteristics are not driving our results. Moreover, it suggests that accounting for other, similar observable characteristics that could be present in some hypothetical data set would not greatly affect our results.

\paragraph{Doughnut hole}

Recall that our baseline sample includes firms with between 6 and 13 workers at announcement. 
Column IV (``Doughnut hole'') in the aforementioned tables displays results from a sample which excludes all firms with either 9 or 10 workers at announcement from our baseline sample, resulting in a total of 5,799 firms. This sample is to account for the mechanical increase in likelihood that the excluded firms' treatment status does not match whether they are bound by EPSW at any given point in time. 
Our results reanalyze the specifications in \eqref{emp_eq1}, \eqref{emp_eq2}, \eqref{emp_eq5}, and \eqref{emp_eq6} with this alternative sample.

\paragraph{Firm growth}

The results from this specification are presented in column V (``Firm growth'') of the aforementioned tables. We additionally interact each of the time varying fixed effects (i.e. $\alpha_{k(j)t}$ in \eqref{emp_eq1} and \eqref{emp_eq2}, and $\alpha_{k(ij)t}$ in \eqref{emp_eq5} and \eqref{emp_eq6}) with a firm-level indicator that tracks the absolute growth of each firm in terms of its workforce in the 6 months prior to EPSW announcement. Specifically, for each firm, we subtract the number of employed workers in December 2008 from the number of employed workers in June 2009. There are 285 firms that are not in existence in December 2008 but are in our baseline sample. We exclude these firms, leading to a total of 6,266 firms in the current sample. We then define an indicator variable for each quartile of the growth distribution. 
Therefore, this specification identifies the effects of EPSW on labor market outcomes of interest only comparing (workers of) firms to those of other firms that are on a similar growth path, and not to those of firms that are on different growth paths and may therefore respond differently to EPSW.

\paragraph{Balanced panel of firms}

The results from this specification are presented in column VI (``balanced sample'') of the aforementioned tables. The sample of firms and workers is defined as all firms that are present in the market in all time periods between January 2005 and December 2013. That is, it includes only the firms present in our main analysis that neither enter nor exit our panel. This leads to a total of 3,418 firms. This sample therefore studies the impact of EPSW on segregation and wage gaps for firms that neither enter nor exit our panel ``in the middle.''

\FloatBarrier
\subsection{Details on data sample}\label{filtering}

We use administrative data from the Chilean unemployment insurance system from January 2005 to December 2013. In our data, an observation is a worker-firm-month. We observe two stratified (by firm size) random samples---a 1\% sample and a 3\% sample---of firms, and for each sampled firm we observe the entire monthly working history of every worker that was ever employed by the sampled firm, regardless of whether the worker remains at the sampled firm or not. Therefore, we observe some workers during time periods they are employed by other, non-sampled firms. We do not directly observe which firms are sampled. For that reason, our dataset includes ``incidental firms'' for which we do not necessarily observe the entire workforce at any given moment in time. This happens, for example, if a worker from a sampled firm switched to a non-sampled firm. 

Because we do not observe the entire workforce for these incidental firms, we do not observe the size (number of workers) for these firms. This naturally leads to a potential concern with our size-based empirical strategy. To address this potential concern, we attempt to filter out the incidental firms. Our filtering approach is built on the notion that, due to the sampling procedure, we anticipate incidental firms to have high variance in the number (and presence) of workers across time periods. 

Our filtering procedure does the following. 1) For each of the 1\% and 3\% samples, we compute for each firm the first and the last month it is observed. We drop firms that employ no workers for some months in between these two dates in any of the two samples, i.e. if the firm has ``holes'' in its employment history. 2) For firms in our sample, we compute the monthly average number of workers in each firm in our data and we drop firms that have an average of fewer than 4 in any of the two samples. 3) For each of the firms present in both the 1\% and 3\% samples, we compute the average number of workers of each firm across time, and we drop firms that have different average numbers of workers in the two samples. Descriptive statistics for the set of firms and workers left after the filtering are presented in column I of Table \ref{tab:stats_samples}.

\end{document}